\documentclass[a4paper,11pt]{article}
\usepackage{graphicx}

\usepackage{cite,doi} %
\usepackage{amsmath,amssymb,amsfonts}
\usepackage{algorithmic}
\usepackage{textcomp}
\usepackage{xcolor}
\usepackage{hyperref}
\usepackage{textcomp}
\usepackage{booktabs}
\usepackage{wrapfig,subfig}
\usepackage{todonotes}
\usepackage{authblk}
\usepackage{amsthm}

\usepackage{pdflscape}
\usepackage{lscape}
\usepackage{multirow}
\usepackage{colortbl}
\usepackage[margin=1in]{geometry}

\DeclareMathOperator{\dir}{dir}
\newcommand{\opt}{\mathrm{opt}}
\newcommand{\Ecw}{E^\mathit{cw}}
\newcommand{\Eccw}{E^\mathit{ccw}}
\DeclareMathOperator{\sect}{sec}
\DeclareMathOperator{\dist}{dist}

\newtheorem{lemma}{Lemma}

\newcommand{\billy}[1]{\multirow{-5}{8pt}{\rotatebox[origin=c]{90}{#1}}}
\newcommand{\centhead}[1]{\multicolumn{1}{c}{#1}}
\newcommand{\centheadbar}[1]{\multicolumn{1}{c|}{#1}}
\newcommand{\centheadbartwo}[1]{\multicolumn{1}{c!{\vrule width 1pt}}{#1}}

\newcommand{\bendname}{bends}
\newcommand{\sectdevname}{sector deviation}
\newcommand{\subper}{\rotatebox[origin=c]{180}{$\Lsh$} per edge}
\newcommand{\distpername}{distortion per edge} 
\newcommand{\runtimename}{time in seconds}
\newcommand{\ChangeRT}[1]{\noalign{\hrule height #1}}

\definecolor{kitgreen}{rgb}{0, 0.588, 0.509}

\begin{document}
\title{Towards Data-Driven Multilinear Metro Maps}
\author{Soeren Nickel}
\author{Martin N\"ollenburg}
\affil{Algorithms and Complexity Group, TU Wien, Vienna, Austria}
\date{}
\maketitle              %
\begin{abstract}
Traditionally, most schematic metro maps in practice as well as metro map layout algorithms adhere to an \emph{octolinear} layout style with all paths composed of horizontal, vertical, and $45^\circ$-diagonal edges.
Despite growing interest in more general \emph{multilinear} metro maps, generic algorithms to draw metro maps based on a system of $k \ge 2$ not necessarily equidistant slopes have not been investigated thoroughly.
In this paper we present and implement an adaptation of the octolinear mixed-integer linear programming approach of N\"ollenburg and Wolff (2011) that can draw metro maps schematized to any set $\mathcal{C}$ of arbitrary orientations. 
We further present a data-driven approach to determine a suitable set $\mathcal{C}$ by either detecting the best rotation of an equidistant orientation system or by clustering the input edge orientations using a $k$-means algorithm.
We demonstrate the new possibilities of our method using several real-world examples. 
\end{abstract}
\section{Introduction}

Metro maps are ubiquitous schematic network diagrams that aid public transit passengers in orientation and route planning in almost all types of urban public transit systems worldwide.
Since Henry Beck's classic schematic London Tube Map of 1933, metro maps have developed a common visual language and adopted similar design principles~\cite{wntrn-stlfdmhp-20}.
{Designing professional metro maps is still mostly a manual task today, even if cartographers and graphic designers are supported by digital drawing tools.}
Algorithms for automated layout of metro maps have received substantial interest in the graph drawing and network visualization communities as well as in cartography and geovisualization over the last 20 years~\cite{n-samlm-14,wntrn-stlfdmhp-20,w-dsms-07}. 
The vast majority of metro map layout algorithms focus on so-called \emph{octolinear} (sometimes also called \emph{octilinear}) metro maps, which are limited to Henry Beck's classical and since then widely adopted 45$^\circ$-angular grid of line orientations~\cite{g-mbum-94}.
However, not all metro maps found in practice are strictly octolinear.
There is empirical evidence from usability studies that the best set of line orientations for drawing a metro map depends on different aspects of the respective transit network, and it may not always be an octolinear one~\cite{rnlhh-ovsmpmuitovas-13,rgl-pvpidbomsrusmmitd-16}.

In this paper we present an algorithmic approach using global optimization for computing (unlabeled) metro maps in the more flexible \emph{$k$-linearity} setting, where each edge in the drawing must be parallel to one of $k\ge 2$ equidistant orientations whose pairwise angles are multiples of $360^\circ / 2k$. 
In this sense, a $k$-linear map for $k=4$ corresponds to the traditional octolinear setting.
{In fact, most octolinear maps use a horizontally aligned orientation system, i.e., a system that includes a horizontal orientation.}
It is possible though, for some transit networks and city geometries, that a rotation of the orientation system by an angular offset yields a more topographically accurate metro map layout. 
Hence we also consider such \emph{rotated} $k$-linear maps.
In addition to equiangular $k$-linear orientation systems, we further study irregular \emph{multilinear} (or \emph{$\mathcal C$-oriented}) maps~\cite{rgl-pvpidbomsrusmmitd-16}, in which the edges are parallel to any given, not necessarily equiangular set $\mathcal C$ of orientations.
There exist a number of metro map layout algorithms (see~\cite{n-samlm-14,wntn-scsnmtci-19,wntrn-stlfdmhp-20} for  comprehensive surveys) that would technically permit an adaptation to a different underlying angular grid, yet most previous papers optimize layouts in the well-known octolinear setting only and do not discuss extensions to different linearities.
A few algorithms for generic multilinear or $k$-linear layouts exist~\cite{n-lswro-99, mg-psmml-07,dgnpr-dsep-13,bmrs-assps-16}, but they are aimed at paths or polygons rather than entire metro maps.
In the field of graph drawing many algorithms for planar orthogonal network layouts with $k=2$ as well as for polyline drawings with completely unrestricted slopes are known~\cite{dg-popda-13}, but they do not generalize to $k$-linearity and multilinearity.

\paragraph{Contributions.} 
We first present two efficient approaches for deriving suitable, data-dependent linearity systems (rotated $k$-linear and irregular multilinear) by minimizing the angular distortion of the input edge slopes (Section~\ref{sec:orientation}).
We then adapt the octolinear mixed-integer linear programming (MIP) model of N\"ollenburg and Wolff~\cite{nw-dlhqm-11} by generalizing their mathematical layout constraints to $k$-linearity and multilinearity (Section~\ref{sec:ilp}).
The main benefit of this model in comparison to other approaches is that it defines sets of hard and soft constraints and guarantees that the computed layout satisfies all the hard constraints, while the soft constraints are globally optimized.
The trade-off for providing such quality guarantees from a global optimization technique is that computation time is typically higher compared to other methods~\cite{wntn-scsnmtci-19}.
By modeling fundamental metro map properties such as strict adherence to the given linearity system and topological correctness as hard constraints, we obtain layouts that satisfy these layout requirements strictly. 
The soft constraints optimize for line straightness, compactness, and topographicity~\cite{r-wytesd-14}, i.e., low topographical distortion.
Our modifications yield a flexible MIP model, whose complexity measured by the number of variables and constraints grows linearly with the number of orientations $k$. %
We finally demonstrate the effect of horizontally aligned and rotated $k$-linear and multilinear orientation systems by providing sample layouts of four metro networks and evaluating the resulting number of bends and angular distortions for typical small values of $k=3,4,5$ (Section~\ref{sec:examples}).

\section{Preliminaries}\label{sec:prelims}

We reuse the notation of N\"ollenburg and Wolff~\cite{nw-dlhqm-11}.
The input is represented as an embedded planar\footnote{For non-planar metro graphs we temporarily introduce a dummy vertex for each edge crossing, which preserves the crossing in the output layout.} metro graph $G=(V, E)$ with $n$ vertices and $m$ edges. Each vertex $v \in V$ represents a metro station with x- and y-coordinates and each edge $e=(u,v) \in E$ is a segment linking vertices $u$ and $v$ that represents a physical rail connection between them.
Let $\mathcal{L}$ be a \emph{line cover} of $G$, i.e., a set of paths in $G$ such that each edge $e \in E$ belongs to at least one path  $L \in \mathcal{L}$. An element $L \in \mathcal{L}$ is
called a \emph{line} and corresponds to a metro line in the underlying transit network.
Note that multiple lines can pass through the same edge.
Finally, $k\ge 2$ is an input parameter that defines the number of available edge orientations in the orientation system $\mathcal{C}$. 
The set $\mathcal C$ and the parameter $k$ can be part of the input or they can be derived automatically from the input geometry, see Section~\ref{sec:orientation}. 
Figure~\ref{fig:fig_otherdimcases} shows some examples of orientation systems.
Since every orientation can be used in two directions this yields $2k$ available drawing directions. 
Let $\mathcal{K}$ be this set of $2k$ directions. 
We note that every edge is assigned exclusively to an outgoing direction of its incident vertices, which implies that the maximum degree of $G$ can be at most $2k$.
Thus the maximum degree in $G$ provides a lower bound on the required number of orientations.

\begin{figure}[tbp]
	\centering
	\hfill
	\subfloat[$k=2$, irregular\label{fig:2irr}]{\includegraphics[width=0.24\textwidth]{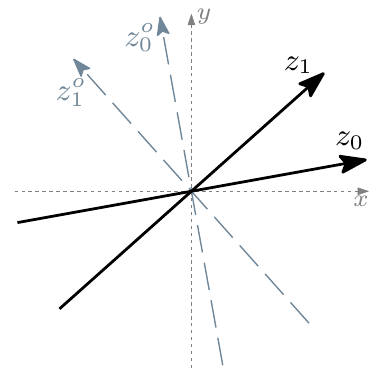}}
	\hfill
	\subfloat[$k=3$, regular aligned\label{fig:3alig}]{\includegraphics[width=0.24\textwidth]{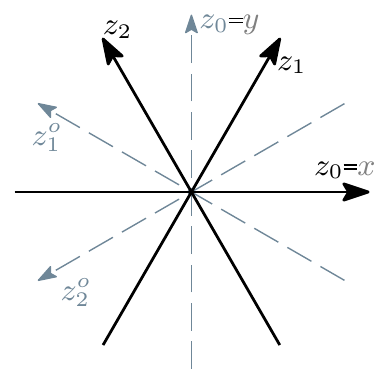}}
	\hfill
	\subfloat[$k=4$, regular rotated\label{fig:4reg}]{\includegraphics[width=0.24\textwidth]{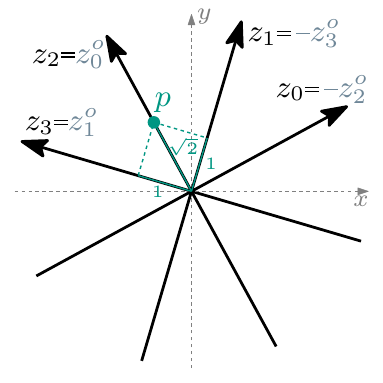}}
	\hfill
	\subfloat[$k=5$, regular rotated\label{fig:5reg}]{\includegraphics[width=0.24\textwidth]{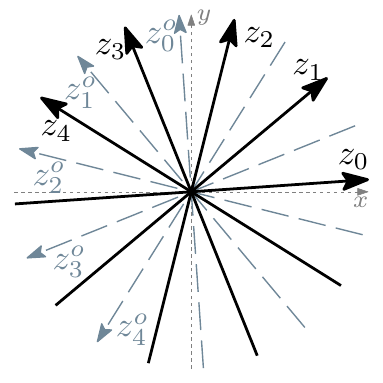}}
	\hfill\null
	\caption{Coordinate axes for different orientation systems. (c) illustrates a point $p$ with the redundant coordinates \textcolor{kitgreen}{$p = (0, 1, \sqrt2, 1)$}.}
	\label{fig:fig_otherdimcases}
\end{figure}

The general algorithmic metro map layout problem studied in this paper is to find a \emph{$\mathcal C$-oriented schematic layout} of $(G,\mathcal L)$, i.e., a graph layout that preserves the input topology, uses only edge directions parallel to an orientation from $\mathcal C$, and optimizes a weighted layout quality function (here composed of line straightness, topographicity, and compactness). 
If $\mathcal C$ corresponds to a $k$-linear orientation system, we also call the layout $k$-linear instead of $\mathcal C$-oriented; otherwise it can alternatively be called multilinear.

\section{Orientation systems}
\label{sec:orientation}
A set of edge orientations (or an \emph{orientation system}) $\mathcal{C} = \{c_1, \dots, c_k\}$ is a set of $k$ angles (expressed in radian), %
where  $0 \le c_1 < \dots < c_k < \pi$. 
We distinguish three different kinds of possible edge orientation sets. 
An edge orientation set $\mathcal{C}$ is called \emph{regular} (or \emph{equiangular}) if the angles $\{c_1, \dots c_k\}$ divide the range $[c_1, c_1 + \pi)$ into $k$ parts of equal size $\pi/k$, i.e., $c_i - c_{i-1} = {\pi}/{k}$ for all $i \in \{2, \dots, k\}$. 
Otherwise we call $\mathcal{C}$ \emph{irregular}. 
The special case of a regular orientation system $\mathcal{C}$, in which $c_1 = 0$ is called \emph{aligned}. 
Note that a classical octolinear layout is based on the aligned orientation system $\mathcal C_o = \{0, \pi/4, \pi/2, 3\pi/4\}$.

Opposed to an aligned orientation system, which is fully specified by defining the number $k$ of orientations, we have more degree of freedom in a regular (non-aligned) and an irregular system. 
The next two sections describe how to derive a suitable system $\mathcal C$ from the geometric properties of the input data. 
The idea behind this approach is to better minimize the topographic distortion of the schematized edges compared to their input direction. 

We measure the distortion $\dist_G(\mathcal{C})$ of a system $\mathcal{C}$ with respect to a metro graph $G$ by summing up the difference in slope between each edge $e \in E$ (with slope $\gamma_e \pmod{\pi}$) and the angle $c \in \mathcal{C}$ which is closest to $\gamma_e$:

\[ \dist_G(\mathcal{C}) = \sum_{e \in E} \left( \min_{c \in \mathcal{C}} |c - \gamma_e| \right). \]

\subsection{Regular orientation systems}
Fixing a single angle in a regular orientation system $ \mathcal{C} $ fixes all other orientations. 
It is therefore sufficient to specify the first orientation $c_1 \in \mathcal C$. 
We denote by $\mathcal C_\opt$ a regular orientation system minimizing the distortion, i.e., $\dist_G(\mathcal C_\opt) \le \dist_G(\mathcal C)$ for any $k$-regular orientation system~$\mathcal C$. 
The next lemma will help us to find such a system efficiently.

\begin{lemma}
	\label{lem:dir_containment}
	For any integer $k$ and metro graph $G$ there is an optimal orientation system $\mathcal{C}$ with $\dist_G(\mathcal C) = \dist_G(\mathcal C_\opt)$, in which at least one orientation $ c \in \mathcal{C} $ is equal to the slope of an input edge.
\end{lemma}

\begin{proof}
	Let $\mathcal C_\opt$ be a minimum-distortion regular orientation system. 
	For each edge $e \in E$ we define $c_\opt(e)$ to be the orientation $c \in \mathcal C_\opt$ that minimizes the distortion $|c - \gamma_e|$ of $e$.	
	Let $\varepsilon$ be a sufficiently small angle such that a rotation of $ \mathcal{C}_\opt $ by $\varepsilon$ in clockwise (resp., counterclockwise) direction results in an orientation system $ \mathcal{C}^{cw}$ (resp., $\mathcal{C}^{ccw} $) with $c^{cw}(e) = c_\opt(e) - \varepsilon$ and $c^{ccw}(e) = c_\opt(e) + \varepsilon$. 
	This implies that in these rotations, every edge $e$ moves (in slope) either strictly closer to of strictly farther away from $c_\opt(e)$. 
	Let $\Ecw_+$ and $\Ecw_-$ be the sets of edges that increase and decrease, respectively, their distance to $c_\opt(e)$ during the clockwise rotation of the orientation system by $\varepsilon$.
	Analogously, we define $\Eccw_+$ and $\Eccw_-$ for a counterclockwise rotation by $\varepsilon$.
	Note that if there is an edge that is in $\Ecw_+$ and $\Eccw_+ $ simultaneously, its slope coincides with a direction in the orientation system and we are done. 
	So assume that every edge is either in $\Ecw_-$ or in $\Eccw_-$ and therefore $|\Ecw_-| + |\Eccw_-| \geq |\Ecw_+| + |\Eccw_+|$. 
	If $|\Ecw_+| < |\Ecw_-|$, then $ \dist_G(\mathcal{C}^{cw}) < \dist_G(\mathcal{C}_\opt) $ which contradicts the minimality of $ \mathcal{C}_\opt $. 
	If $|\Ecw_+| > |\Ecw_-|$ then $ |\Eccw_+| < |\Eccw_-| $ and $ \dist_G(\mathcal{C}^{ccw}) < \dist_G(\mathcal{C}_\opt) $, which again contradicts the minimality of $ \mathcal{C}_\opt $. 
	So finally, we must have $|\Ecw_+| = |\Ecw_-|$ and then $ \dist_G(\mathcal{C}^{cw}) = \dist_G(\mathcal{C}_\opt) $. 
	We can thus continue the clockwise rotation until one of two things will happen. 
	Either the slope of an edge will coincide with a direction in the rotated orientation system, in which case we are done, or the bisector between two of the rotated orientations coincides with the slope of an edge. 
	If we continue the rotation, this edge will change from $\Ecw_+$ to $\Ecw_-$ and hence $|\Ecw_+| < |\Ecw_-|$.
	A further minimal rotation will thus result in $ \dist_G(\mathcal{C}^{cw}) < \dist_G(\mathcal{C}_\opt) $, which contradicts the minimality of $ \mathcal{C}_\opt $.
\end{proof}

By this lemma we can restrict our search to orientation systems in $\mathfrak C(E) = \{\mathcal{C} \mid \exists e\in E: \gamma_e \in \mathcal{C}\}$, i.e., to orientation systems, where at least one orientation coincides with the slope of an edge in $E$. 
The set $\mathfrak C(E)$ contains $O(|E|)$ elements and we select $\mathcal C_\opt$ as the one yielding the minimum  $\dist_G(\mathcal C)$ for all $\mathcal C \in \mathfrak C(E)$.

\subsection{Irregular orientation systems}
In an irregular orientation system $\mathcal C$ with $k$ orientations, each orientation can be selected independently. 
We can interpret the orientation system as a clustering of the set $\Gamma = \{\gamma_e \mid e \in E\}$ of all input edge slopes, where each cluster is formed around the closest orientation in $\mathcal C$.
Our goal is to find a set of $k$ orientations (clusters) that minimizes $\dist_G(\mathcal C)$.
To this end we can use an exact 1-dimensional $k$-means clustering algorithm of Nielsen and Nock~\cite{nielsen2014optimal} applied to the set $\Gamma$. 
This algorithm has running time $O(n^2k)$ using a precomputed auxiliary matrix as a look up table.

\section{MIP model}\label{sec:ilp}

Next we sketch how the constraints of the MIP model of N\"ollenburg and Wolff~\cite{nw-dlhqm-11} must be modified in order to compute more general $\mathcal C$-oriented metro maps for a given arbitrary set $\mathcal C$ of $k$ orientations. In our description we focus on the differences to the octolinear MIP model~\cite{nw-dlhqm-11}.

\subsection{Hard constraints}

The hard constraints comprise four aspects: $\mathcal C$-oriented coordinate system, assignment of edge directions, combinatorial embedding, and planarity.

\subsubsection{Coordinate system}
Every vertex $u$ of $G$ has two Cartesian coordinates in the plane $\mathbb R^2$, specified as $x(u)$ and $y(u)$. In order to address vertex coordinates for a flexible number $k$ of orientations, we define a redundant system of $k$ coordinates $z_0, \dots, z_{k-1}$, which are all real-valued variables in the MIP model and can all be obtained by rotating the x-axis counterclockwise by one of the angles in the orientation system $\mathcal{C} = \{\theta_0, \dots, \theta_{k-1}\} \subset [0,\pi)$.
We define the coordinate $z_i(u)$ using $x(u)$ and $y(u)$ as
$z_i(u) = \cos {\left(\theta_i\right)} \cdot x(u) + \sin {\left(\theta_i\right)} \cdot y(u)$.

In order to be able to express later that two vertices $u, v$ are collinear on a line with slope in $\mathcal C$, we need the orthogonal orientation  $z^o_i$ for each coordinate~$z_i$.
Note that while $z^o_i$ can coincide with other coordinates, this is guaranteed only in a regular orientation system with an even number of orientations.
In general, this is not the case and we need to define a second set of redundant coordinates, see Figures~\ref{fig:2irr}, \ref{fig:3alig} and~\ref{fig:5reg}.
Using a rotation by $\pi/2$ we obtain
$z^o_i(u) = -\sin {\left(\theta_i\right)} \cdot x(u) + \cos {\left(\theta_i\right)} \cdot y(u)$.

\subsubsection{Edge directions and minimum length}\label{sec:edgedirs}
Every edge $(u, v) \in E$ has an original direction in the input layout of $G$, defined as the direction from $u$ to $v$. Our $\mathcal C$-oriented coordinate system splits the plane into exactly $2k$ sectors numbered from $0$ to $2k-1$ for each vertex $u \in V$, see Figure~\ref{fig:circular_order}. We store the sector in which an edge $(u, v)$ lies in the input drawing as a constant $\sect_u(v)$ that we call the \emph{original sector} of $(u,v)$. %

Next we define an integer variable $\dir(u, v)$ to encode the selected direction of the edge $(u, v)$ in a $\mathcal C$-oriented solution. The range for $\dir(u, v)$ includes the original sector $\sec_u(v)$ and $s\ge 1$ admissible neighboring sectors in both directions. 
The original MIP model~\cite{nw-dlhqm-11} uses $s=1$, which results in a range of three admissible edge directions for each edge.

For each edge $(u,v)$ we define the set $\mathcal{S}(u,v)$ of admissible directions\footnote{All index calculations are modulo $2k$.} as $\mathcal{S}(u,v) = \{i \mid \sec_u(v) - s \leq i \leq \sec_u(v) + s\}$. 
For each $i\in\mathcal{S}(u,v)$ we define its corresponding direction number as $\sec^i_u(v)$ and define a binary variable $\alpha_i(u, v)$ of which only one can be true at any given time (\ref{eq:sector_binaries}). These are then used to assign the correct value of $\dir(u, v)$ (\ref{eq:sector_directions_a}). %
\begin{equation} \label{eq:sector_binaries}
\sum_{i\in\mathcal{S}(u,v)}\alpha_i(u, v) = 1
\end{equation}
\begin{equation} \label{eq:sector_directions_a}
\dir(u,v) = \sum_{i\in\mathcal{S}(u,v)} \sec^i_u(v)\cdot \alpha_i(u, v)
\end{equation}
We further define $\dir(v,u) = \dir(u,v) + k$ for the opposite edge $(v,u)$.

To guarantee that the output layout draws the edge $(u,v)$ in the selected direction $\dir(u,v)$ we need to ensure that the variables of $u$ and $v$ for the orthogonal coordinate axis $z_i^o$ are equal, i.e., $z^o_{\dir(u,v)}(u) = z^o_{\dir(u,v)}(v)$ (\ref{eq:sector_ortho_equal}) and that the  coordinates $z_{\dir(u,v)}(u)$ and $z_{\dir(u,v)}(v)$ differ by at least the minimum edge length $L_{min}$, i.e., $z_{\dir(u,v)}(v) - z_{\dir(u,v)}(u) \geq L_{min}$ (\ref{eq:sector_min_dist_direction}).
\begin{subequations} \label{eq:sectors}
	\begin{equation} \label{eq:sector_ortho_equal}
	\begin{aligned}
	z^o_{i'}(u) - z^o_{i'}(v) & \leq M(1-\alpha_i(u, v))\\
	-z^o_{i'}(u) + z^o_{i'}(v) & \leq M(1-\alpha_i(u, v))
	\end{aligned}
	\end{equation}
	\begin{equation} \label{eq:sector_min_dist_direction}
	\begin{aligned}
	z_{i'}(v) - z_{i'}(u) \geq -M(1 - \alpha_i(u, v)) + L_{\min} &\quad \textrm{if } i<k\\
	z_{i'}(u) - z_{i'}(v) \geq -M(1 - \alpha_i(u, v)) + L_{\min} &\quad \textrm{if } i\geq k\\
	\end{aligned}
	\end{equation}
\end{subequations}
Note four things. First, the constraints are created for every $ i\in\mathcal{S}(u,v)$. 
Second, we use $i' = i \bmod k$, since we only have $k$ coordinates, but $2k$ possible directions. 
Third, we need to distinguish whether the direction $i$ is smaller than the number $k$ of orientations, in which case $u$ must have a smaller value than $v$ in coordinate $z_i$, or otherwise if $i\ge k$ then $v$ must have the smaller coordinate and we need to invert the difference in (\ref{eq:sector_min_dist_direction}). And fourth, every triple of constraints for which $\alpha_i(u,v)=0$ is trivially satisfied by using a sufficiently big constant $M$ in the constraints. Due to (\ref{eq:sector_binaries}), $\alpha_i(u,v) = 1$ for exactly one index $i$ and only for that index $i$ the constraints have the desired effect on the coordinates.

\subsubsection{Combinatorial embedding}

\begin{figure}
	\centering
	\includegraphics{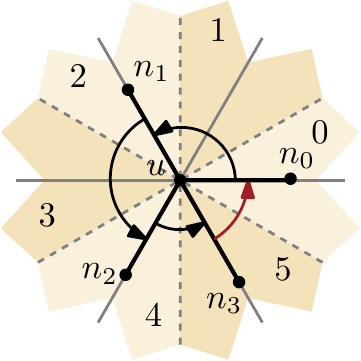}
	\caption{For edges $(u,n_3), (u,n_0)$ the direction value decreases from $5$ to $0$.}
	\label{fig:circular_order}
\end{figure}
We want to keep the combinatorial embedding, i.e., the topology of the input layout, which translates into preserving the cyclic order of the neighbors of each vertex. This can be expressed by requiring that the edge direction values strictly increase when visiting the incident edges in counterclockwise input order. There is exactly one exception, namely when going from the last used sector to the first one. 
Figure~\ref{fig:circular_order} illustrates this situation, where the crossover point lies between the neighbors $n_3$ and $n_0$, marked in red. Here we can add an offset of $2k$ instead to make the condition hold. Since this must happen exactly once, we can use binary variables $\beta_1(v), \beta_2(v), \dots, \beta_{\deg(v)}(v)$ to select the respective edge pair in equations (\ref{eq:overflow_embedding}) and (\ref{eq:overflow_embedding_b}).
\begin{equation} \label{eq:overflow_embedding}
\sum_{i=1}^{\deg(v)}\beta_i(v) = 1
\end{equation}
\begin{equation}\label{eq:overflow_embedding_b}
\begin{aligned}
\textrm{dir}(v, u_1) + 1 & \leq & \textrm{dir}(v, u_2) + 2k\cdot \beta_1(v)\\
\textrm{dir}(v, u_2) + 1& \leq & \textrm{dir}(v, u_3) + 2k\cdot \beta_2(v) \\
& \vdots & \\
\textrm{dir}(v, u_{\deg(v)}) + 1 & \leq & \textrm{dir}(v, u_1)  + 2k\cdot \beta_{\deg(v)}(v)
\end{aligned}
\end{equation}

\subsubsection{Planarity}\label{sec:planarity}
For every pair of non-adjacent edges $e=(u, v)$ and $e' = (u', v')$ we need to find (at least) one separation line between $e$ and $e'$ in a direction of $\mathcal K$ to guarantee that $e, e'$ do not intersect. 
We define  a set of $2k$ binary variables $\gamma_i(e, e')$ for which we require that at least one of them is set to true. %
\begin{equation}\label{eq:plan_binaries}
\sum_{i=0}^{2k-1} \gamma_i(e, e') \geq 1
\end{equation}
Now we ensure that every pair of edges $e, e'$ has a minimum distance $d_{\min}$ in the selected directions, i.e., %
both endpoints of $e$ have a distance of at least $d_{\min}$ to both endpoints of $e'$.
\begin{equation}\label{eq:plan_min_dist}
\begin{aligned}
z_{i'}(u') - z_{i'}(u)  & \geq & -M (1-\gamma_i(e, e')) +d_{\min} \\
z_{i'}(u') - z_{i'}(v)  & \geq & -M (1-\gamma_i(e, e')) +d_{\min} \\
z_{i'}(v') - z_{i'}(u)  & \geq & -M (1-\gamma_i(e, e')) +d_{\min} \\
z_{i'}(v') - z_{i'}(v)  & \geq & -M (1-\gamma_i(e, e')) +d_{\min} \\
\end{aligned}
\end{equation}
Note that the constraints are created for every $ 0\leq i < 2k$, that we use $i' = i \mod k$ and that the first $k$ sets of these equation look like (\ref{eq:plan_min_dist}), while the rest needs to invert the differences, e.g., $-z_{i'}(u') + z_{i'}(u)$, since they change sides with respect to direction $z_{i'}$. %

\subsection{Soft constraints}\label{sub:soft}

Soft constraints model the aesthetic quality criteria to be optimized in the layout. We adapt the three criteria of the octolinear MIP~\cite{nw-dlhqm-11} to arbitrary orientation systems: line straightness, topographicity, and compactness. Each requires a set of linear constraints and a corresponding linear term in the objective function.

\subsubsection{Line straightness}
We optimize for line straightness by minimizing the number and angles of bends along the metro lines in $\mathcal L$.
First we create a variable $\theta(u_1, u_2, u_3)$ for all pairs of consecutive edges $e_1 = (u_1, u_2), e_2 = (u_2, u_3)$ along some path $L \in \mathcal L$ that represents the cost of a potential bend between $e_1$ and $e_2$ on the metro line $L$. 
To assign $\theta(u_1, u_2, u_3)$ we subtract the direction of $e_2$  from the direction of $e_1$. If the edges do not have the same direction, the difference $\dir(u_1, u_2) - \dir(u_2, u_3)$, which we will call $\Delta \dir_{u_1, u_2, u_3}$, will either be positive or negative and $\Delta \dir_{u_1, u_2, u_3} \in [-2k+1, 2k-1]$. 
From~\cite{nw-dlhqm-11} we know that $\theta(u_1, u_2, u_3) = \min\{|\Delta \dir_{u_1, u_2, u_3}|, 2k - |\Delta \dir_{u_1, u_2, u_3}|\}$, i.e., $\theta \in [-k+1, k-1]$. 
Using two binary correction variables $\delta_1$ and $\delta_2$ we can ensure that $\theta$ takes the desired minimal value (\ref{eq:bend_deviation}), which then lets us define the bend cost function (\ref{eq:obj_dev}).
\begin{equation}
\begin{aligned} \label{eq:bend_deviation}
-\theta(u_1, u_2, u_3) & \leq \Delta \dir_{u_1, u_2, u_3} - 2k \cdot \delta_1 + 2k \cdot \delta_2\\
\theta(u_1, u_2, u_3) & \geq \Delta \dir_{u_1, u_2, u_3} - 2k \cdot \delta_1 + 2k \cdot \delta_2\\
\end{aligned}
\end{equation}
\begin{equation} \label{eq:obj_dev}
\textrm{cost}_\textrm{bends} = \sum_{L \in \mathcal{L}} \> \sum_{(u_1, u_2), (u_2, u_3) \in L} \theta(u_1, u_2, u_3)
\end{equation} 

\subsubsection{Topographicity}
\label{sec:topographicity}
In order to support the mental map~\cite{mels-lam-95} of the user, we want the shape of the output drawing to resemble the input drawing as closely as possible. 
For this we try to preserve the input directions of the edges. 
Formally we want to minimize the difference between the input direction and the output direction, i.e., $\sum_{(u,v) \in E} |\dir(u, v) - \sec_u(v)|$. 
In order to minimize the absolute value we define a new variable $\xi(u, v) = |\dir(u, v) - \sec_u(v)|$ by imposing (\ref{eq:rel_position_abs}) and minimizing $\xi(u,v)$ in the cost function (\ref{eq:obj_rel}). The topographicity cost function is simply the sum over all $\xi$-variables (\ref{eq:obj_rel}).
\begin{equation}\label{eq:rel_position_abs}
\begin{aligned}
\dir(u, v) - \sec_u(v) &\leq \xi(u, v)\\
-\dir(u, v) + \sec_u(v) &\leq \xi(u, v)
\end{aligned}
\end{equation}
\begin{equation}\label{eq:obj_rel}
\textrm{cost}_\textrm{topo} = \sum_{(u,v) \in E} \xi(u, v)
\end{equation}

\subsubsection{Compactness.}
To ensure a compact layout we minimize the total edge length of the output drawing.
Here we use that the Euclidean length of an edge $e=(u,v)$ in a $\mathcal C$-oriented layout is defined by the maximum absolute value $|z_i(u) - z_i(v)|$ in all $k$ coordinates (the projections in all other directions are shorter), which we model by a variable $\lambda(u,v)$. The compactness cost function is the sum of all edge lengths.

\begin{equation} \label{eq:edge_length_abs}
\begin{aligned}
z_i(u) - z_i(v) &\leq \lambda(u, v)\\
-z_i(u) + z_i(v) &\leq \lambda(u, v)
\end{aligned}
\end{equation}
\begin{equation} \label{eq:obj_dis}
\textrm{cost}_\textrm{length} = \sum_{(u, v) \in E} \lambda(u, v)
\end{equation}

\subsubsection{Objective function}
The objective function to be minimized is composed of the three different terms $\textrm{cost}_\textrm{bends}, \textrm{cost}_\textrm{topo}$ and $\textrm{cost}_\textrm{length}$ defined above. Each term can be weighted with factors $f_1, f_2, f_3$ depending on their relative importance.

\begin{equation} \label{eq:obj_function}
\text{minimize } f_1\cdot\textrm{cost}_\textrm{bends} + f_2\cdot\textrm{cost}_\textrm{topo} + f_3\cdot\textrm{cost}_\textrm{length}
\end{equation}

\subsection{Improvements}\label{sec:improvements}

Further, N\"ollenburg and Wolff~\cite{nw-dlhqm-11} devised several practical improvements to accelerate their method.
For instance, the number of planarity constraints (Sect.~\ref{sec:planarity}) grows quadratically with the number of edges, but most of them are never critical as any reasonable layout satisfies them trivially. 
So they suggested ways of reducing the number of necessary constraints and to add them only on demand, which immediately carries over to our generalized implementation.

\section{Experiments}\label{sec:examples}

We performed experiments on real-world data to compare the computational performance and  visual quality of metro maps with different linearity systems.

\subsection{Setup}
We generated schematic layouts of the metro networks of Montreal ($n=65$, $m=66$), Vienna ($n=90$, $m=96$), Washington DC ($n=97$, $m=101$) and Sydney ($n=173$, $m=181$) using aligned, regular and irregular orientation systems with $k \in \{3, 4, 5\}$ orientations. {All layouts were created with two different weight vectors  for the objective function, a more balanced setting $(f_1, f_2, f_3) = (3, 2, 1)$ and one emphasizing bend minimization $(f_1, f_2, f_3) = (10, 5, 1)$}, resulting in 72 instances in total. For all layouts we added planarity constraints on demand, used $s=1$ admissible neighboring sectors for each original edge direction (recall Section~\ref{sec:edgedirs}) and concentrated on the overall layout geometry and interchanges without showing the individual stops along the lines.

The experiments were run on a computing cluster with 16 nodes, each with 160GB RAM and two 10-core Intel Xeon E5-2640 v4, 2.40GHz. %
The operating system runs a Linux kernel version 4.15.0-72. 
Our implementation uses IBM ILOG CPLEX 12.8 to solve the integer linear programs as single threads. 
Each experiment had an allocated available memory space of 32GB (Montreal), 64GB (Vienna, Washington), and 150GB (Sydney). 
The increase in memory was necessary since some lager instances like Sydney ran out of memory (consistently for 64GB and sporadically even for 128GB).

To judge the quality and performance of a layout, we use several measurements. %
Line straightness was measured by the bend cost in the MIP (see Section~\ref{sub:soft}).
The sector deviation is a coarse measure of topographicity, counting how many edges are not drawn in their preferred direction (see~Section~\ref{sec:topographicity}). 
Sector deviation is measured in total and on average per edge. 
Another measure of topographicity is the angular distortion, i.e., the actual angular difference between input edges and schematized output edges, which is measured on average per edge. 
Finally we measure the runtime in seconds.

\subsection{Results}

\setlength{\tabcolsep}{2pt}
\begin{figure}[!t]
	\begin{tabular}{ccccc}
		\fbox{\includegraphics[width=.19\textwidth]{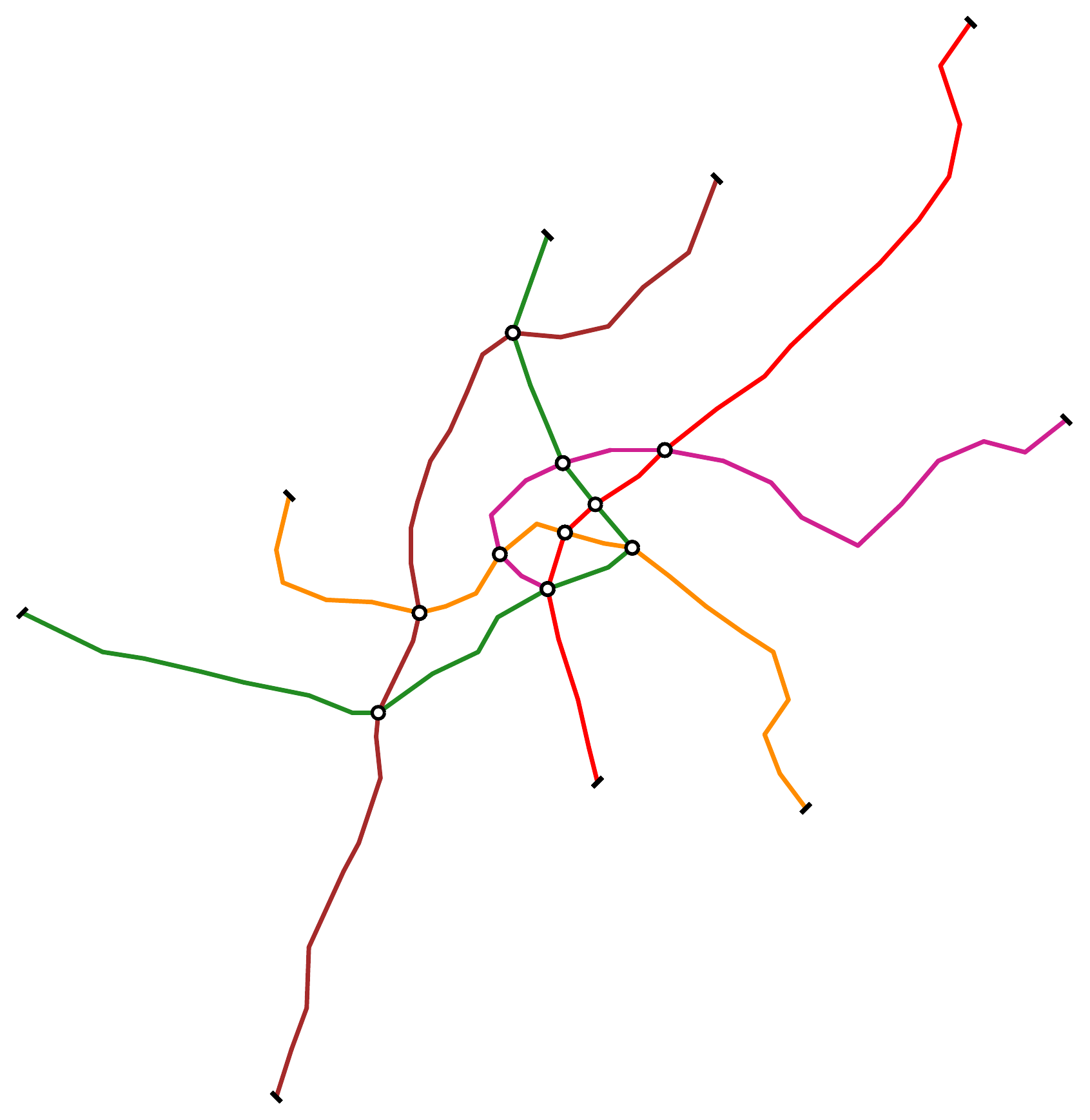}}&\includegraphics[width=.19\textwidth]{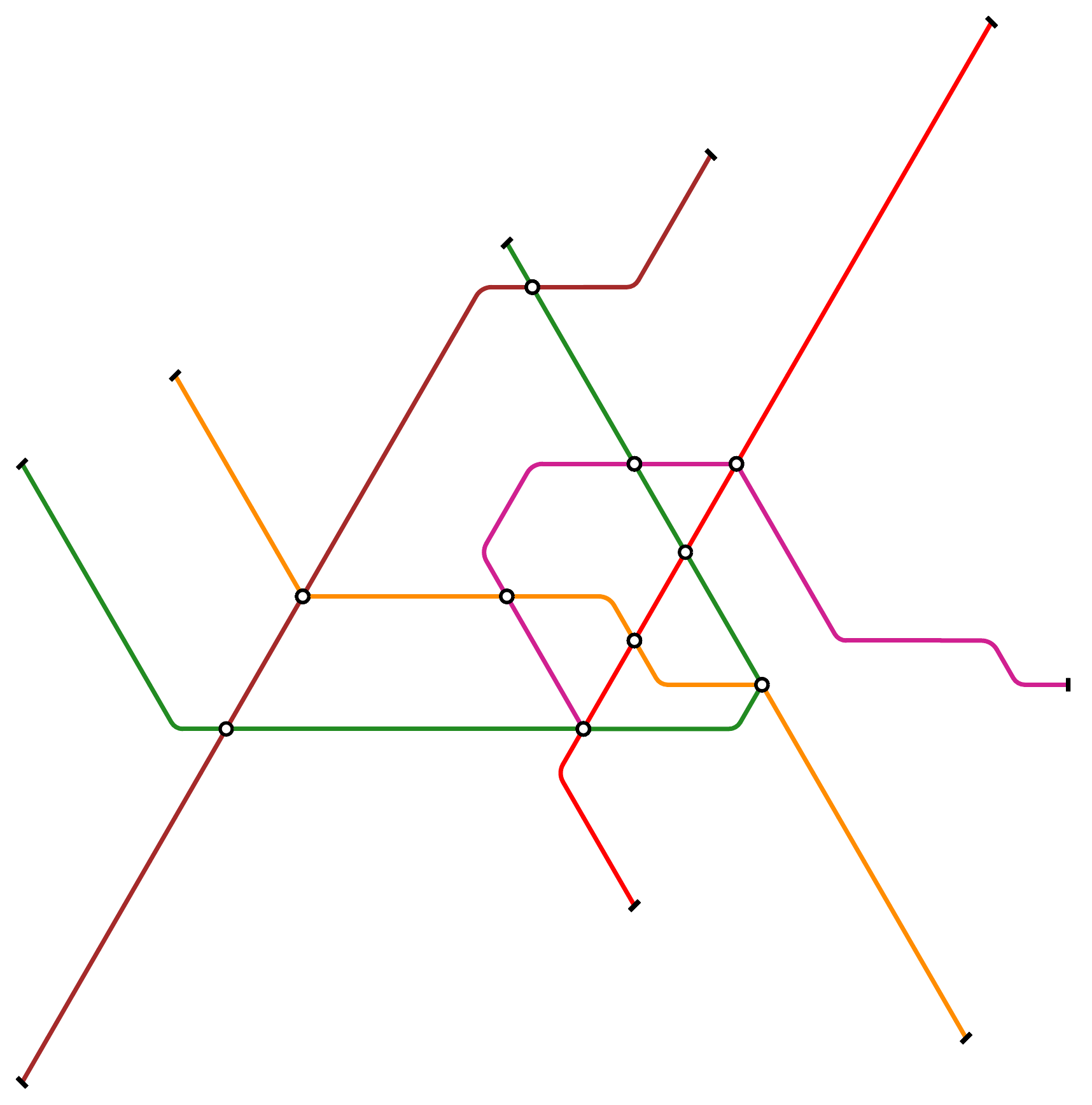} &   \includegraphics[width=.19\textwidth]{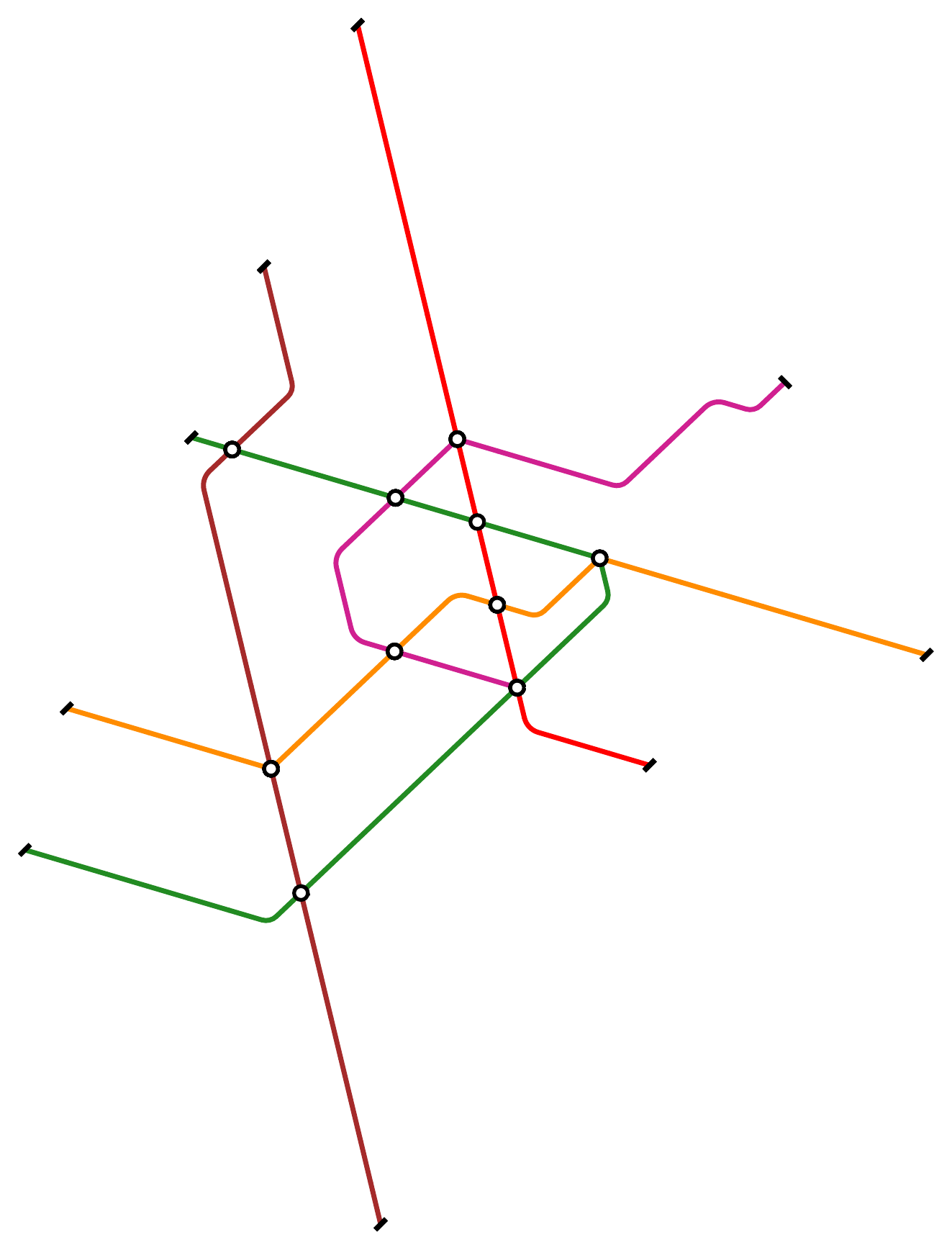} & \includegraphics[width=.19\textwidth]{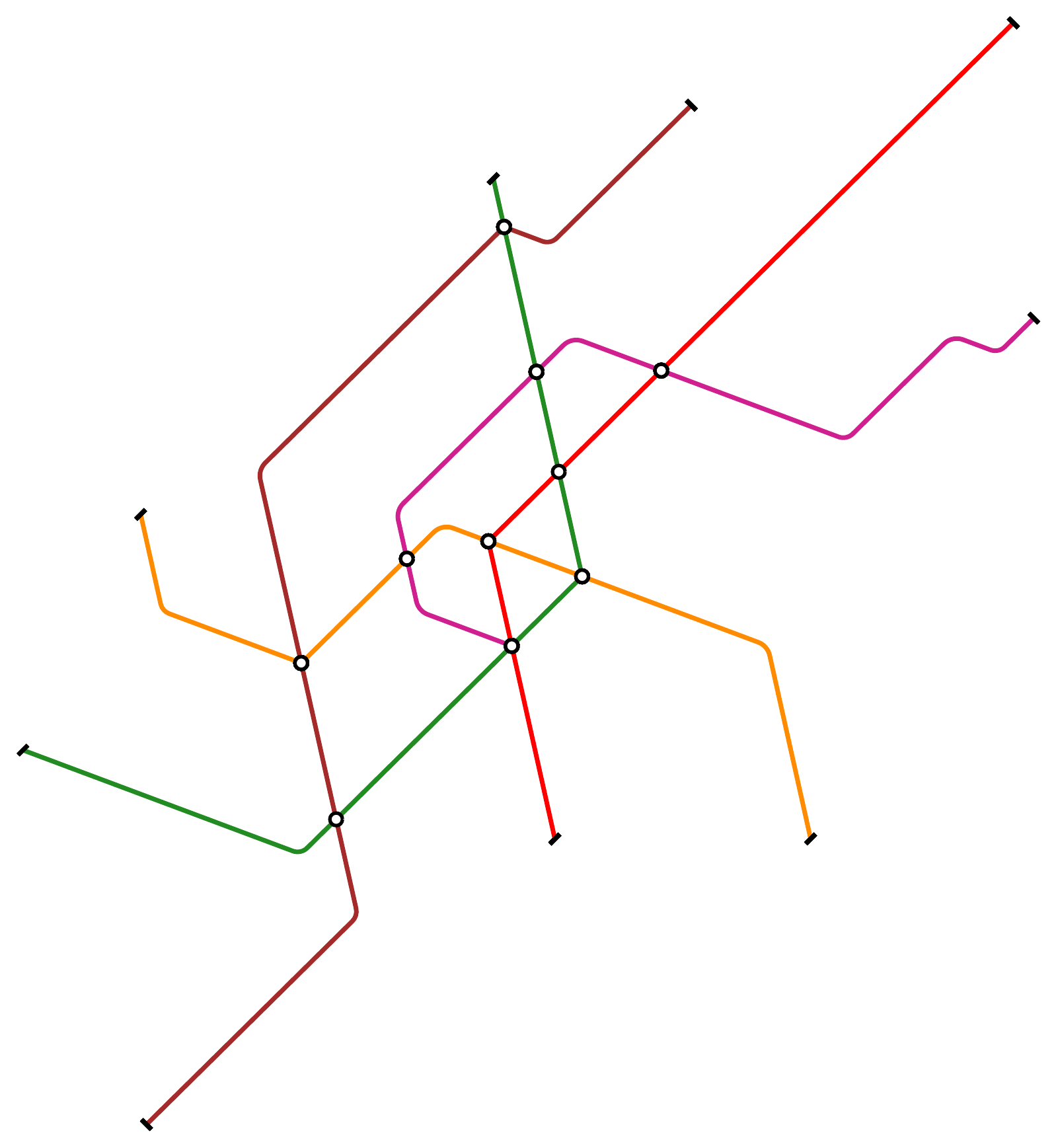} & \includegraphics[width=.19\textwidth]{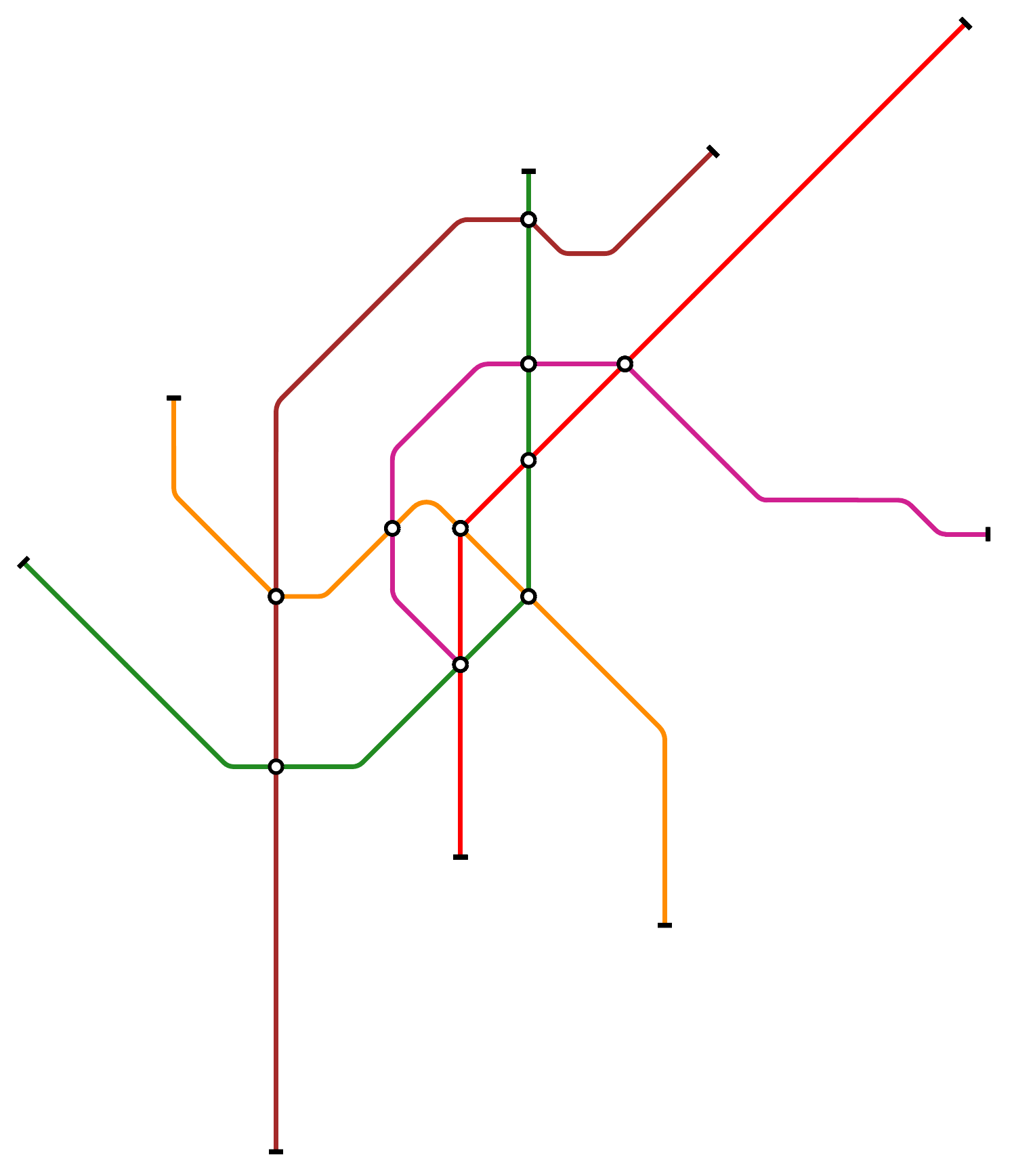}\\
		Input&(a) 3-A&(b) 3-R&(c) 3-I&(d) 4-A\\[6pt]
		\includegraphics[width=.19\textwidth]{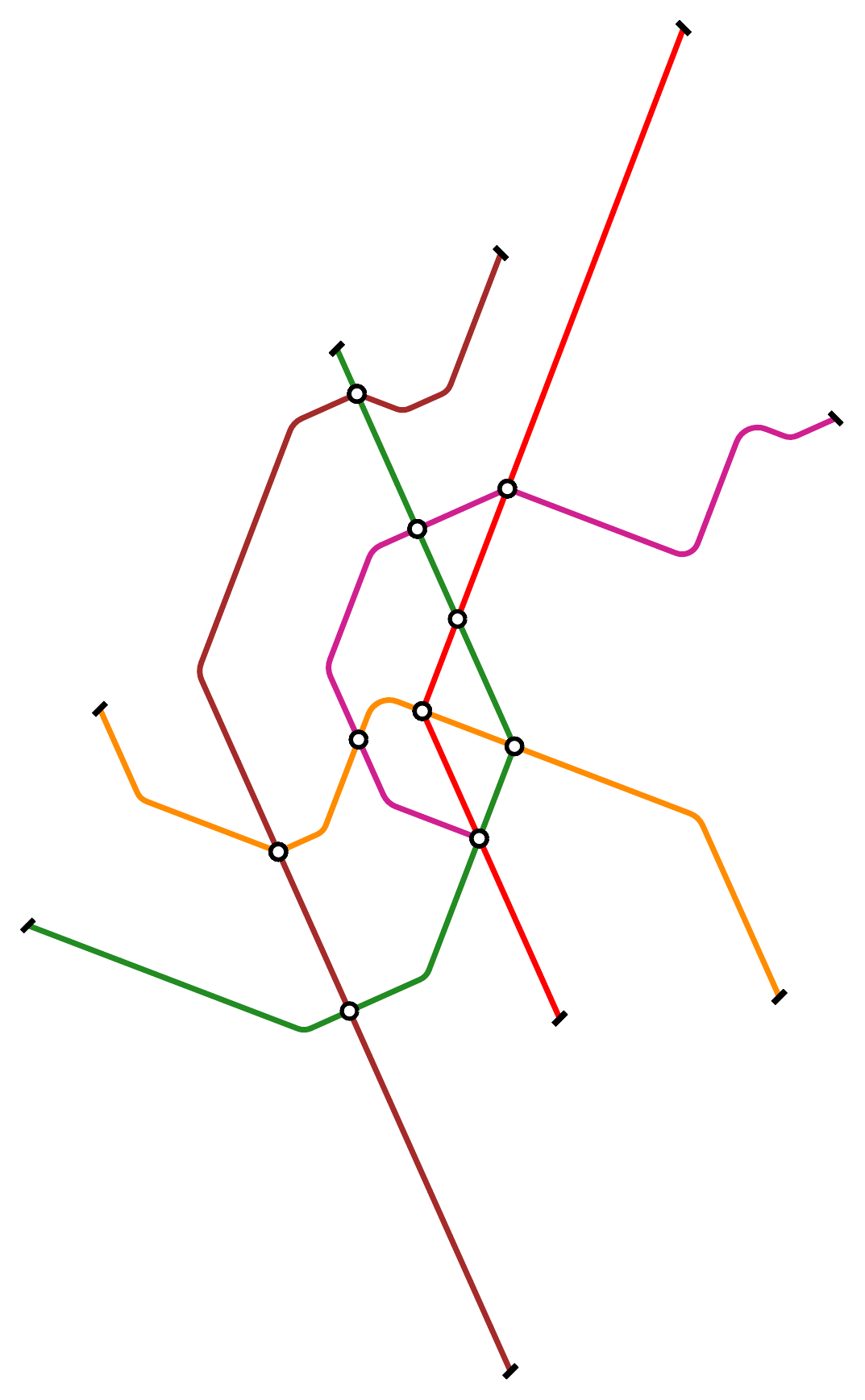} & \includegraphics[width=.19\textwidth]{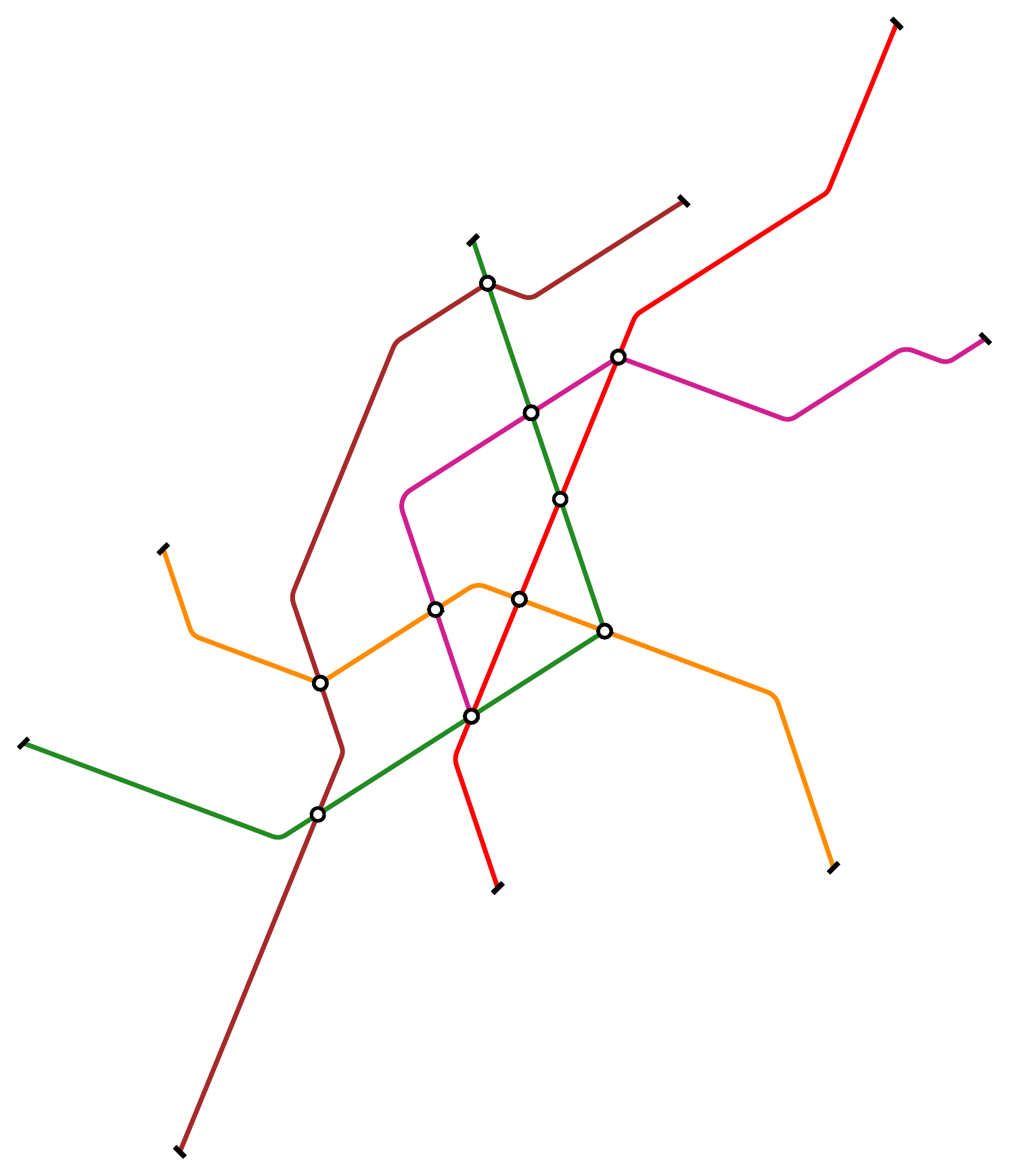} & \includegraphics[width=.19\textwidth]{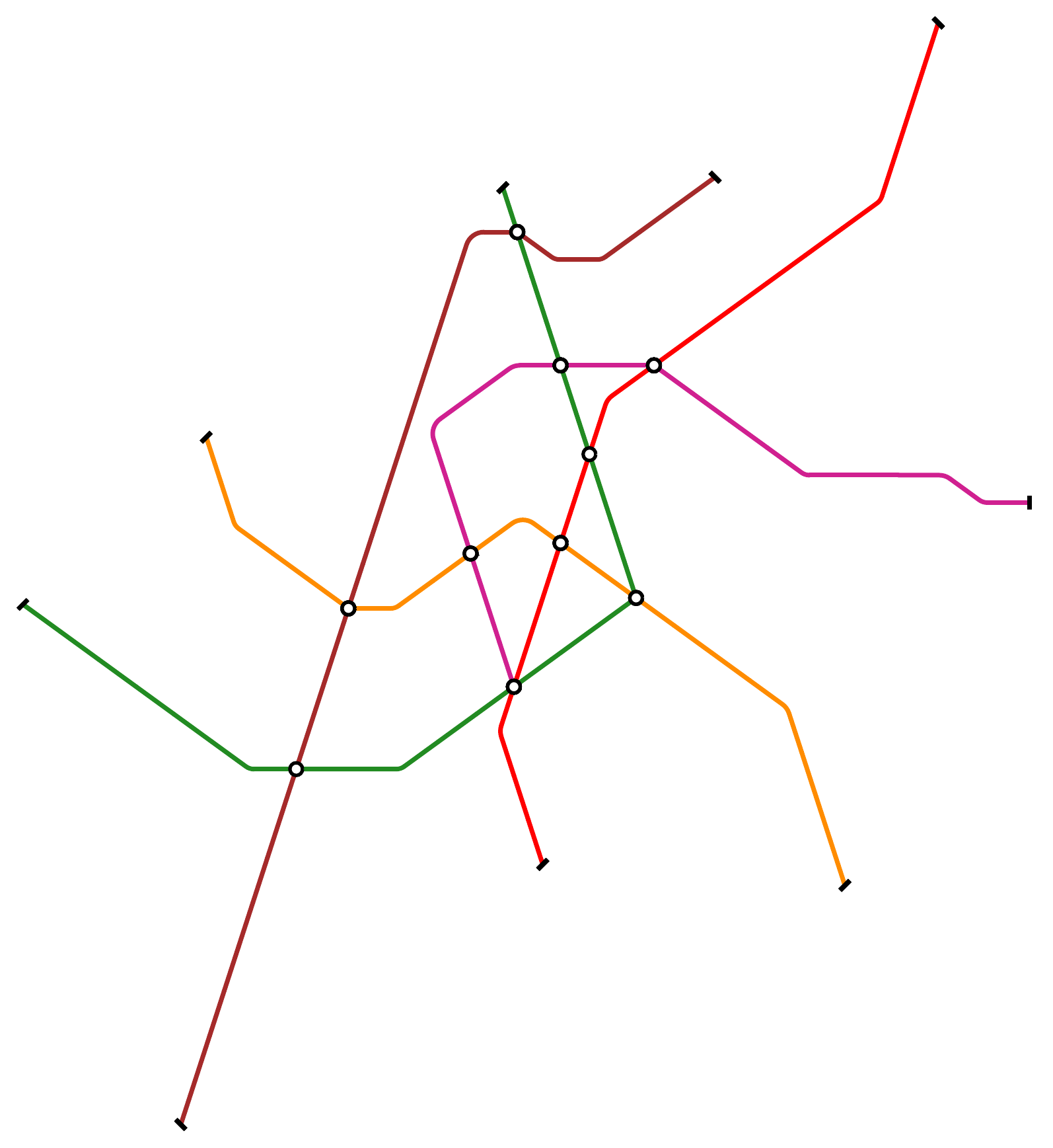} &   \includegraphics[width=.19\textwidth]{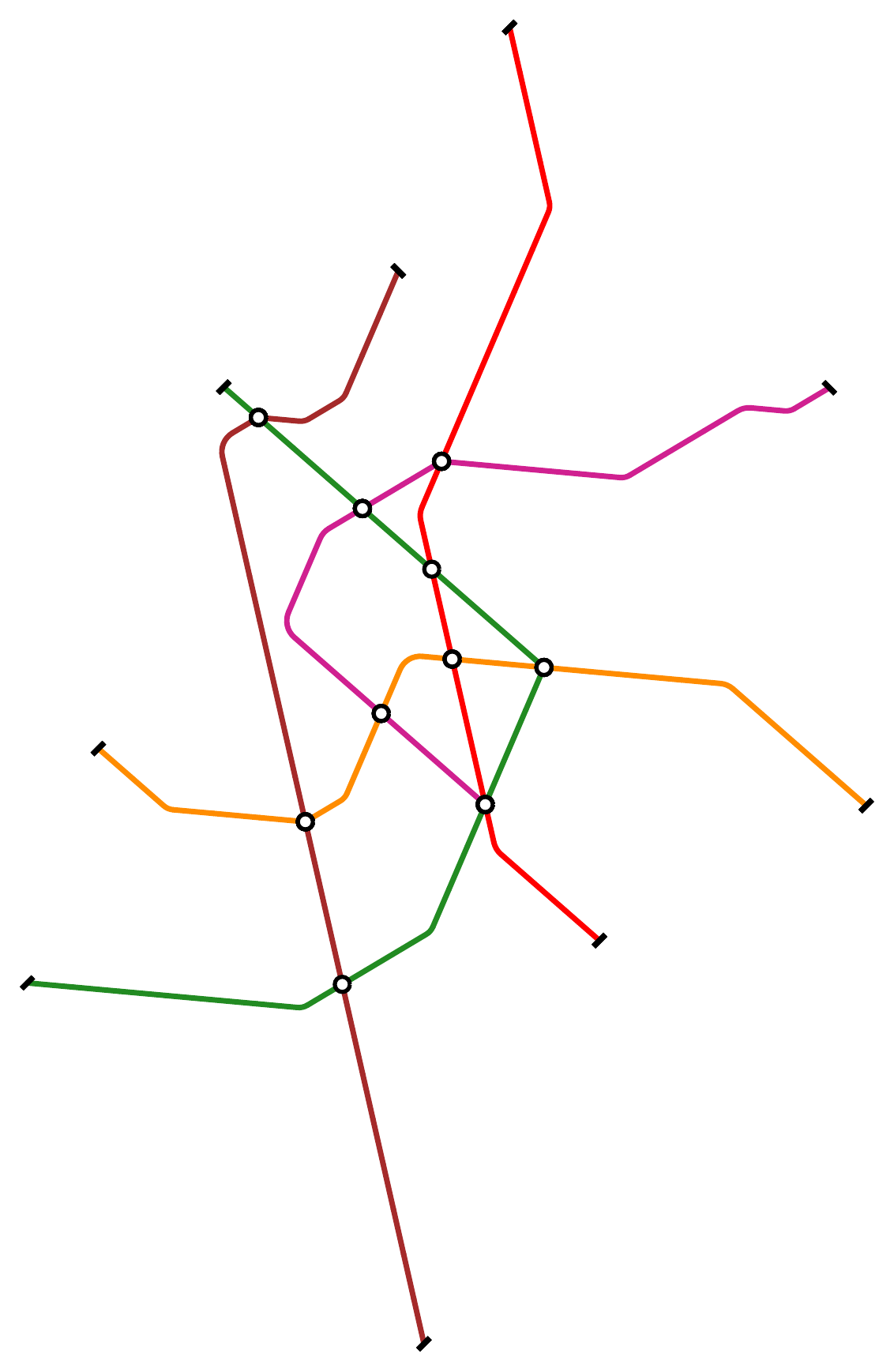} & \includegraphics[width=.19\textwidth]{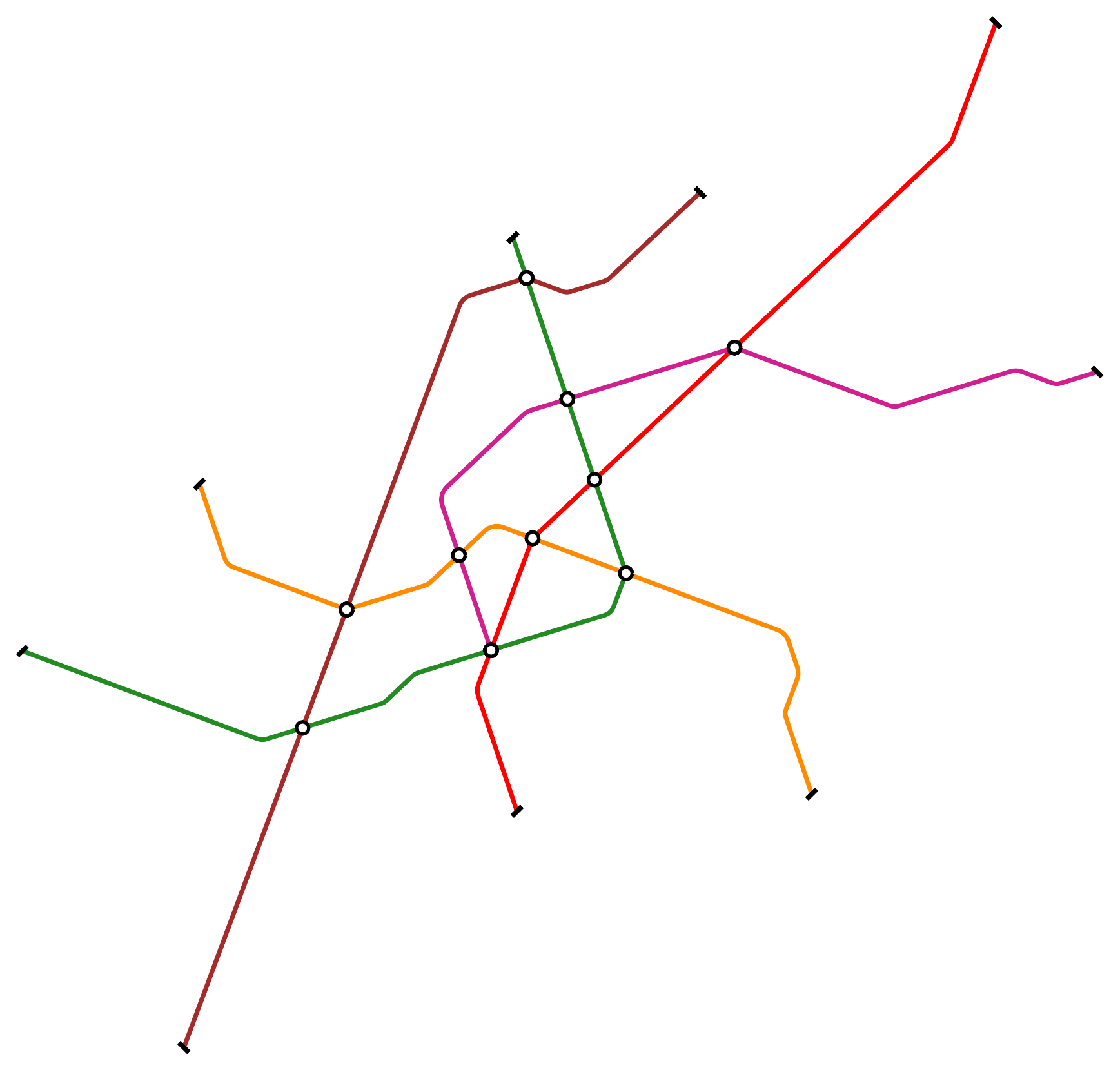}\\		
		(e) 4-R&(f) 4-I&(g) 5-A&(h) 5-R&(i) 5-I\\[6pt]
	\end{tabular}
	\caption{Examples of Vienna generated with objective function weights $(f_1, f_2, f_3) = (3, 2, 1)$ for different $k \in \{3, 4, 5\} $ and aligned ($k$-A), regular ($k$-R) and irregular ($k$-I) orientation systems.}
	\label{fig:vienna321}
\end{figure}

Here we show a representative set of nine layouts for Vienna in Figure~\ref{fig:vienna321}, the performance and quality measurements for Vienna in Table~\ref{tab:eval_v} and plots of the performance and quality measurements for the objective function weights $(f_1, f_2, f_3) = (3, 2, 1)$ in Figure~\ref{fig:stats321}. The same plots for the objective function weights $(f_1, f_2, f_3) = (10, 5, 1)$ (Figure~\ref{fig:stats1051}), the complete table of the performance and quality measurements (Table~\ref{tab:eval_2}) and the computed networks for Montreal (Figures~\ref{fig:ap_montreal321}--\ref{fig:ap_montreal1051}), Vienna (Figures~\ref{fig:ap_vienna321}--\ref{fig:ap_vienna1051}), Washington (Figures~\ref{fig:ap_washington321}--\ref{fig:ap_washington1051}) and Sydney (Figures~\ref{fig:ap_sydney321}--\ref{fig:ap_sydney1051}) can be found in the appendix. 
We set a 10-hour time limit for CPLEX, which was reached by most Sydney instances. 
Note however that during the process of solving the MIP, we have access to intermediate, but possibly suboptimal solutions.
Even in instances, which run out of time (like Sydney), we tend to find intermediate solutions, which are visually already quite close to the final solutions after only a few seconds of computation. 
Since solving these instances as close to optimality as possible increases comparability between instances via the taken measurements, we set the rather long time limit.

\setlength{\tabcolsep}{3pt}
\begin{table}
	\caption{Results for the Vienna network. The model parameters are the number of available directions ($k$) and the orientation system (\textbf{A}ligned, \textbf{R}egular, \textbf{I}rregular). The measures are the number of bends, sector deviation (total and per edge), distortion per edge and the runtime in seconds. For one set of objective function weights, the best results across different $k$ and orientation systems for every measure are marked in bold numbers.}
	\label{tab:eval_v}
	
	\begin{center}
		\begin{tabular}{c!{\vrule width 1pt}l|rrr|rrr|rrr}
			&\multicolumn{1}{l|}{instance} & \multicolumn{9}{c}{Vienna}\\
			\cmidrule{2-11}
			&\multicolumn{1}{l|}{drawing dimensions}& \multicolumn{3}{c|}{$k=3$} & \multicolumn{3}{c|}{$k=4$} & \multicolumn{3}{c}{$k=5$}\\
						\cmidrule{2-11}
			\multirow{-3}{*}{\rotatebox[origin=c]{90}{weights}}&\multicolumn{1}{l|}{orientation system} & \centhead{A}     & \centhead{R} & \centheadbar{I}   & \centhead{A}     & \centhead{R}     & \centheadbar{I}   & \centhead{A}     & \centhead{R}     & \centhead{I}\\
			\ChangeRT{1pt}
			
			\rowcolor{gray!20}& \bendname        & \textbf{16}          & \textbf{16}      & 17          & 22          & 24          & 21          & 25          & 25          & 29\\
			\rowcolor{gray!20}& \sectdevname    & 27          & 27      & \textbf{13}          & 21          & 18          & 17          & 24          & 24          & 21\\
			\rowcolor{gray!20}& \subper             & 0.28        & 0.28    & \textbf{0.14}        & 0.22        & 0.19        & 0.18        & 0.25        & 0.25        & 0.22\\
			\rowcolor{gray!20}& \distpername & 31.47       & 36.07   & 15.96       & 23.18       & 22.96       & 16.07       & 19.45       & 26.68       & \textbf{14.46}\\
			\rowcolor{gray!20}\billy{$(3, 2, 1)$}     & \runtimename  & 308         & 349     & \textbf{8}           & 108         & 116         & 299         & 69          & 217         & 113\\
			\ChangeRT{0.4pt}
			
			& \bendname        & 16      & 16      & \textbf{15}          & 19      & 19          & 19          & 25          & 25          & 29\\
			& \sectdevname    & 25      & 25      & \textbf{19}          & 27      & 27          & \textbf{19}          & 24          & 24          & 23\\
			& \subper             & 0.26    & 0.26    & \textbf{0.2}         & 0.28    & 0.28        & \textbf{0.2}         & 0.25        & 0.25        & 0.24\\
			& \distpername & 31.18   & 33.5    & 17.76       & 25.19   & 23.53       & 16.35       & 19.45       & 26.68       & \textbf{15.01}\\
			\billy{$(10, 5, 1)$}   & \runtimename  & 53      & 39      & \textbf{8}           & 140     & 115         & 41          & 44          & 27          & 51\\

		\end{tabular}
	\end{center}
\end{table}

Specific instances in this section will be referred to by name followed by their number of orientations $k$ and the weights $f_1, f_2, f_3$ in parentheses. %
Our first observation from generalizing the model of N\"ollenburg and Wolff~\cite{nw-dlhqm-11} is that the MIP model size, i.e., the numbers of constraints and variables scale linearly with the number $k$ of orientations. 
So as long as $k$ is a (small) constant the asymptotics with respect to the graph size parameters $n$ and $m$ remain the same.
Yet, in practice, doubling the size of the model may yield a significant slow-down in the actual solution time.

Comparing the number of bends, we can see that increasing $k$ leads to a greater number of bends. This can be explained in part by an increase in forced bends, where the probability that consecutive edges in a metro line cannot possibly be drawn in the same direction increases with $k$. 
This could be counteracted by increasing the parameter $s$, i.e., allowing more than two admissible neighboring sectors for each edge. 
The difference in bends between the aligned and the regular orientation systems is small, with only a few instances where the difference is up to 2. 
A very similar picture emerges between the aligned and irregular systems, however the spread is a bit bigger (1--4 bends difference on average) with some extreme examples, e.g., Washington, $k=4$, (3, 2, 1), where we save 7 bends with an irregular system. 
Unsurprisingly we have a similar or smaller amount of bends, when emphasizing the objective function that minimizes bends from $f_1=3$ to $f_1=10$.

Montreal and Sydney see the smallest sector deviation for $k=3$, Vienna and Washington depend on the parameter setting, i.e., they have the lowest value for $k=4$ with settings (3, 2, 1) and $k=5$ for settings (10, 5, 1), respectively. While the values of the aligned and regular systems are very similar, we see improvements ranging from slight ($k=5$) to significant ($k=4$) in the irregular orientation systems except for Montreal ($k=3$).

The general picture of the actual angle distortion is that, while a regular but rotated orientation system increases the distortion (with the notable exception of Vienna, $k=4$), we achieve better values across the board when employing the irregular system, sometimes drastically so, with a reduction of more than 50\% of the distortion in some instance, e.g., Vienna, $k=3$, (3, 2, 1). Again, rather unsurprisingly, we increase the angle distortion, when we emphasize line straightness with weights (10, 5, 1). 

The behavior regarding $k$ shows that the distortion decreases with increasing $k$, since the maximally possible angle distortion for each edge decreases. This is a trend we can see generally. However, this makes the few outliers all the more interesting. These are Montreal, (3, 2, 1), where $k=4$ achieves the lowest distortion and Vienna, (3, 2, 1), where $k=3$ gives a lower distortion than $k=4$.

\begin{figure}
	\centering
	\includegraphics{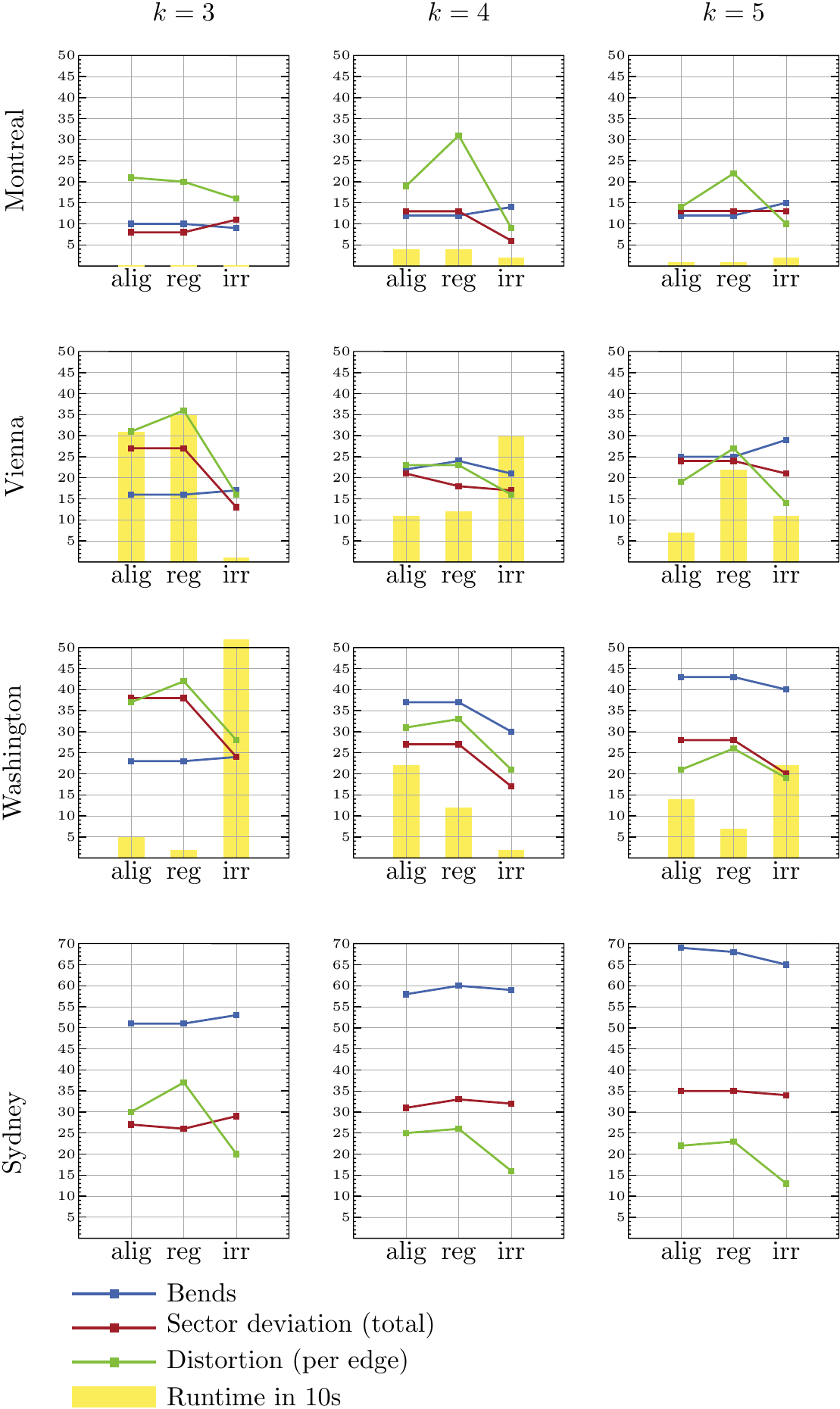}
	\caption{Plots of the results of the experiments for the objective function weights $(f_1, f_2, f_3) = (3, 2, 1)$.}
	\label{fig:stats321}
\end{figure}

Even though the size of the MIPs does not change, the runtime between instances changes, sometimes drastically so. On average the more extreme setting of objective function weights (10, 5, 1), in which a single goal is more pronounced over the other goals is solved faster than the more balanced setting (3, 2, 1). The influence of the irregular orientation system over the aligned one is unclear. Some instances are solved significantly faster, e.g., Vienna, (3, 2, 1) with $k=3$ drops by 300 seconds to just 8, other instances heavily increase in runtime, e.g., Washington, (3, 2, 1) with $k=3$ jumps from 52 to 524 seconds, while yet other instances see no significant difference in runtime. While the runtimes of the regular systems seem to be more similar to their aligned counterparts, it is in general unclear how exactly the orientation system influences the MIPs runtime.

\subsection{Discussion}

Our approach of increasing topographicity in metro maps through data-driven orientation systems seems to be working in almost all computed instances. Choosing an irregular orientation system  is a valid option to increase topographicity, even if the irregular set of slopes is unfamiliar. Moreover we can see that in isolated instances distortion can be lower when restricted to a smaller set of directions, which might be an indication that some input maps lend themselves more naturally to a specific $k$, which is not always the octolinear $k=4$.

Looking at the actual metro maps produced by our system, we can see one major caveat of our approach to minimize distortion by deciding the directions based on the input. While for most edges we have a very suitable representative direction in the orientation system, the constraints of the MIP might still force an edge to be drawn in a different sector. In Figure~\ref{fig:vienna321}(e) the lower end of the brown line is drawn downwards, to the right, despite there being a direction available which is closer to its general direction in the input. However due to the previous direction of the line and crossings with the orange and green lines and the goal of avoiding bends it is drawn in a direction with more topographic distortion.

On a positive note we can see that irregular orientation systems can create metro maps that resemble the input more closely than typical aligned systems. This can be seen when comparing Figures~\ref{fig:vienna321} (a) and (c), (d) and (f), and (g) and (i), respectively.

We can also see that most of the layouts, which are not using an aligned orientation system do not include the horizontal direction. This might be helpful in labeling these metro maps, since it is difficult to place the visually preferred horizontal labels along a horizontal line with clear association to a station.

\section{Conclusions}

We presented and implemented an adaptation of an existing MIP model for octolinear metro maps~\cite{nw-dlhqm-11} that can draw metro maps schematized to any set $\mathcal{C}$ of arbitrary orientations. 
This is supplemented by a data-driven approach to optimize the set $\mathcal{C}$ based on $k$-means clustering of the input edge orientations or by finding the best rotation of a regular orientation system. Finally we performed, presented, and discussed experiments of our system and its results for different real-world metro maps. 

An approach to choose a suitable $k$ for a given input might be to use the smallest $k$ which generates visually appealing layouts in order to find a middle ground between the schematic appearance of the metro map and geographic accuracy. This leads to the general idea, that it is still important for a metro map designer to consider a number of possible layouts in different linearities for a given input network in order to find the most suitable metro map style. Hence our system should not be understood as a stand-alone method to metro map generation, but rather as an automated tool to help pre-select possible candidates for layouts and increase the number of layout settings a designer can explore at low time cost. This pre-selection might be refined in the future if a more global metric to judge the quality of a metro maps is devised.

As future work, we want to include station labeling~\cite{nh-aflnm-18} and investigate additional improvements of the practical performance, as well as integrating it into a human-in-the-loop tool for schematic map design.

\medskip
{\small
\noindent \textbf{Acknowledgments.} We thank Maxwell J. Roberts for discussions about non-standard linearity models.}

\newpage
\appendix

\section*{Appendix}

\begin{figure}[!h]
	\centering
	\includegraphics{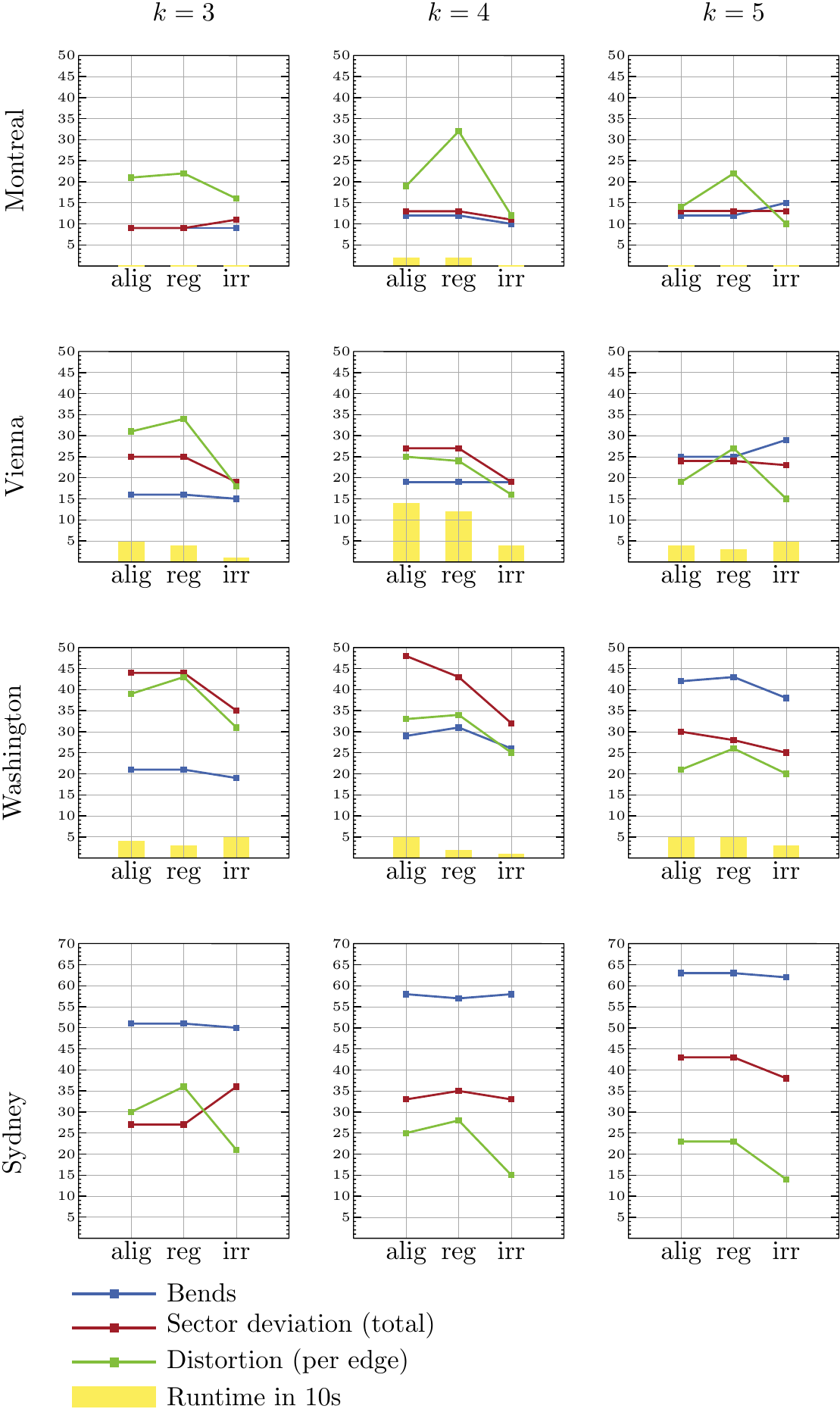}
	\caption{Plots of the results of the experiments for the objective function weights $(f_1, f_2, f_3) = (10, 5, 1)$.}
	\label{fig:stats1051}
\end{figure}

\begin{landscape}
	
	\setlength{\tabcolsep}{2pt}
	\begin{table}
		\caption{Full results for all experiments. The model parameters are the number of available directions ($k$) and the orientation system (\textbf{A}ligned, \textbf{R}egular, \textbf{I}rregular), as well as the objective function weights. The measures are the number of bends, sector deviation (total and per edge), distortion per edge and the runtime in seconds. For one set of objective function weights, the best results across different $k$ and orientation systems for every measure are marked in bold numbers. Runtimes marked as ``--'' reached the 10 hour time limit.}
		\label{tab:eval_2}
		
		\begin{center}
			\begin{tabular}{c!{\vrule width 1pt}l|rrr|rrr|rrr!{\vrule width 1pt}rrr|rrr|rrr}
				&\multicolumn{1}{l|}{obj. fct. weights} & \multicolumn{9}{c!{\vrule width 1pt}}{$(f_1, f_2, f_3) = (3, 2, 1)$} & \multicolumn{9}{c}{$(f_1, f_2, f_3) = (10, 5, 1)$}\\
				\cmidrule{2-20}
				&\multicolumn{1}{l|}{drawing dimensions}& \multicolumn{3}{c|}{$k=3$} & \multicolumn{3}{c|}{$k=4$} & \multicolumn{3}{c!{\vrule width 1pt}}{$k=5$}
				& \multicolumn{3}{c|}{$k=3$} & \multicolumn{3}{c|}{$k=4$} & \multicolumn{3}{c}{$k=5$}\\
				\cmidrule{2-20}
				\multirow{-3}{*}{\rotatebox[origin=c]{90}{$\leftarrow$ Inst.}}&\multicolumn{1}{l|}{orientation system} & \centhead{A}     & \centhead{R} & \centheadbar{I}   & \centhead{A}     & \centhead{R}     & \centheadbar{I}   & \centhead{A}     & \centhead{R}     & \centheadbartwo{I}
				& \centhead{A}    & \centhead{R} & \centheadbar{I}   & \centhead{A}     & \centhead{R}     & \centheadbar{I}   & \centhead{A}     & \centhead{R}     & \centhead{I}\\
				\ChangeRT{1pt}
				\rowcolor{gray!20}& \bendname        & 10          & 10      & \textbf{9}           & 12          & 12          & 14          & 12          & 12          & 15    
				& \textbf{9}       & \textbf{9}       & \textbf{9}           & 12      & 12          & 10          & 12          & 12          & 15    \\
				
				\rowcolor{gray!20}& \sectdevname    & 8           & 8       & 11          & 13          & 13          & \textbf{6}           & 13          & 13          & 13          
				& \textbf{9}       & \textbf{9}       & 11          & 13      & 13          & 11          & 13          & 13          & 13\\
				
				\rowcolor{gray!20}& \subper             & 0.12        & 0.12    & 0.17        & 0.2         & 0.2         & \textbf{0.09}        & 0.2         & 0.2         & 0.2         
				& \textbf{0.14}    & \textbf{0.14}    & 0.17        & 0.2     & 0.2         & 0.17        & 0.2         & 0.2         & 0.2\\
				
				\rowcolor{gray!20}& \distpername & 21.34       & 20.24   & 15.63       & 18.73       & 30.98       & \textbf{8.99}        & 14.34       & 22.43       & 10.22       
				& 21.28   & 22.38   & 15.63       & 18.69   & 32.34       & 11.86       & 14.34       & 22.43       & \textbf{10.22}\\
				
				\rowcolor{gray!20}
				\billy{Montreal}   & \runtimename  & \textbf{2}           & 3       & 3           & 43          & 39          & 18          & 10          & 9           & 15          
				& \textbf{1}       & 2       & 2           & 20      & 15          & 4           & 3           & 4           & 3\\
				\ChangeRT{0.4pt}
				& \bendname        & \textbf{16}          & \textbf{16}      & 17          & 22          & 24          & 21          & 25          & 25          & 29    
				& 16      & 16      & \textbf{15}          & 19      & 19          & 19          & 25          & 25          & 29\\
				& \sectdevname    & 27          & 27      & \textbf{13}          & 21          & 18          & 17          & 24          & 24          & 21          
				& 25      & 25      & \textbf{19}          & 27      & 27          & \textbf{19}          & 24          & 24          & 23\\
				& \subper             & 0.28        & 0.28    & \textbf{0.14}        & 0.22        & 0.19        & 0.18        & 0.25        & 0.25        & 0.22        
				& 0.26    & 0.26    & \textbf{0.2}         & 0.28    & 0.28        & \textbf{0.2}         & 0.25        & 0.25        & 0.24\\
				& \distpername & 31.47       & 36.07   & 15.96       & 23.18       & 22.96       & 16.07       & 19.45       & 26.68       & \textbf{14.46}       
				& 31.18   & 33.5    & 17.76       & 25.19   & 23.53       & 16.35       & 19.45       & 26.68       & \textbf{15.01}\\
				\billy{Vienna}     & \runtimename  & 308         & 349     & \textbf{8}           & 108         & 116         & 299         & 69          & 217         & 113         
				& 53      & 39      & \textbf{8}           & 140     & 115         & 41          & 44          & 27          & 51\\
				\ChangeRT{0.4pt}
				
				\rowcolor{gray!20}& \bendname        & \textbf{23}          & \textbf{23}      & 24          & 37          & 37          & 30          & 43          & 43          & 40    
				& 21      & 21      & \textbf{19}          & 29      & 31          & 26          & 42          & 43          & 38\\
				
				\rowcolor{gray!20}& \sectdevname    & 38          & 38      & 24          & 27          & 27          & \textbf{17}          & 28          & 28          & 20          
				& 44      & 44      & 35          & 48      & 43          & 32          & 30          & 28          & \textbf{25}\\
				
				\rowcolor{gray!20}& \subper             & 0.38        & 0.38    & 0.24        & 0.27        & 0.27        & \textbf{0.17}        & 0.28        & 0.28        & 0.2         
				& 0.44    & 0.44    & 0.35        & 0.48    & 0.43        & 0.32        & 0.3         & 0.28        & \textbf{0.25}\\
				
				\rowcolor{gray!20}& \distpername & 37.18       & 41.78   & 28.25       & 30.79       & 33.47       & 21.35       & 21.15       & 26.09       & \textbf{19.3}        
				& 38.61   & 42.72   & 31.08       & 33.01   & 34.05       & 24.79       & 21.45       & 26.09       & \textbf{20.4}\\
				
				\rowcolor{gray!20}
				\billy{Washington} & \runtimename  & 52          & \textbf{23}      & 524         & 224         & 123         & \textbf{23}          & 140         & 69          & 215         
				& 36      & 32      & 48          & 47      & 23          & \textbf{3}           & 48          & 54          & 25\\
				\ChangeRT{0.4pt}
				& \bendname        & \textbf{51}          & \textbf{51}      & 53          & 58          & 60          & 59          & 69          & 68          & 65    
				& 51      & 51      & \textbf{50}          & 58      & 57          & 58          & 63          & 63          & 62\\
				& \sectdevname    & 27          & \textbf{26}      & 29          & 31          & 33          & 32          & 35          & 35          & 34          
				& \textbf{27}      & \textbf{27}      & 36          & 33      & 35          & 33          & 43          & 43          & 38\\
				& \subper             & 0.15        & \textbf{0.14}    & 0.16        & 0.17        & 0.18        & 0.17        & 0.19        & 0.19        & 0.19        
				& \textbf{0.15}    & \textbf{0.15}    & 0.2         & 0.18    & 0.19        & 0.18        & 0.23        & 0.23        & 0.21\\
				& \distpername & 30.41       & 36.68   & 20.26       & 24.68       & 26.31       & 15.78       & 21.61       & 23.33       & \textbf{13.3}        
				& 30.17   & 36.36   & 21.13       & 24.91   & 27.56       & 15.06       & 22.5        & 23.28       & \textbf{14.16}\\
				\billy{Sydney}     & \runtimename  & -- & \textbf{19825}   & -- &-- & -- & -- & -- & -- & -- 
				& 27088   & \textbf{1370}    & -- & 10125   & -- & -- & -- & -- & --\\
			\end{tabular}
		\end{center}
	\end{table}
	
\end{landscape}

\begin{figure}[b!]
\centering
\begin{tabular}{ccc}
	&\fbox{\includegraphics[scale=.25]{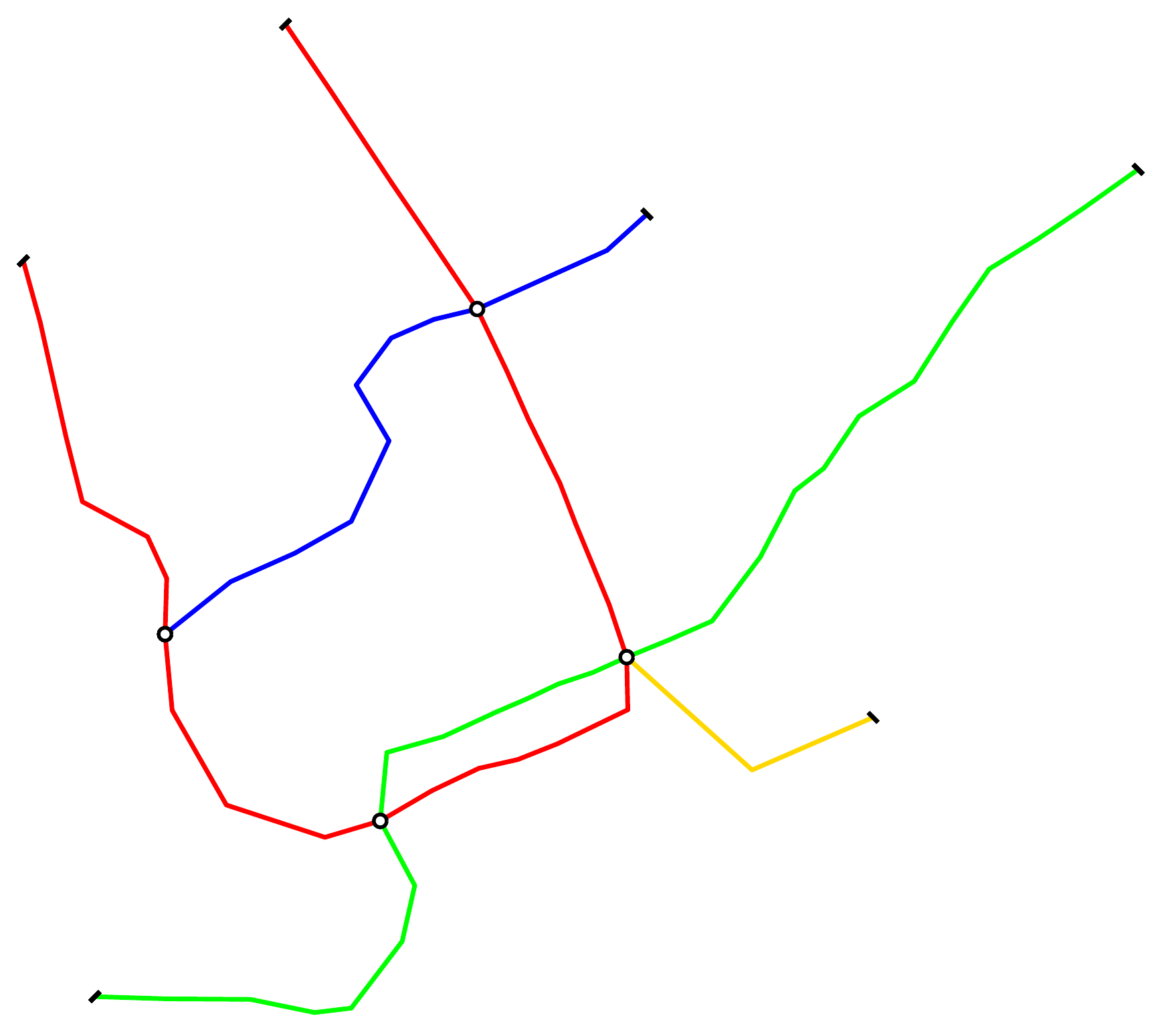}}&\\
	\includegraphics[scale=.25]{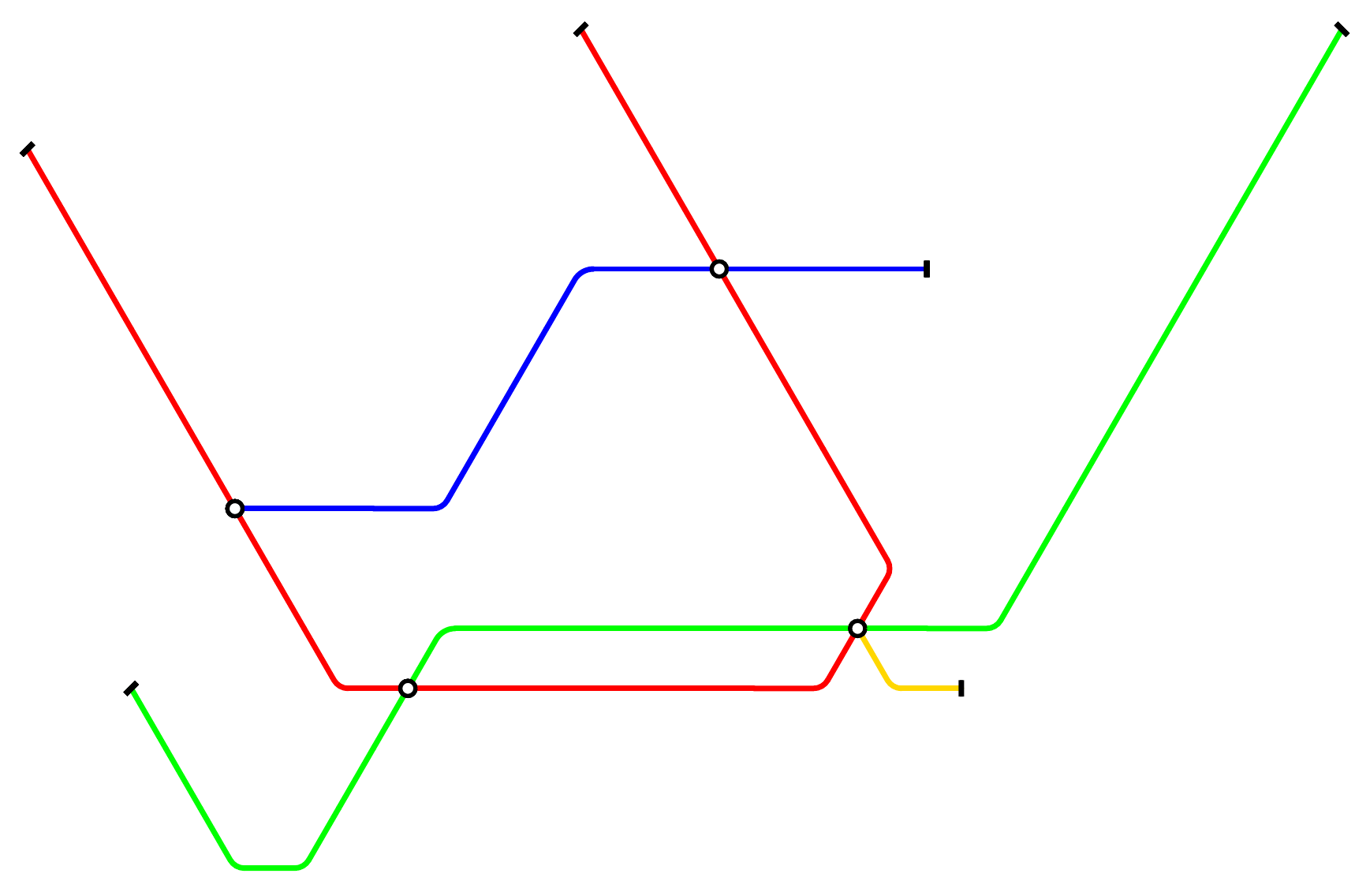} &   \includegraphics[scale=.25]{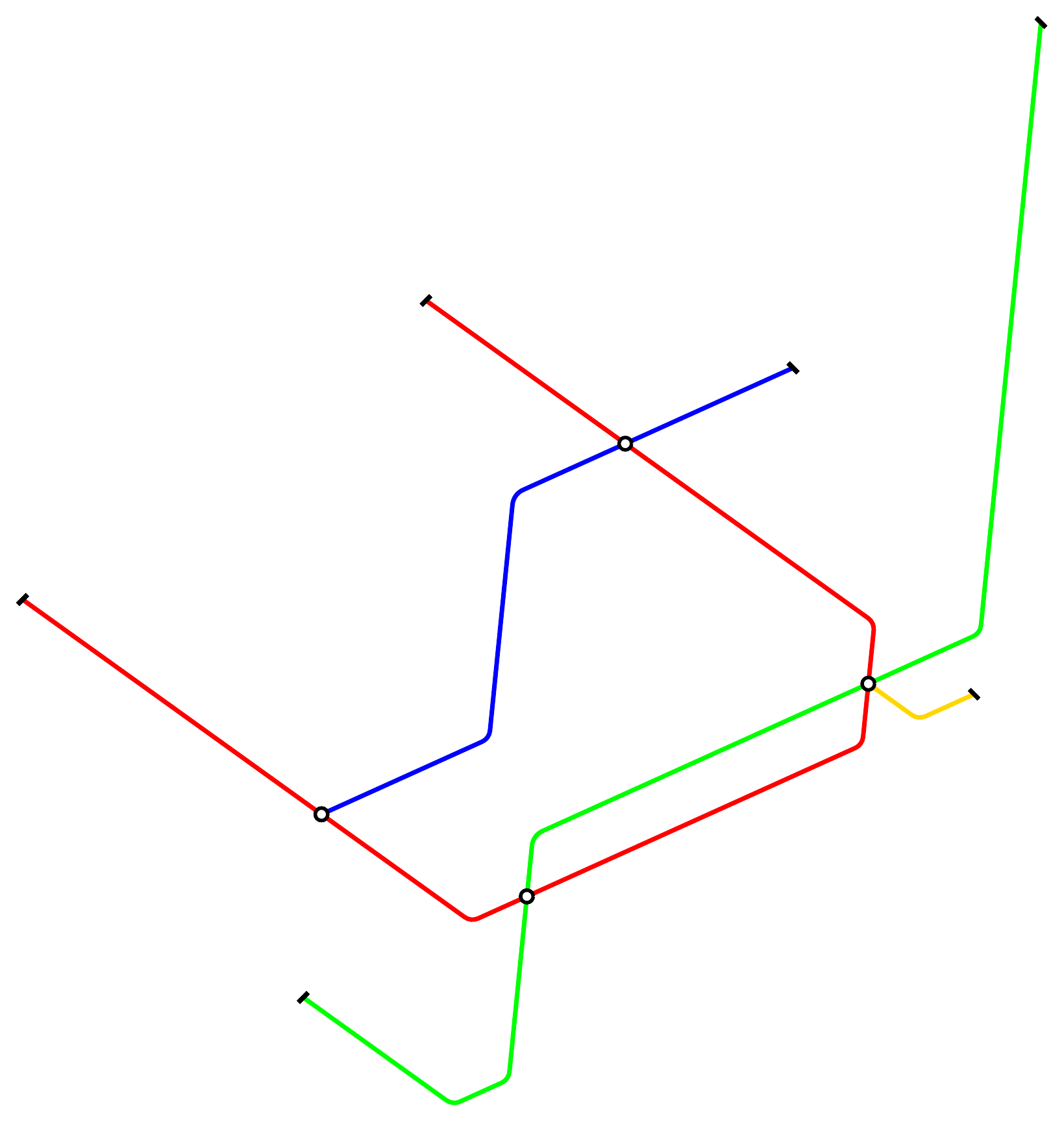} & \includegraphics[scale=.25]{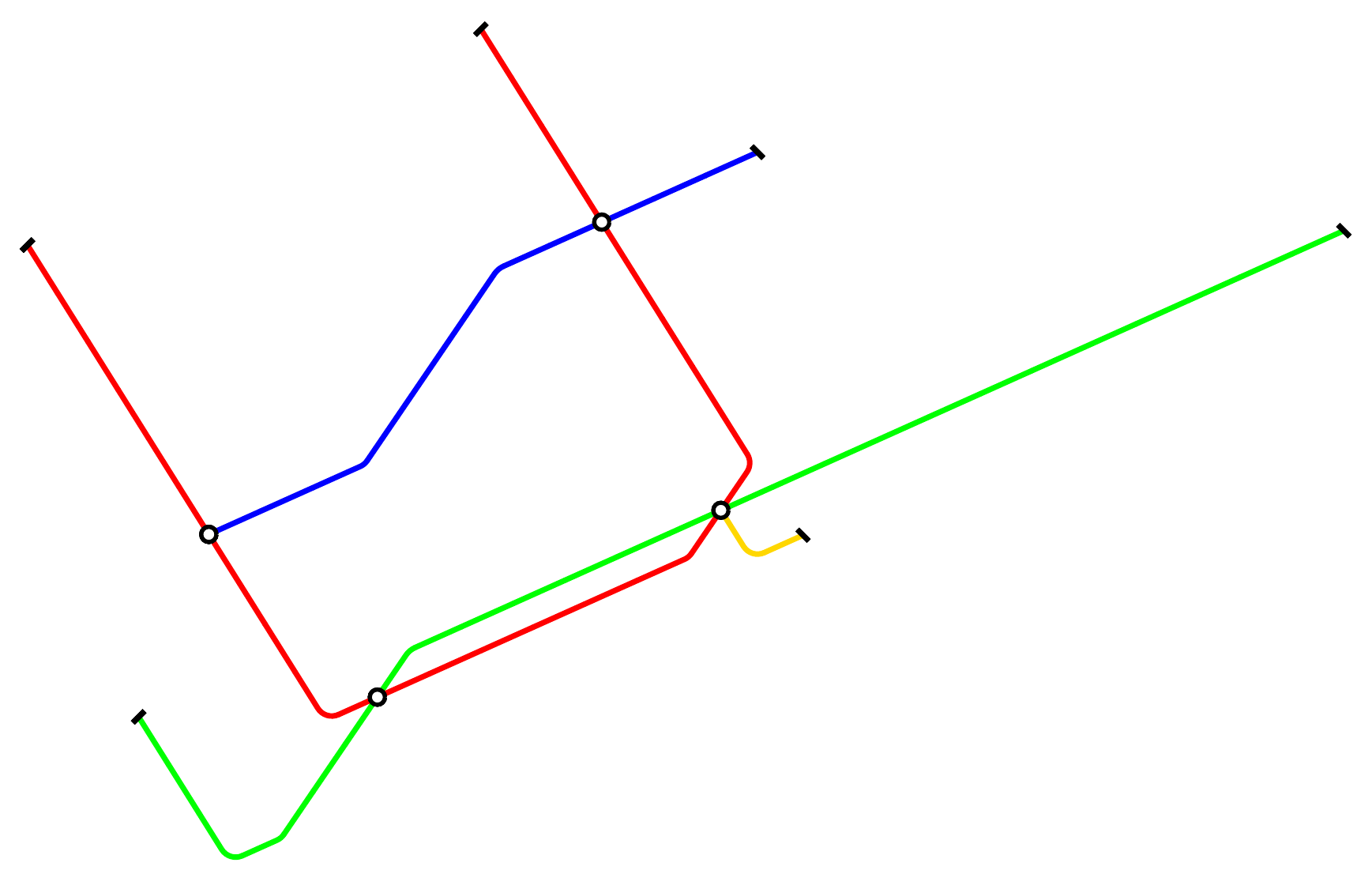} \\
	(a) 3-A & (b) 3-R & (c) 3-I \\[6pt]
	\includegraphics[scale=.25]{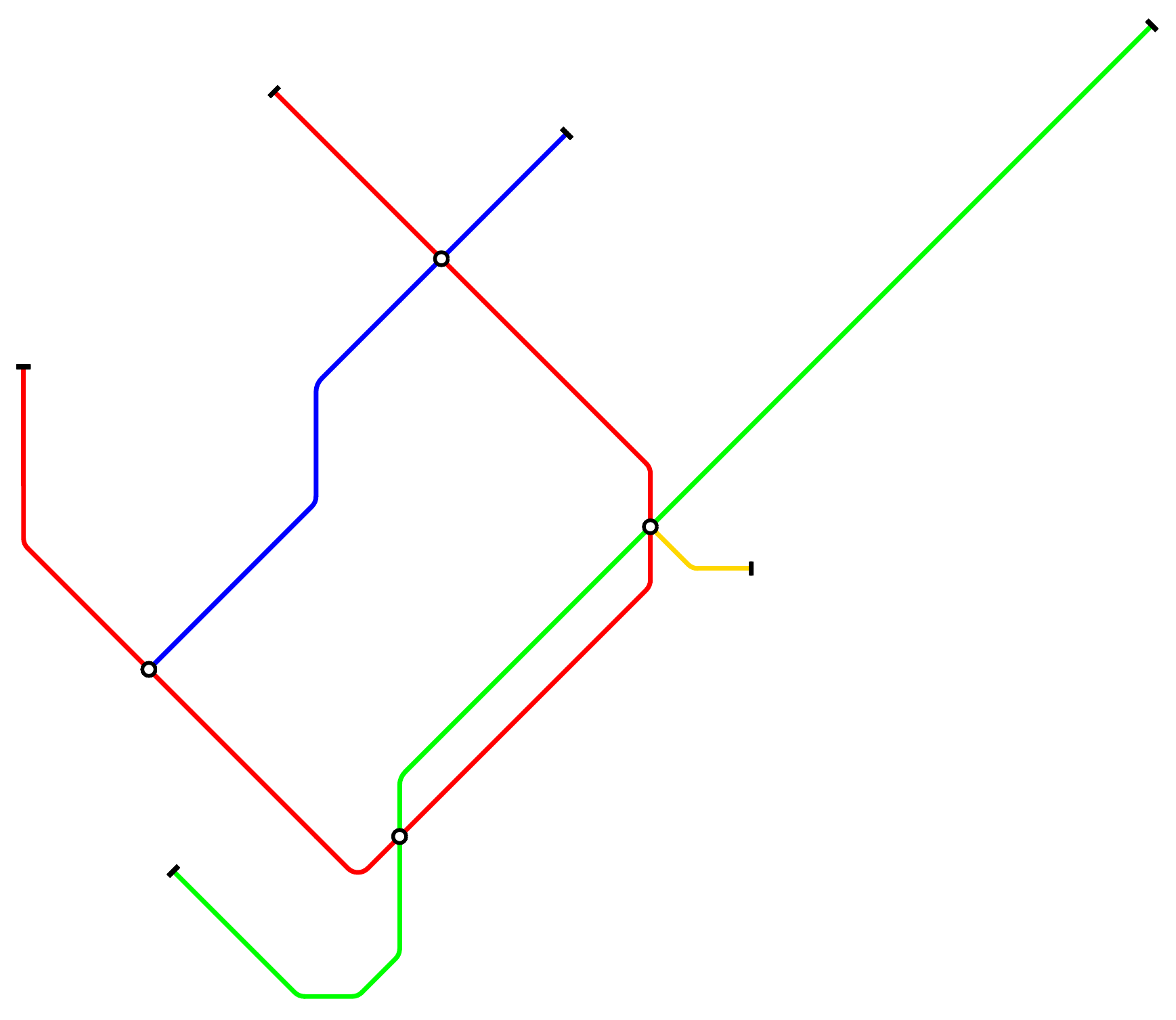} &   \includegraphics[scale=.25]{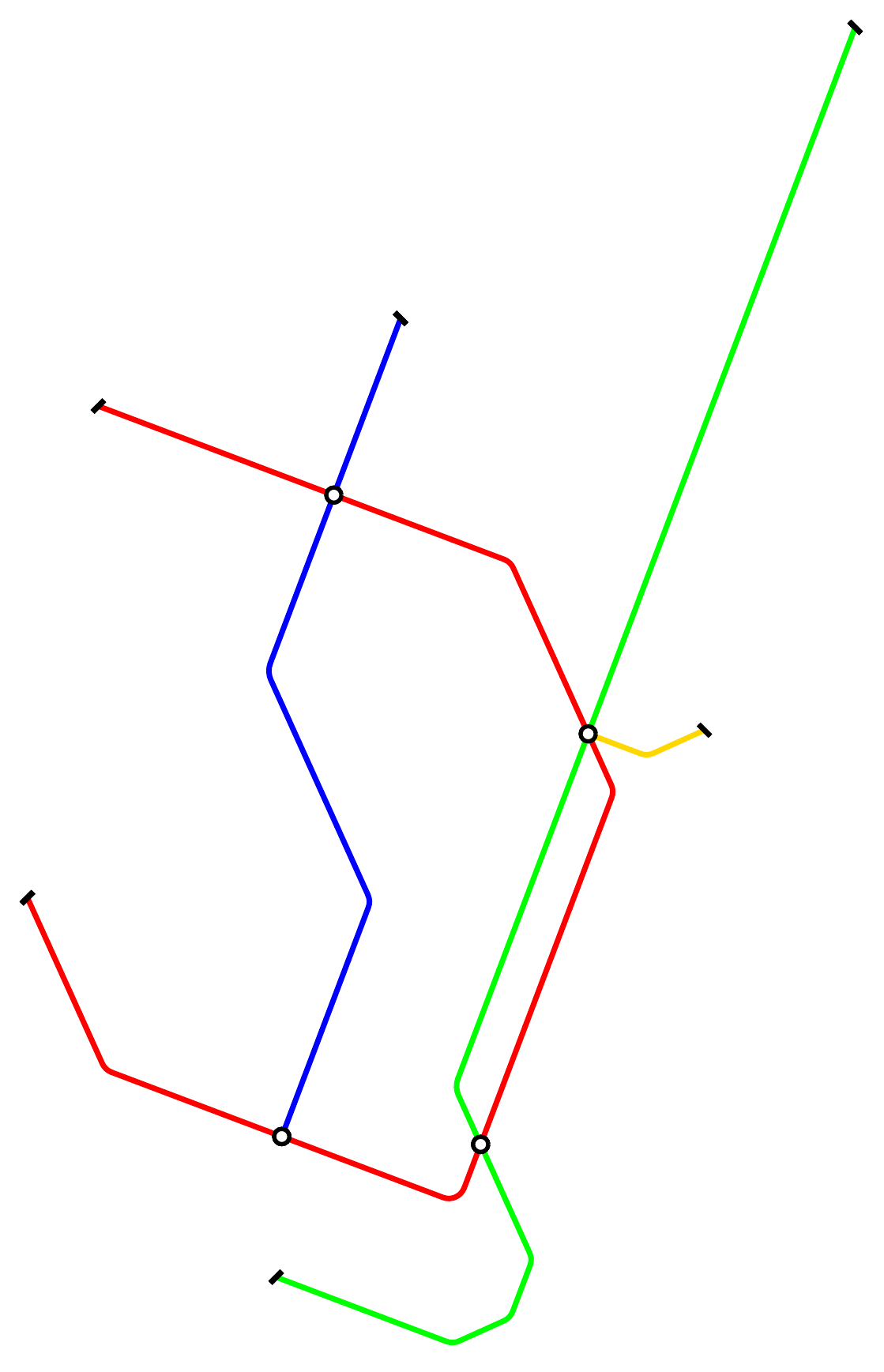} & \includegraphics[scale=.25]{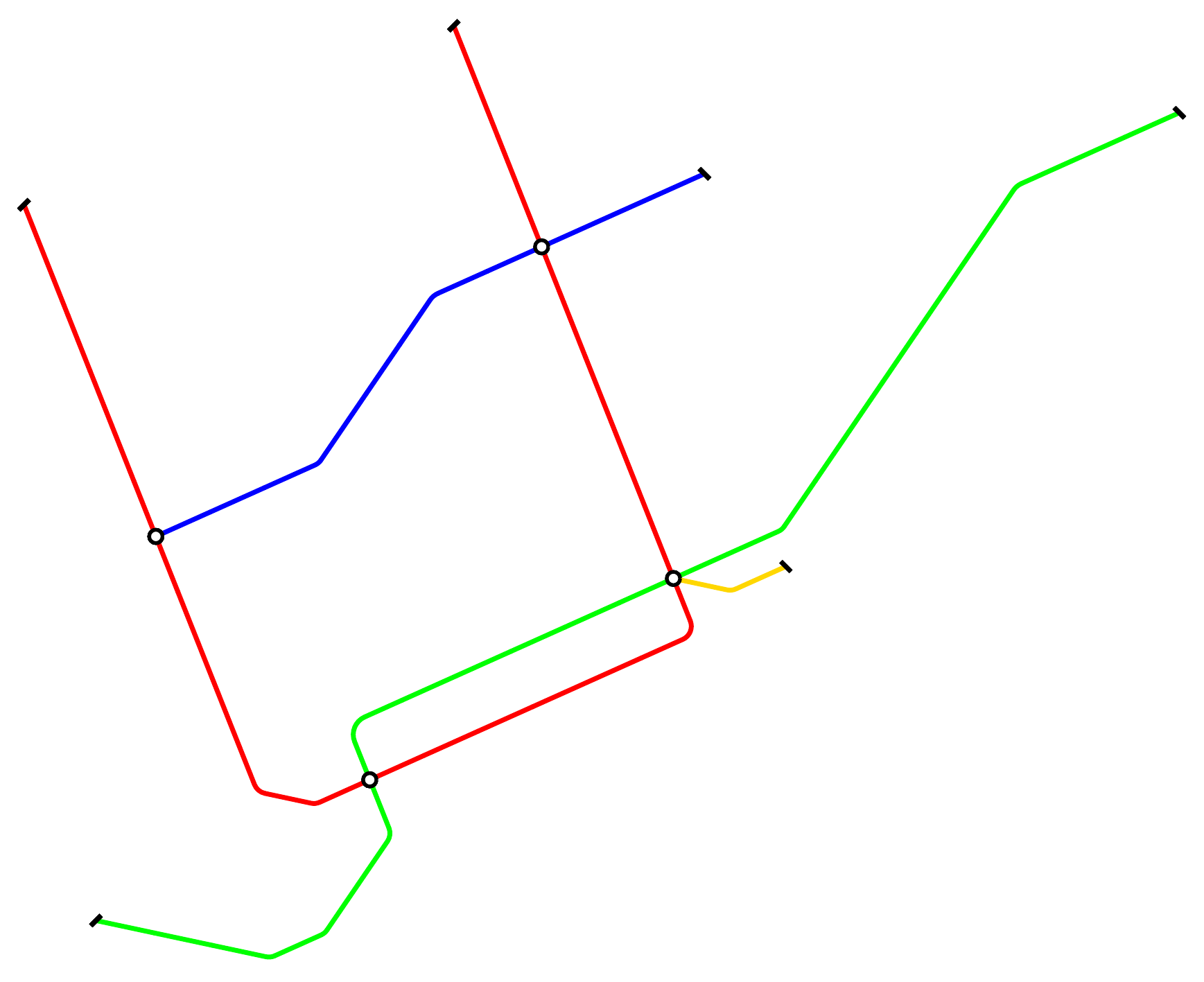} \\
	(d) 4-A & (e) 4-R & (f) 4-I \\[6pt]
	\includegraphics[scale=.25]{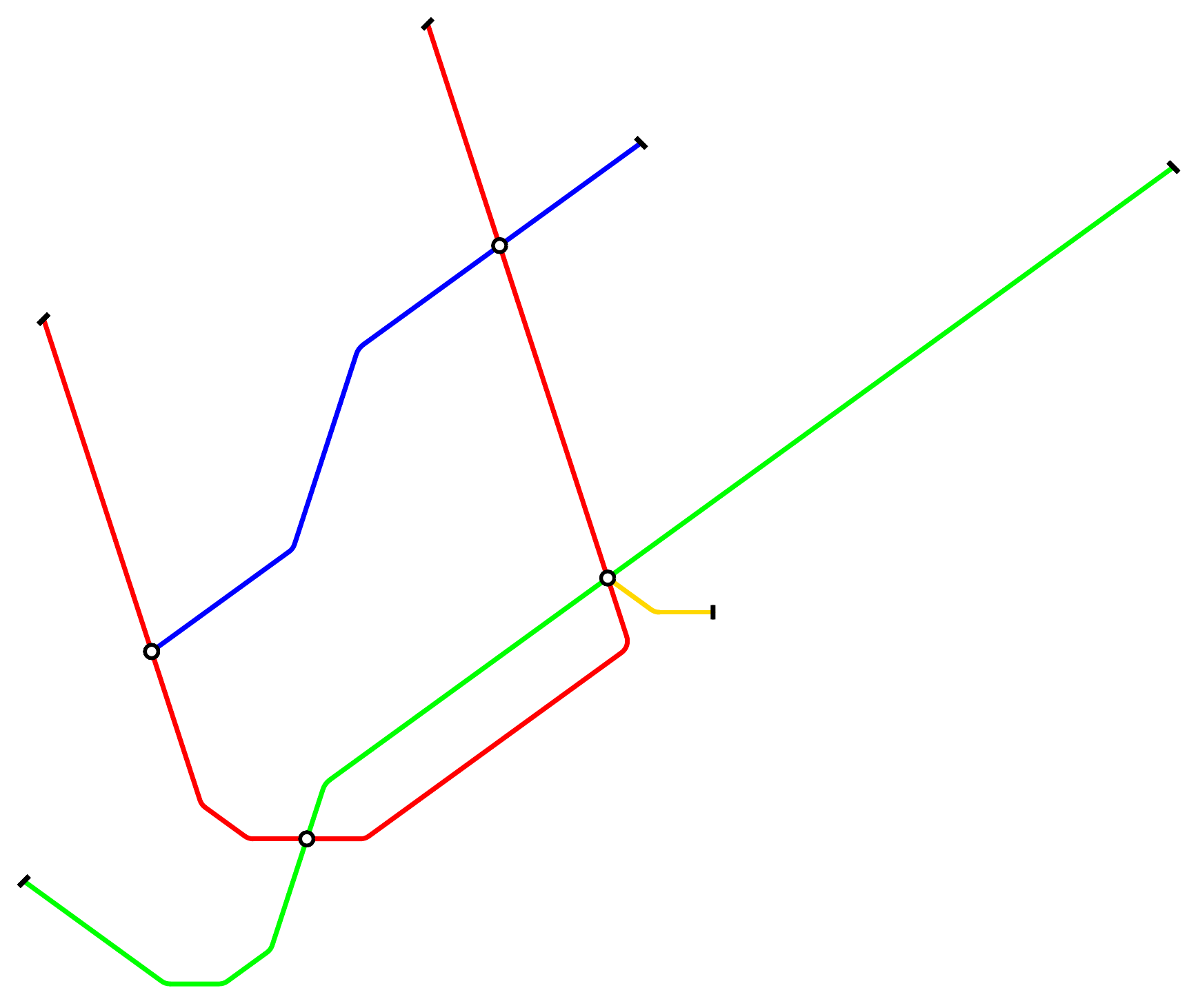} &   \includegraphics[scale=.25]{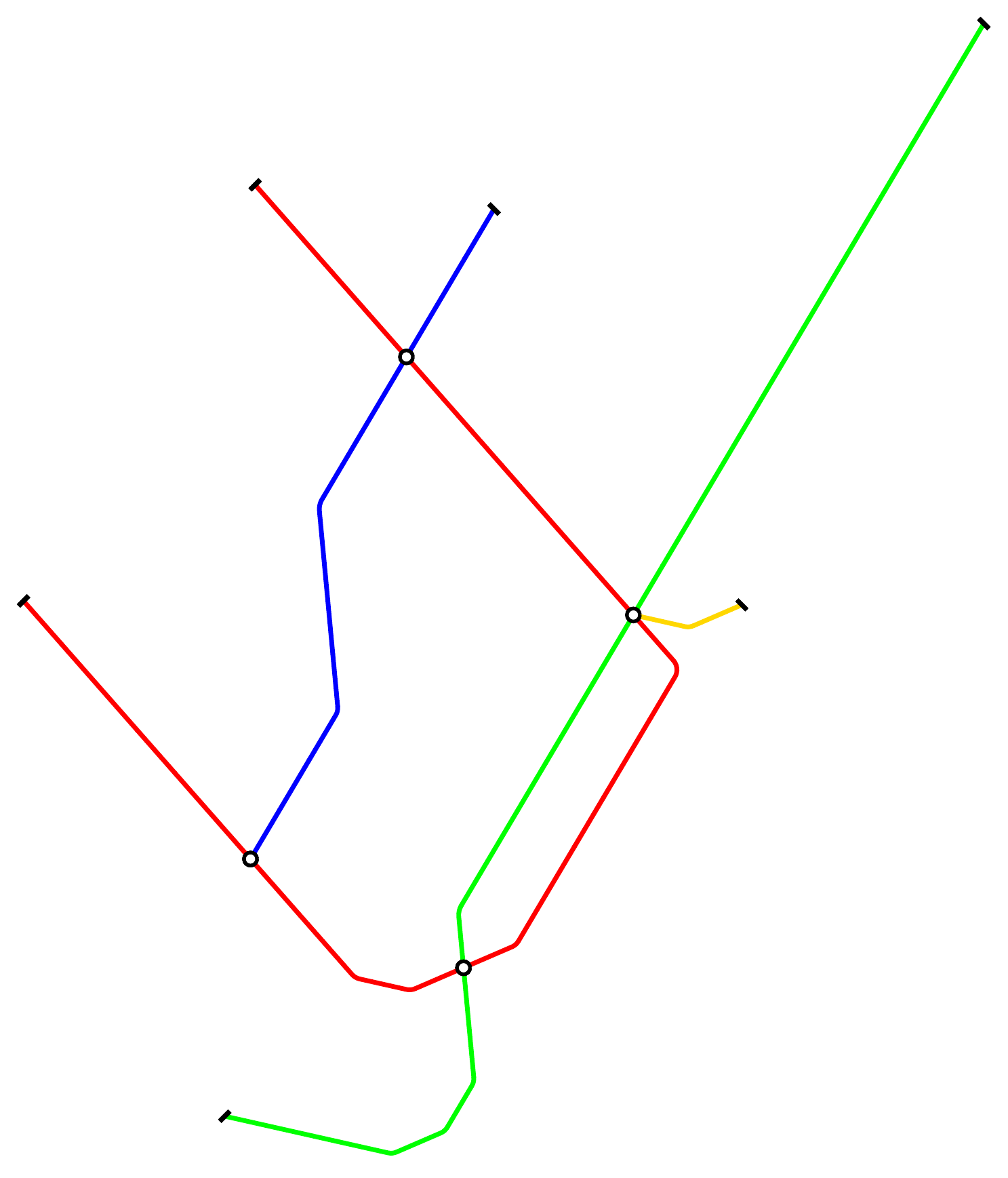} & \includegraphics[scale=.25]{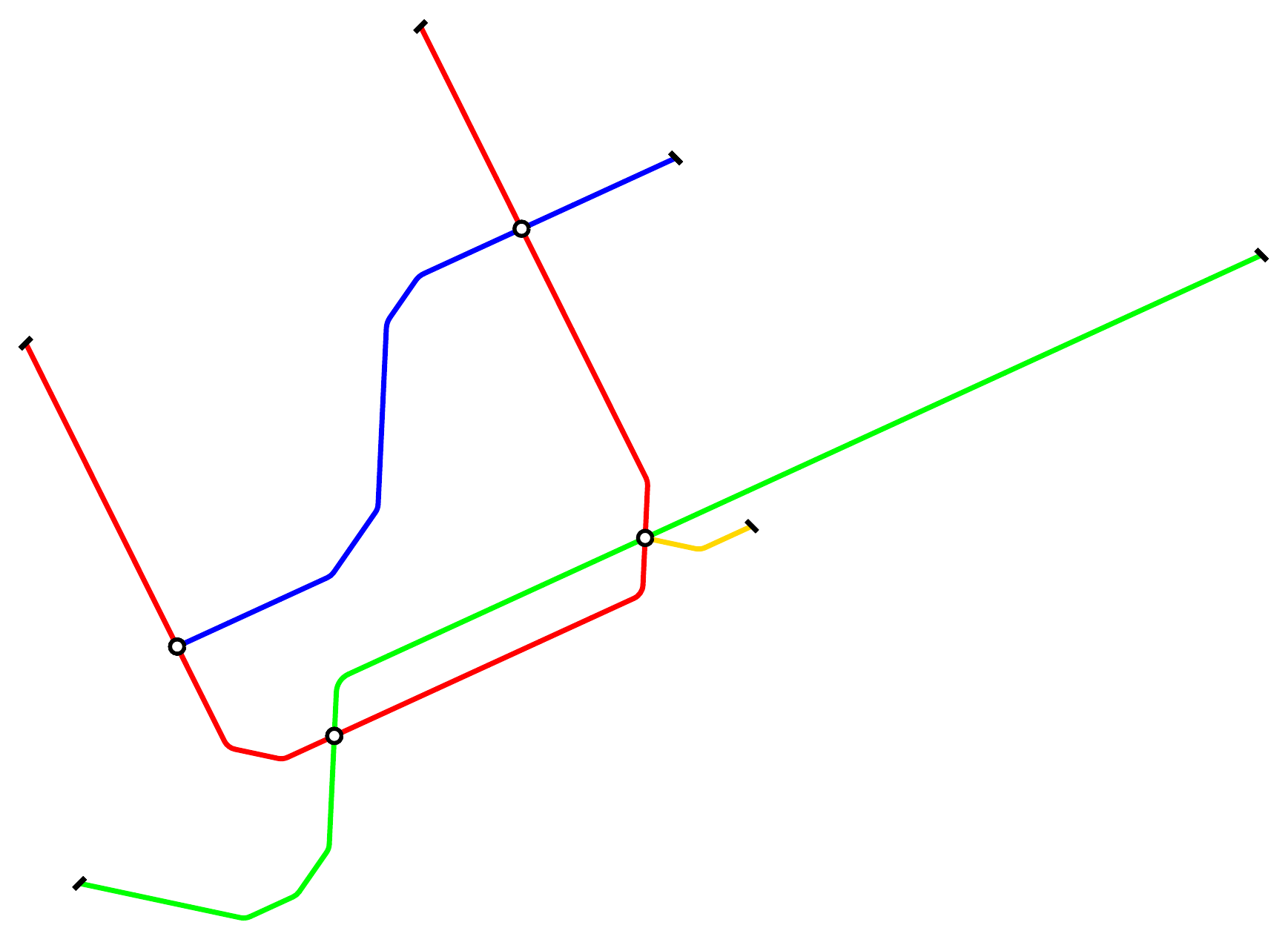} \\
	(g) 5-A & (h) 5-R & (i) 5-I \\[6pt]
\end{tabular}
\caption{Examples of Montreal generated with objective function weights $(f_1, f_2, f_3) = (3, 2, 1)$. Rows are $k = 3, k=4, k=5$ from top to bottom, columns are aligned ($k$-A), regular ($k$-R) and irregular ($k$-I) orientation system from left to right.}\label{fig:ap_montreal321}
\end{figure}
\begin{figure}[b!]
\centering
	\begin{tabular}{ccc}
		&\fbox{\includegraphics[scale=.25]{pictures/metros/input/montreal_input.pdf}}&\\
		\includegraphics[scale=.25]{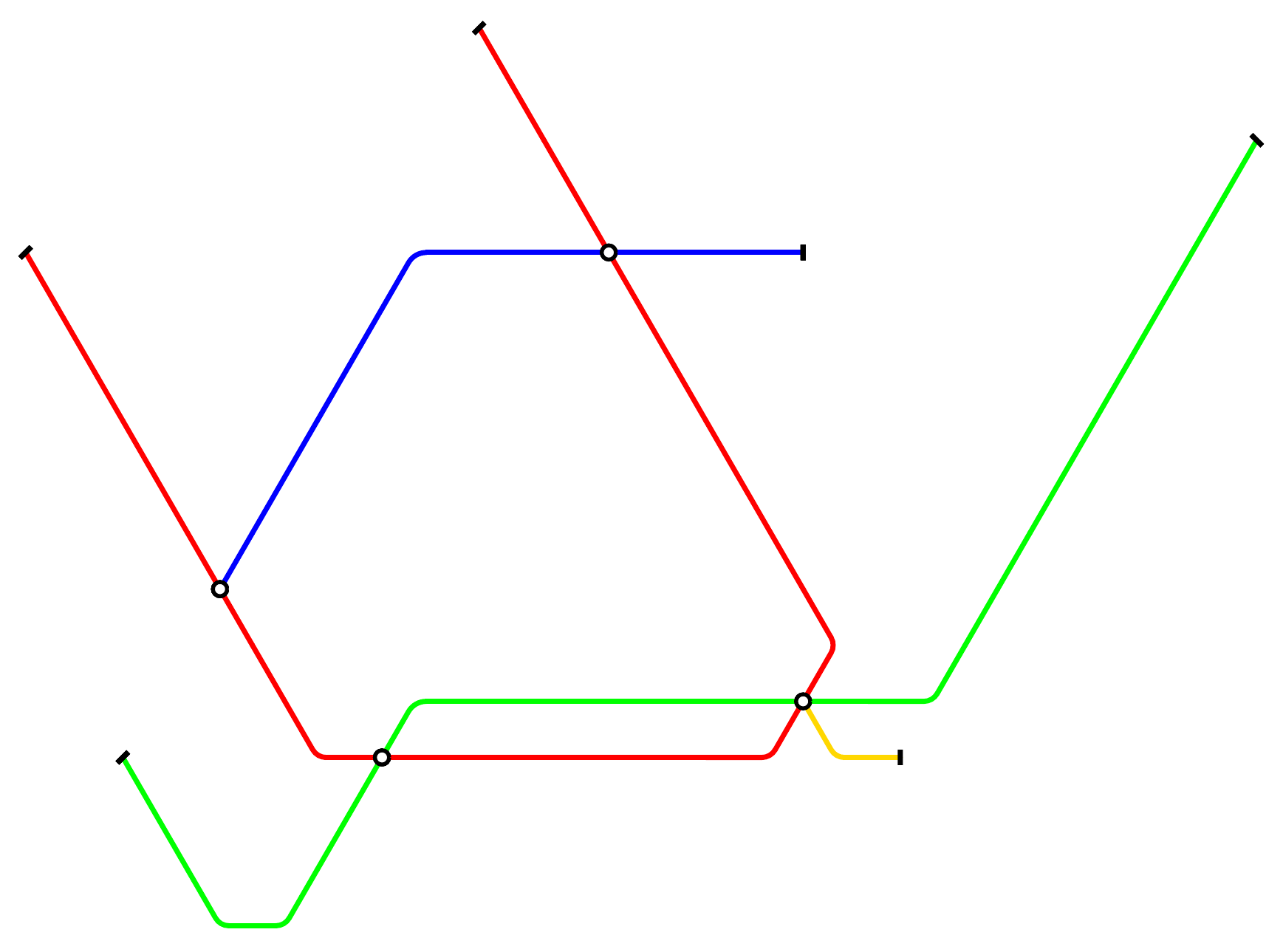} &   \includegraphics[scale=.25]{experiments/DIAGRAMS20-321/montreal-3-CO.pdf} & \includegraphics[scale=.25]{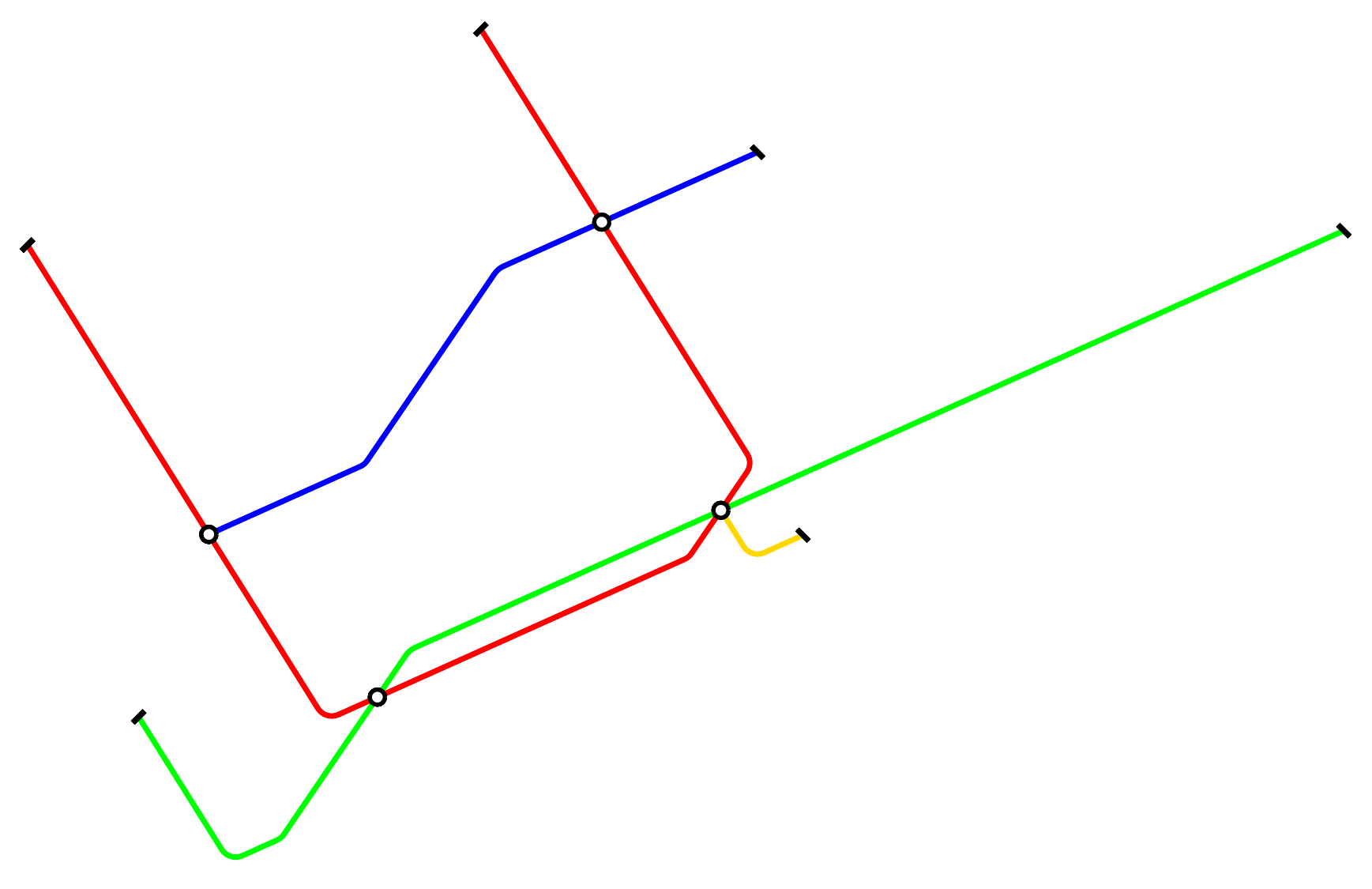} \\
		(a) 3-A & (b) 3-R & (c) 3-I \\[6pt]
		\includegraphics[scale=.25]{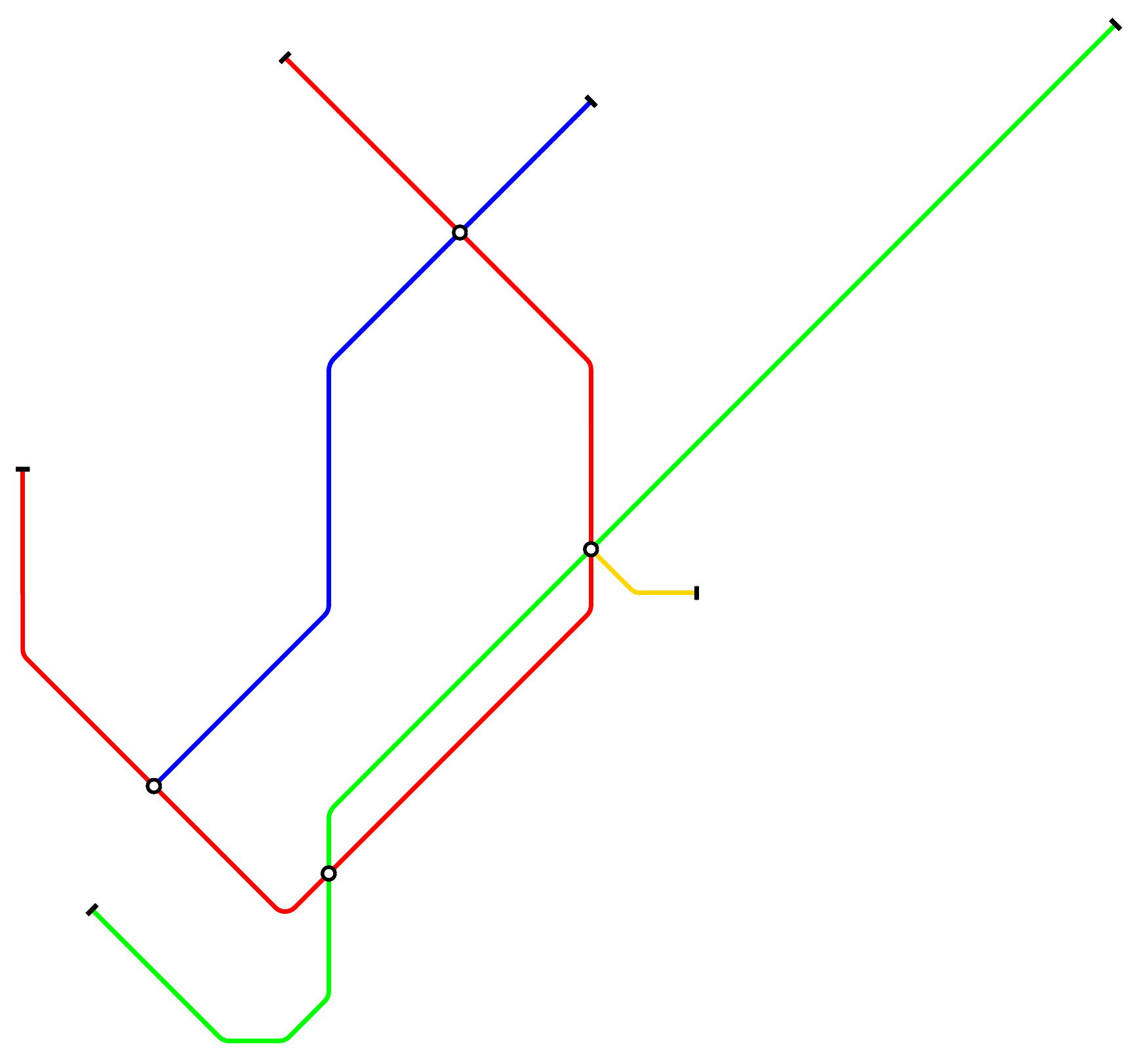} &   \includegraphics[scale=.25]{experiments/DIAGRAMS20-321/montreal-4-CO.pdf} & \includegraphics[scale=.25]{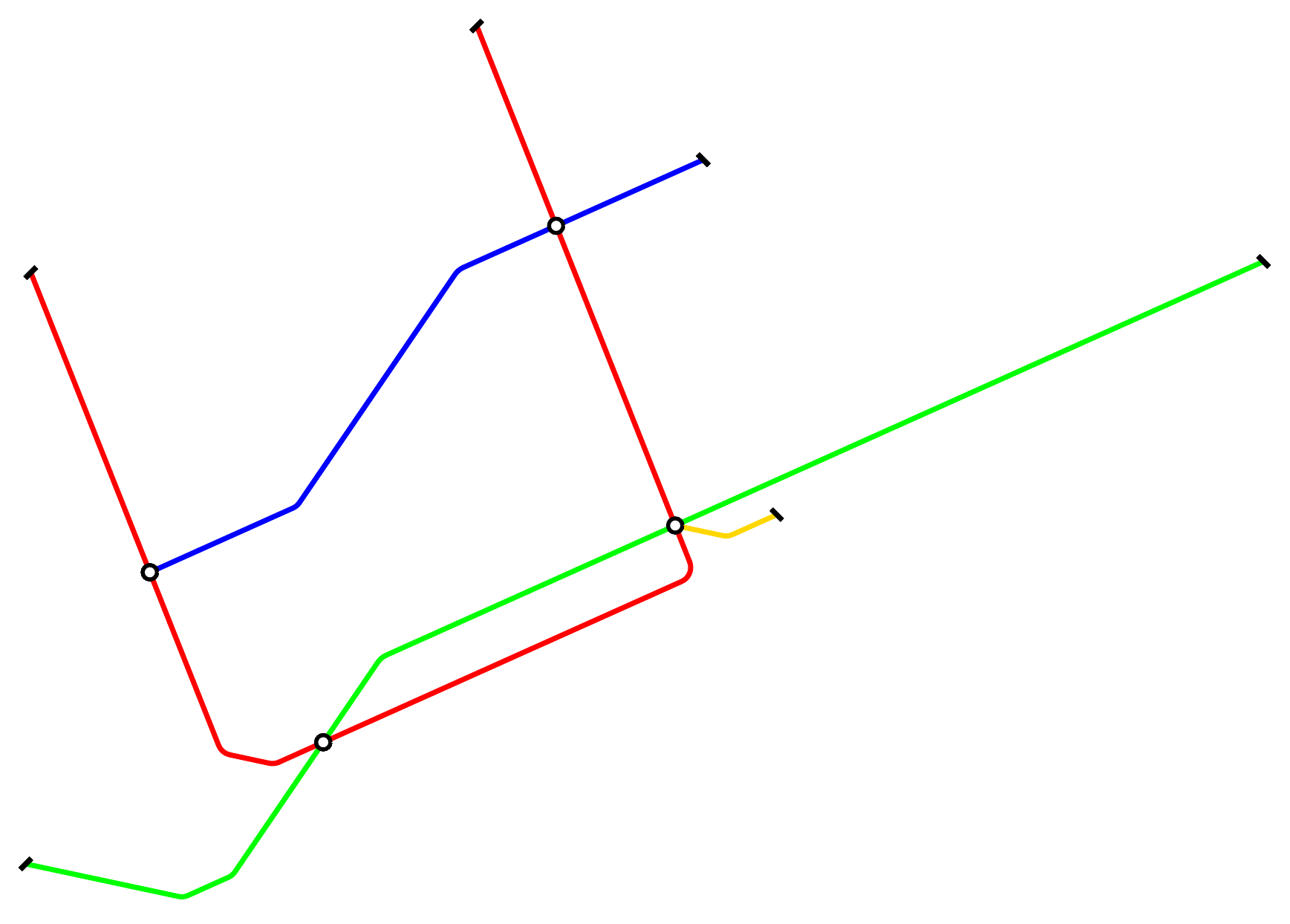} \\
		(d) 4-A & (e) 4-R & (f) 4-I \\[6pt]
		\includegraphics[scale=.25]{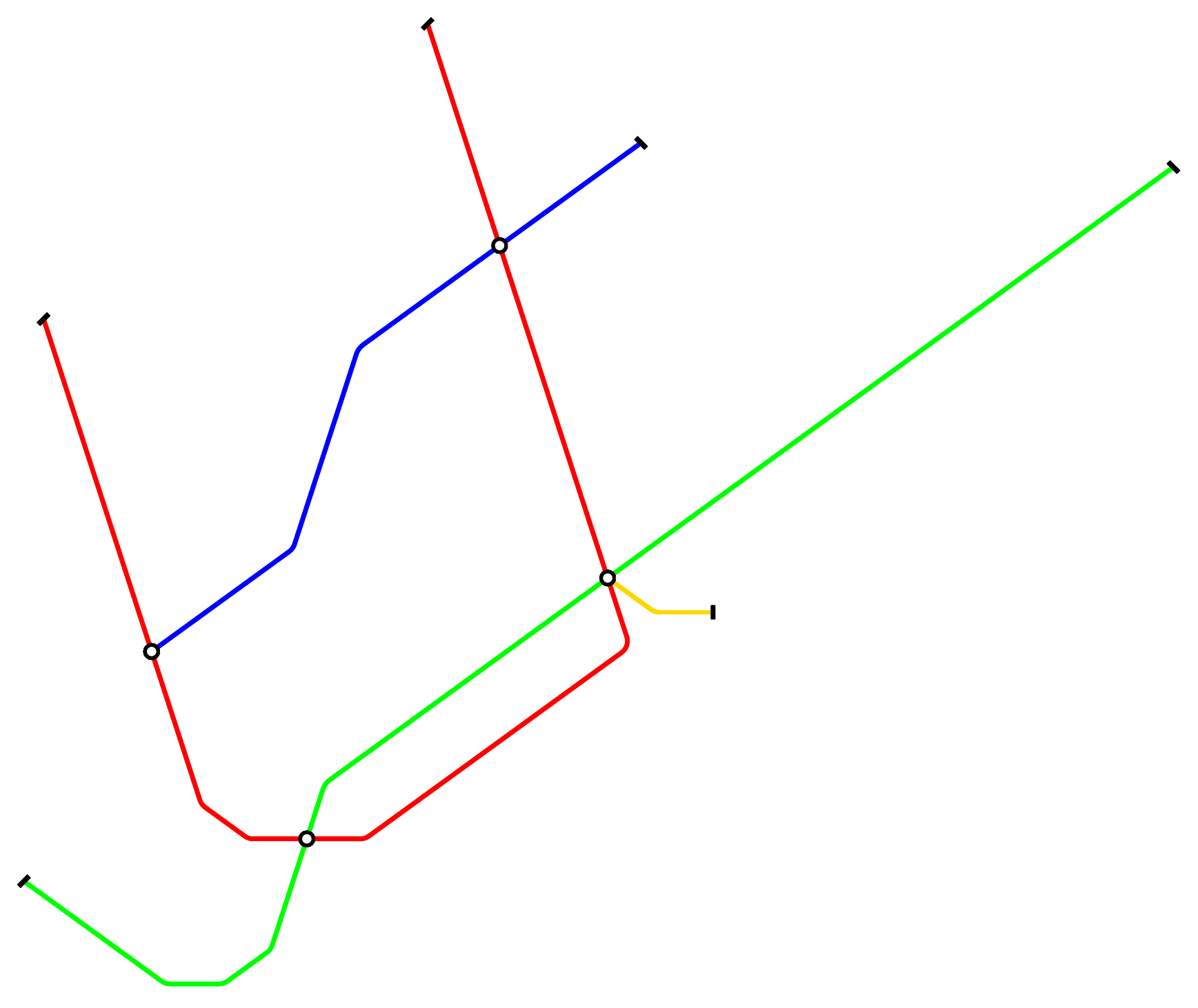} &   \includegraphics[scale=.25]{experiments/DIAGRAMS20-321/montreal-5-CO.pdf} & \includegraphics[scale=.25]{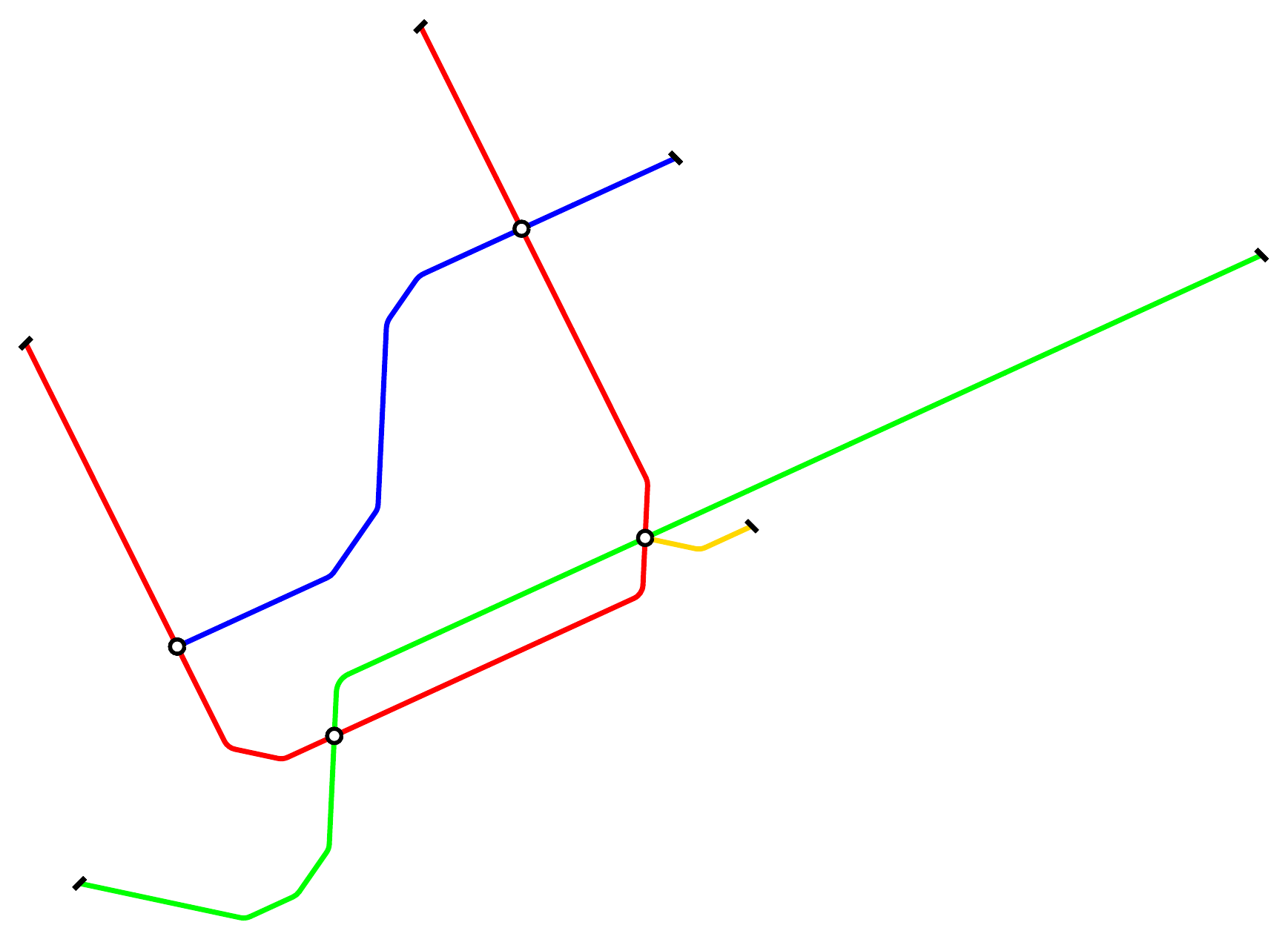} \\
		(g) 5-A & (h) 5-R & (i) 5-I \\[6pt]
	\end{tabular}
	\caption{Examples of Montreal generated with objective function weights $(f_1, f_2, f_3) = (10, 5, 1)$. Rows are $k = 3, k=4, k=5$ from top to bottom, columns are aligned ($k$-A), regular ($k$-R) and irregular ($k$-I) orientation system from left to right.}\label{fig:ap_montreal1051}
\end{figure}

\begin{figure}[b!]
	\centering
	\begin{tabular}{ccc}
		&\fbox{\includegraphics[scale=.25]{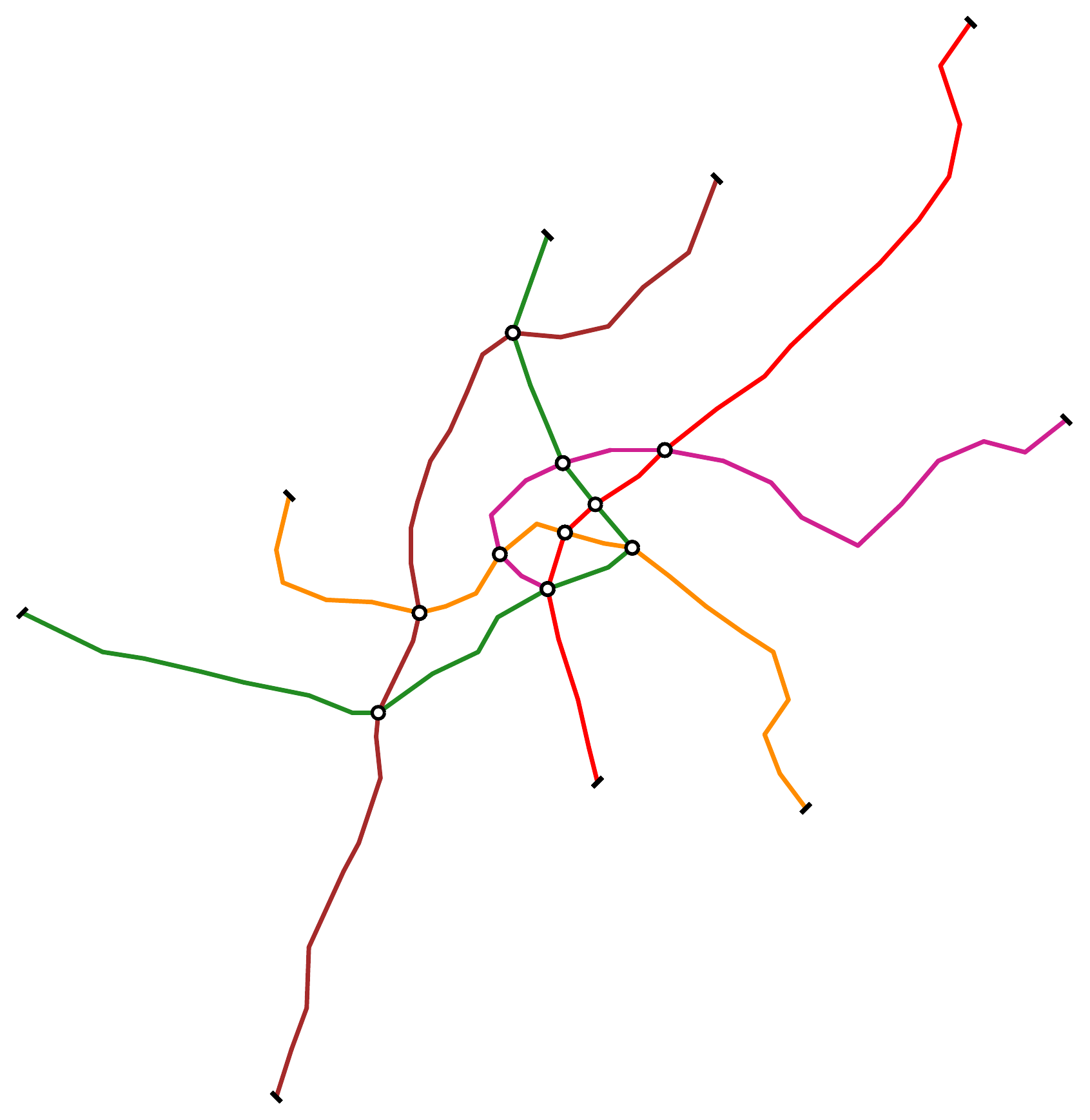}}&\\
		\includegraphics[scale=.25]{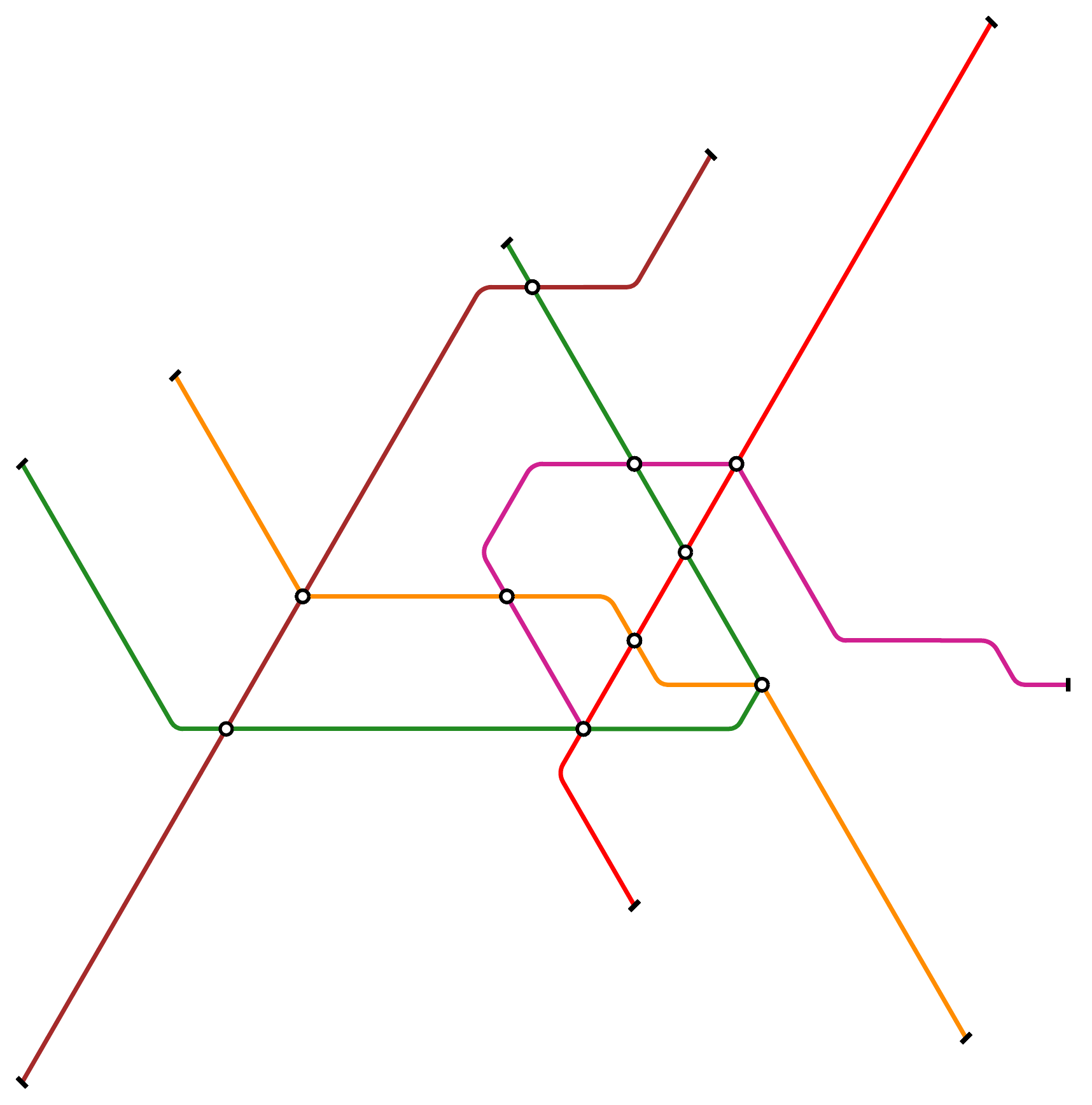} &   \includegraphics[scale=.25]{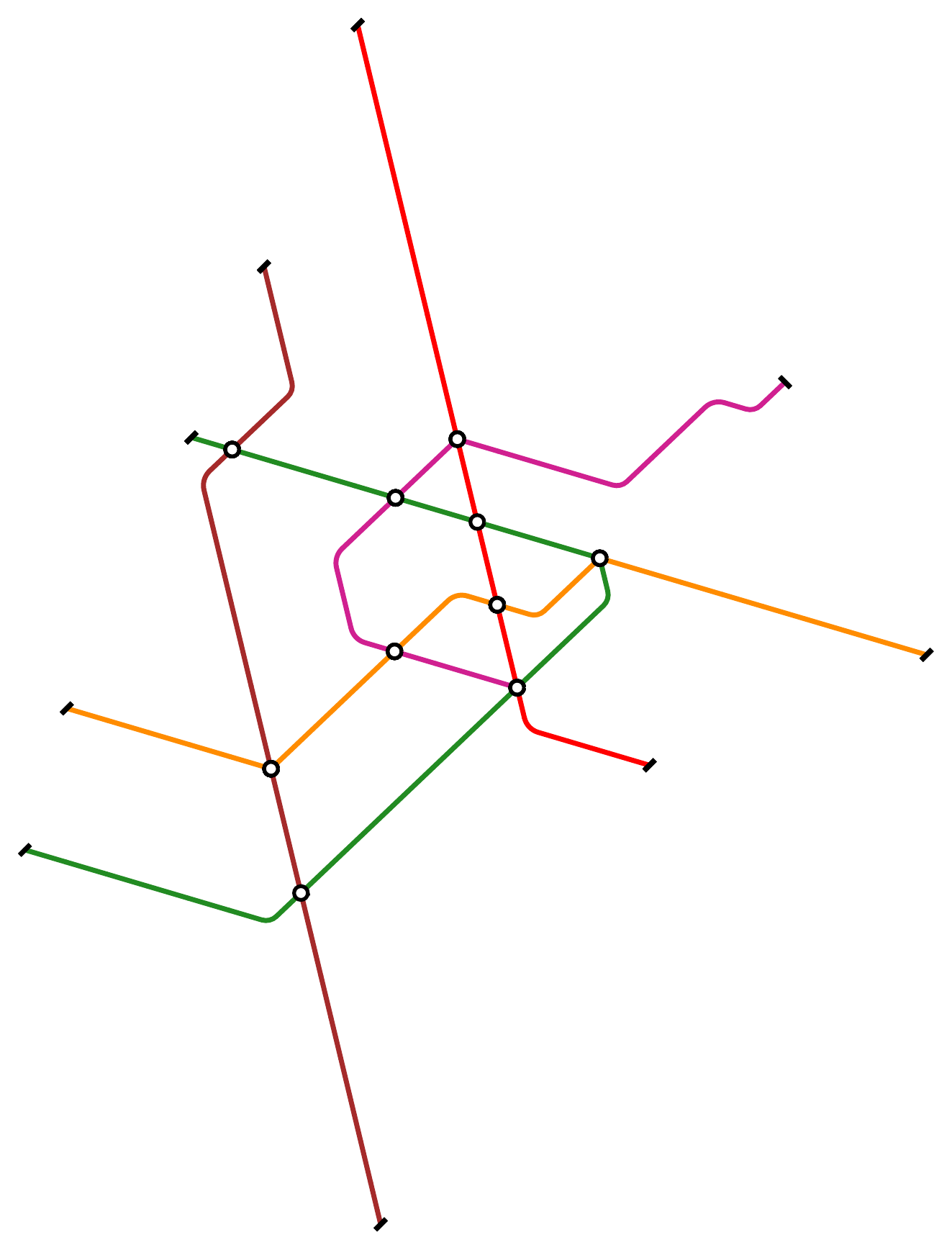} & \includegraphics[scale=.25]{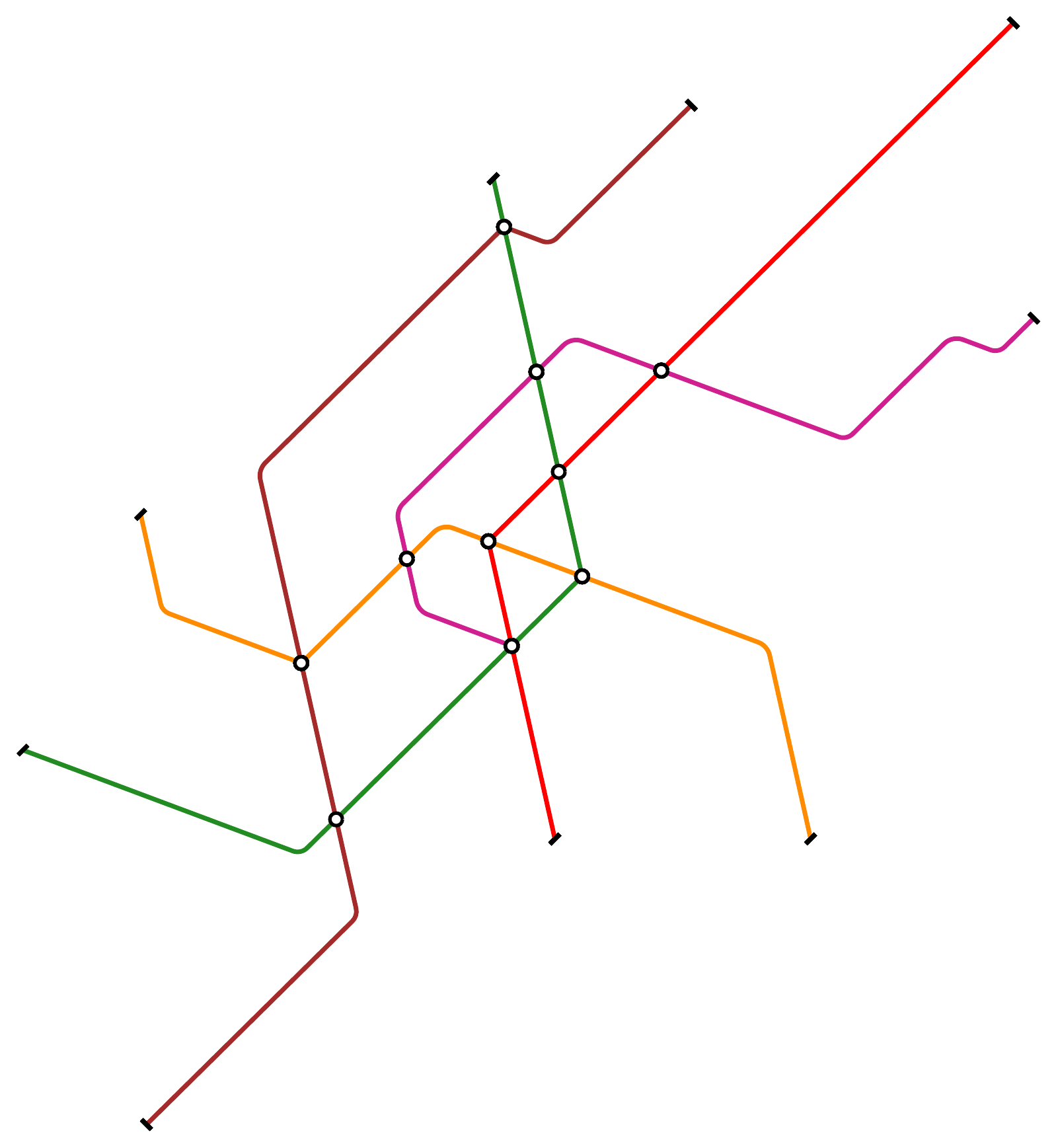} \\
		(a) 3-A & (b) 3-R & (c) 3-I \\[6pt]
		\includegraphics[scale=.25]{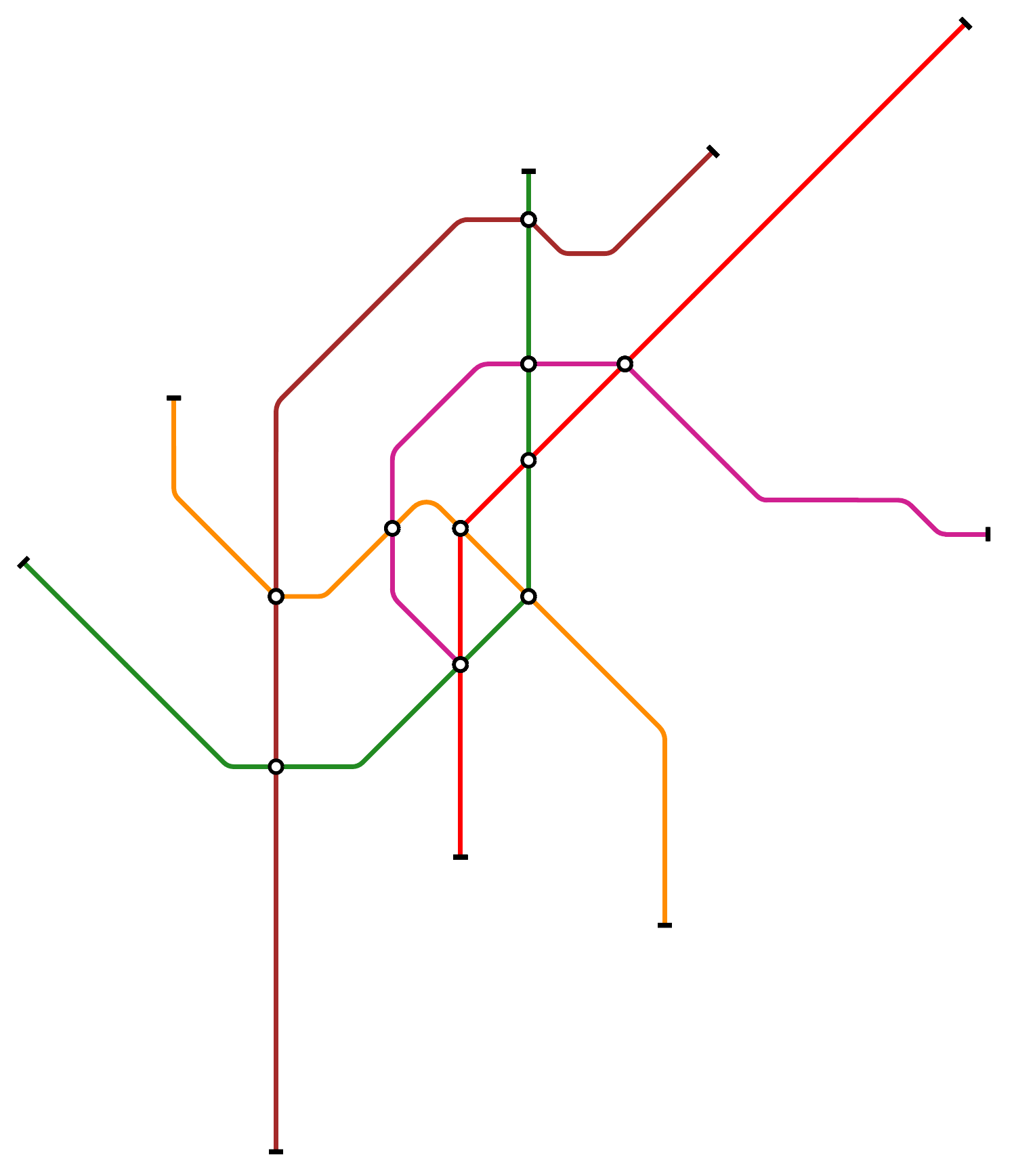} &   \includegraphics[scale=.25]{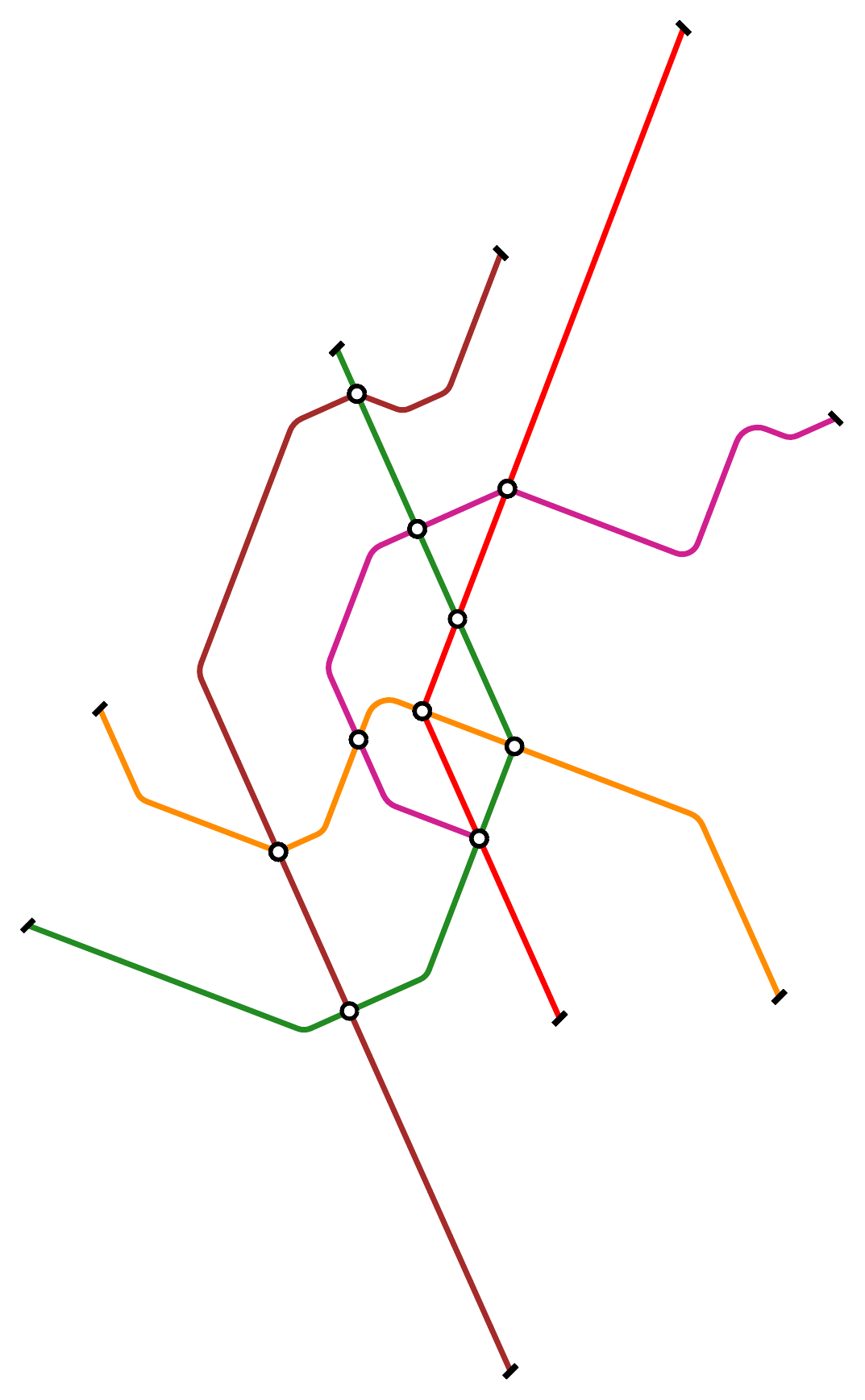} & \includegraphics[scale=.25]{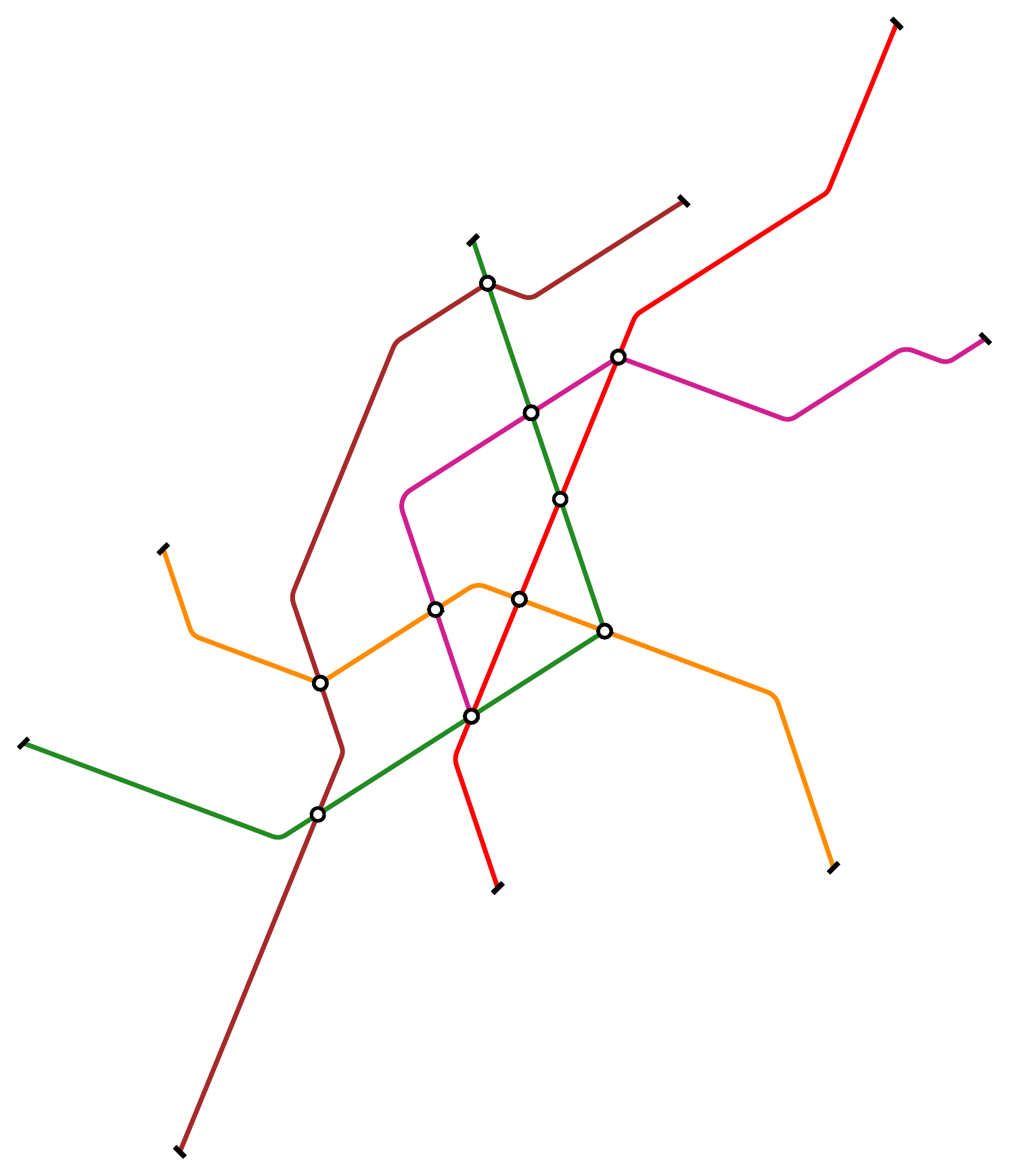} \\
		(d) 4-A & (e) 4-R & (f) 4-I \\[6pt]
		\includegraphics[scale=.25]{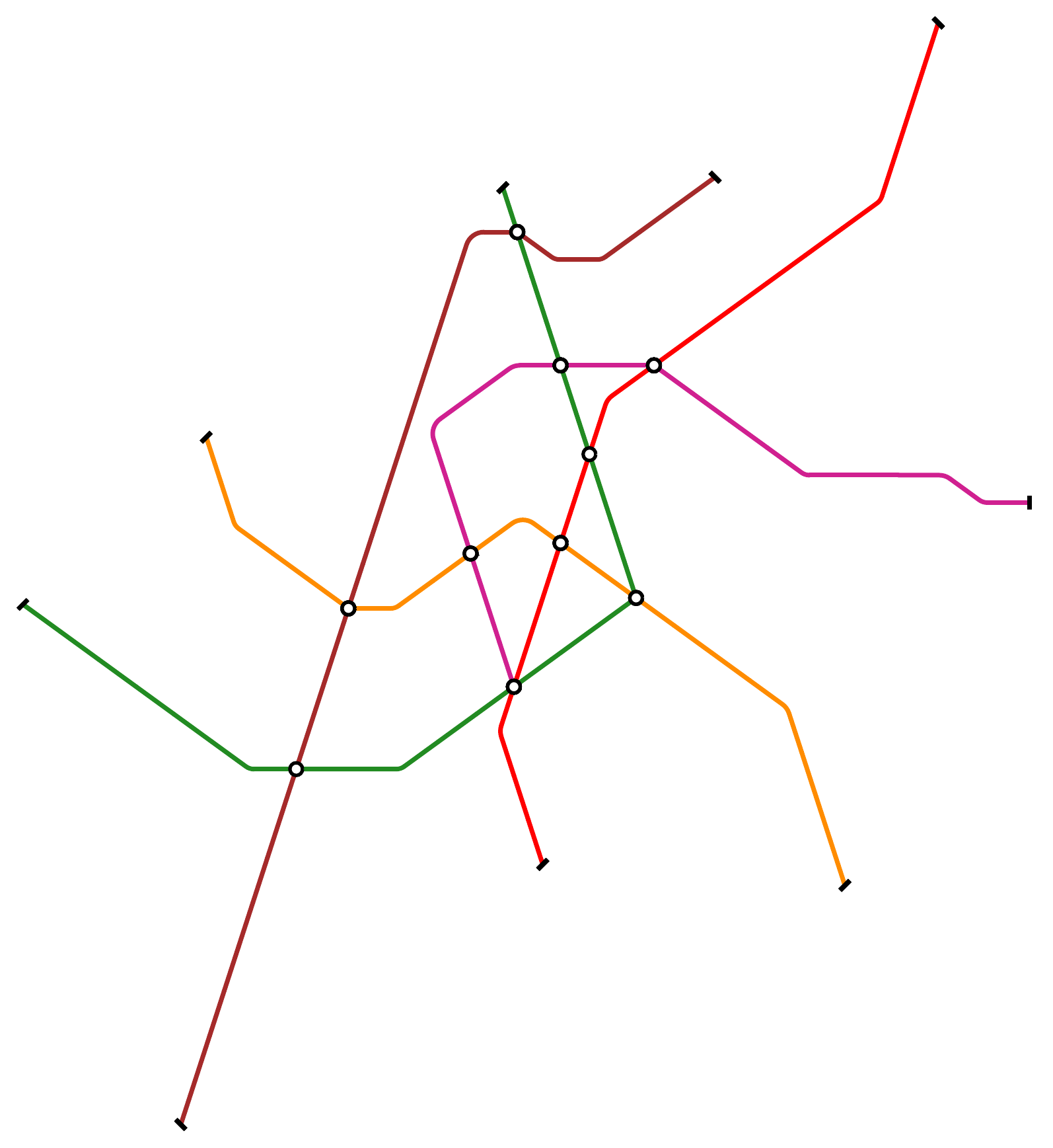} &   \includegraphics[scale=.25]{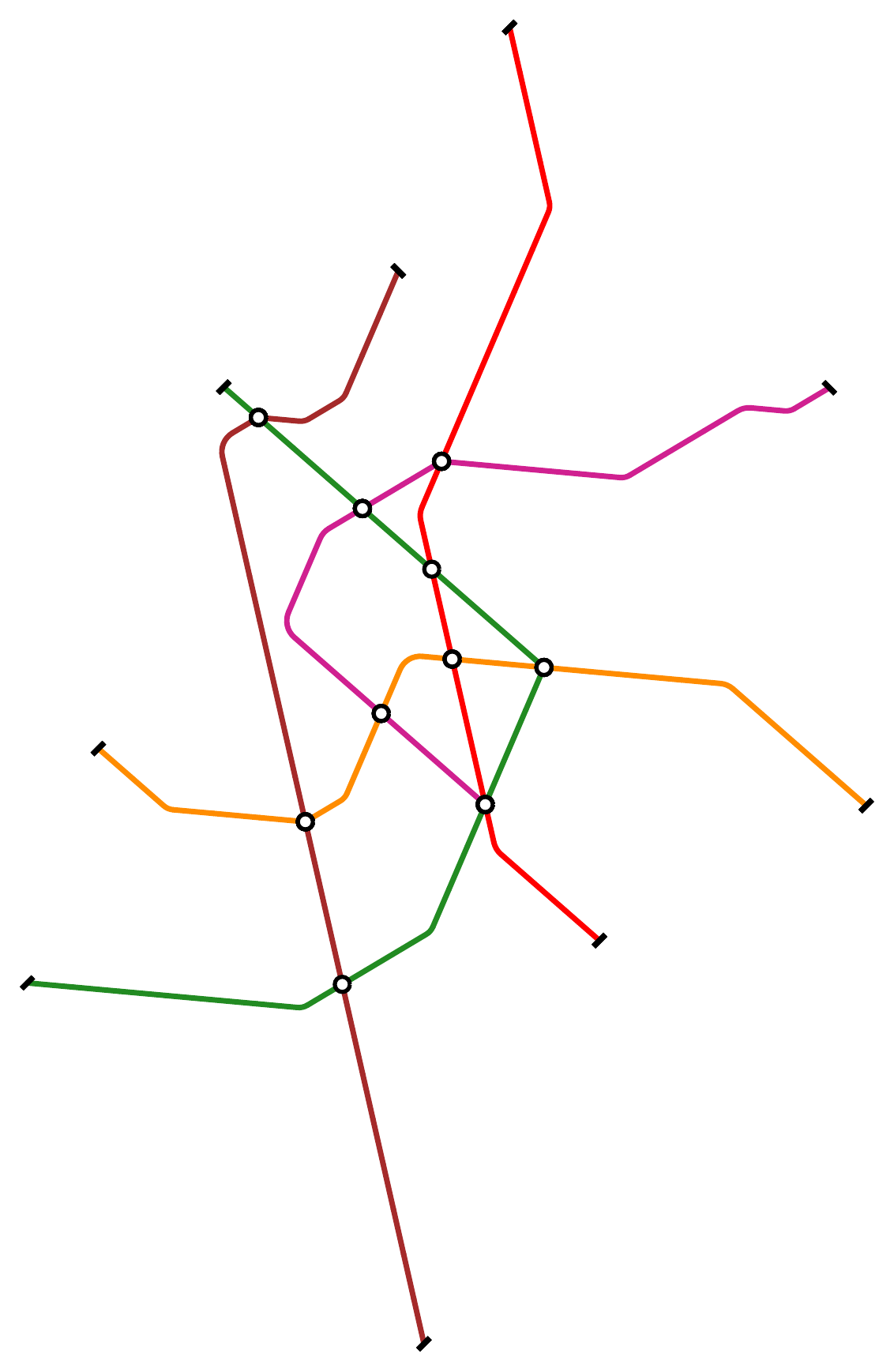} & \includegraphics[scale=.25]{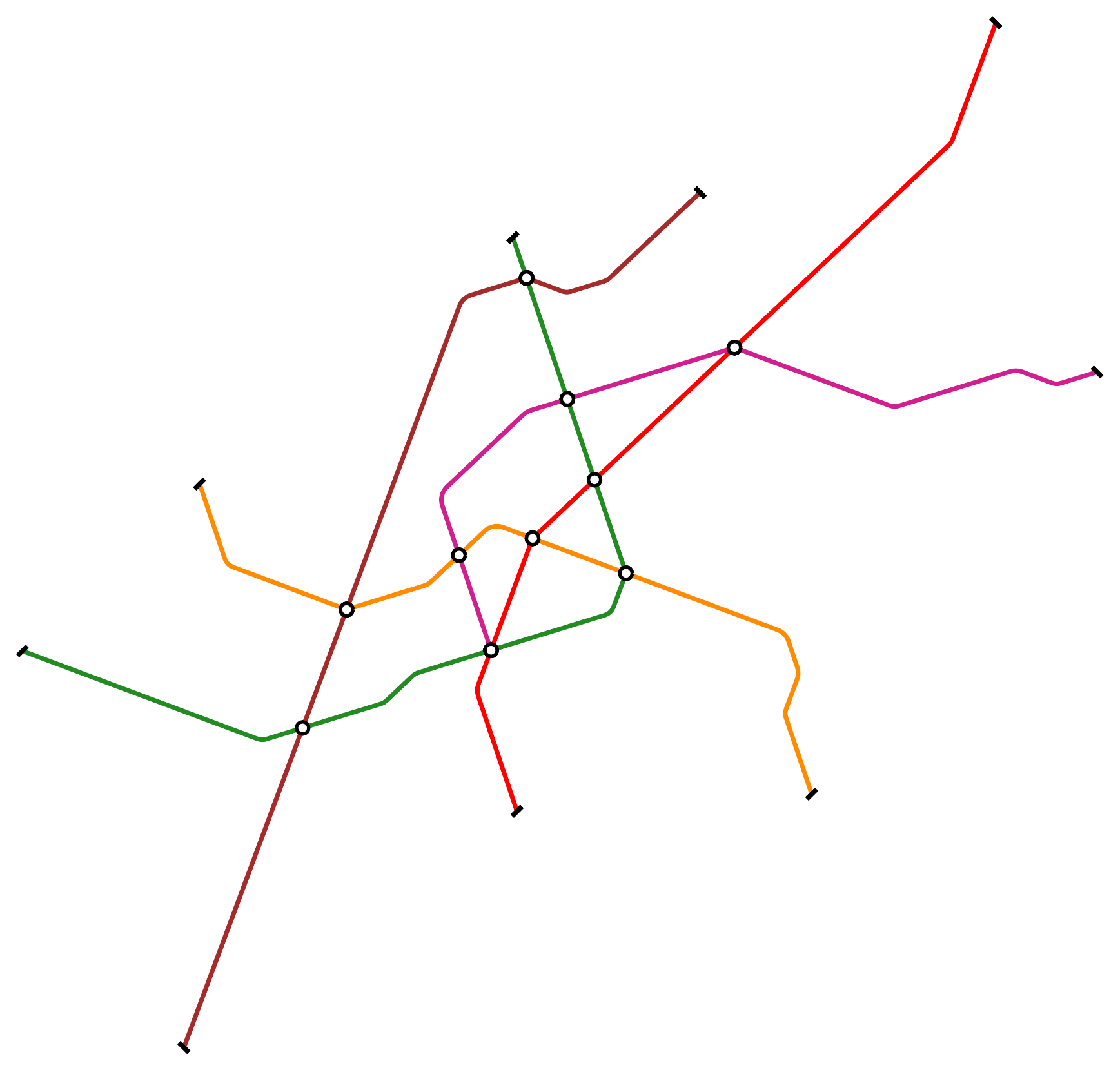} \\
		(g) 5-A & (h) 5-R & (i) 5-I \\[6pt]
	\end{tabular}
	\caption{Examples of Vienna generated with objective function weights $(f_1, f_2, f_3) = (3, 2, 1)$. Rows are $k = 3, k=4, k=5$ from top to bottom, columns are aligned ($k$-A), regular ($k$-R) and irregular ($k$-I) orientation system from left to right.}\label{fig:ap_vienna321}
\end{figure}
\begin{figure}[b!]
	\centering
	\begin{tabular}{ccc}
		&\fbox{\includegraphics[scale=.25]{pictures/metros/input/wien_input.pdf}}&\\
		\includegraphics[scale=.25]{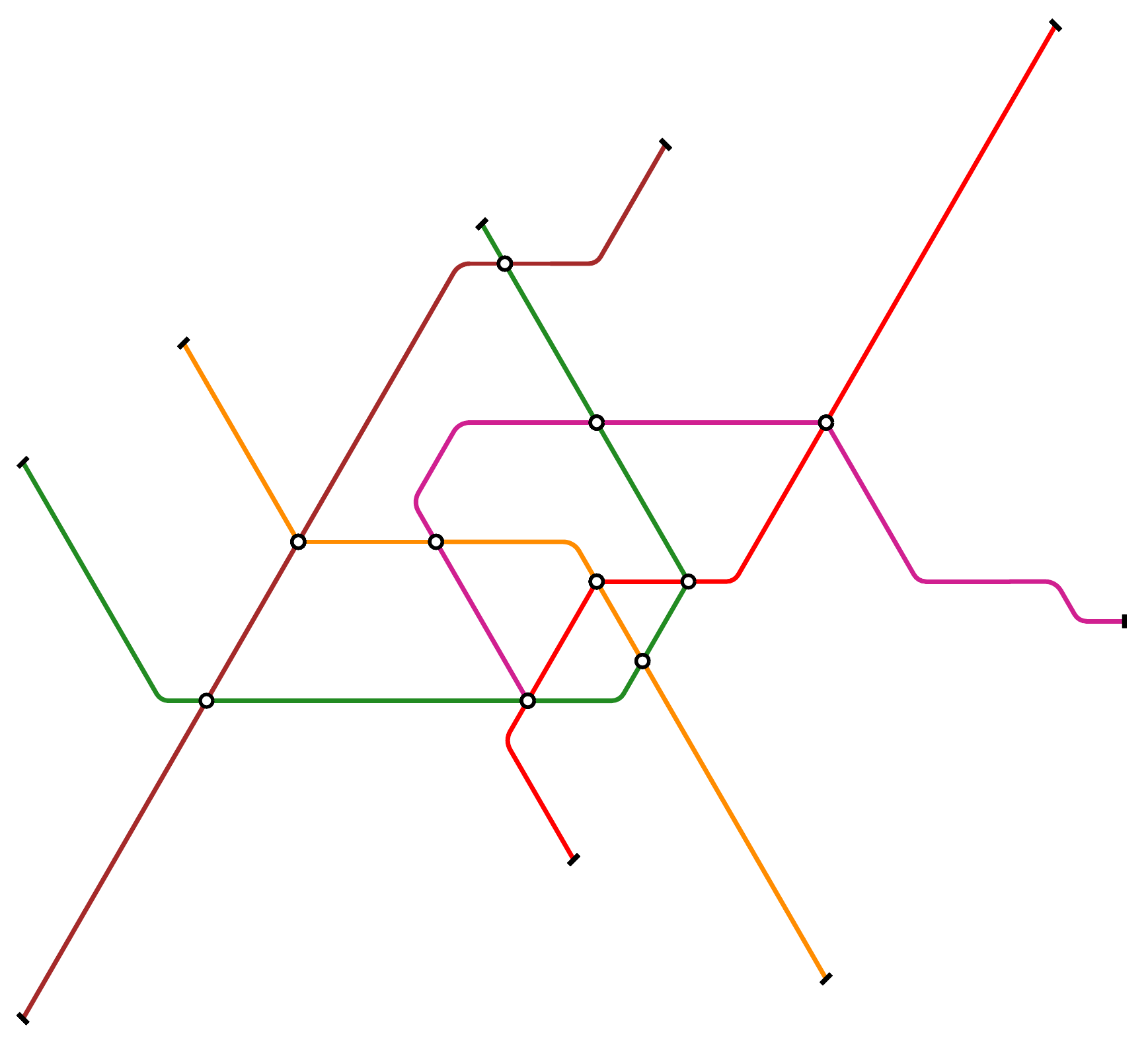} &   \includegraphics[scale=.25]{experiments/DIAGRAMS20-321/wien-3-CO.pdf} & \includegraphics[scale=.25]{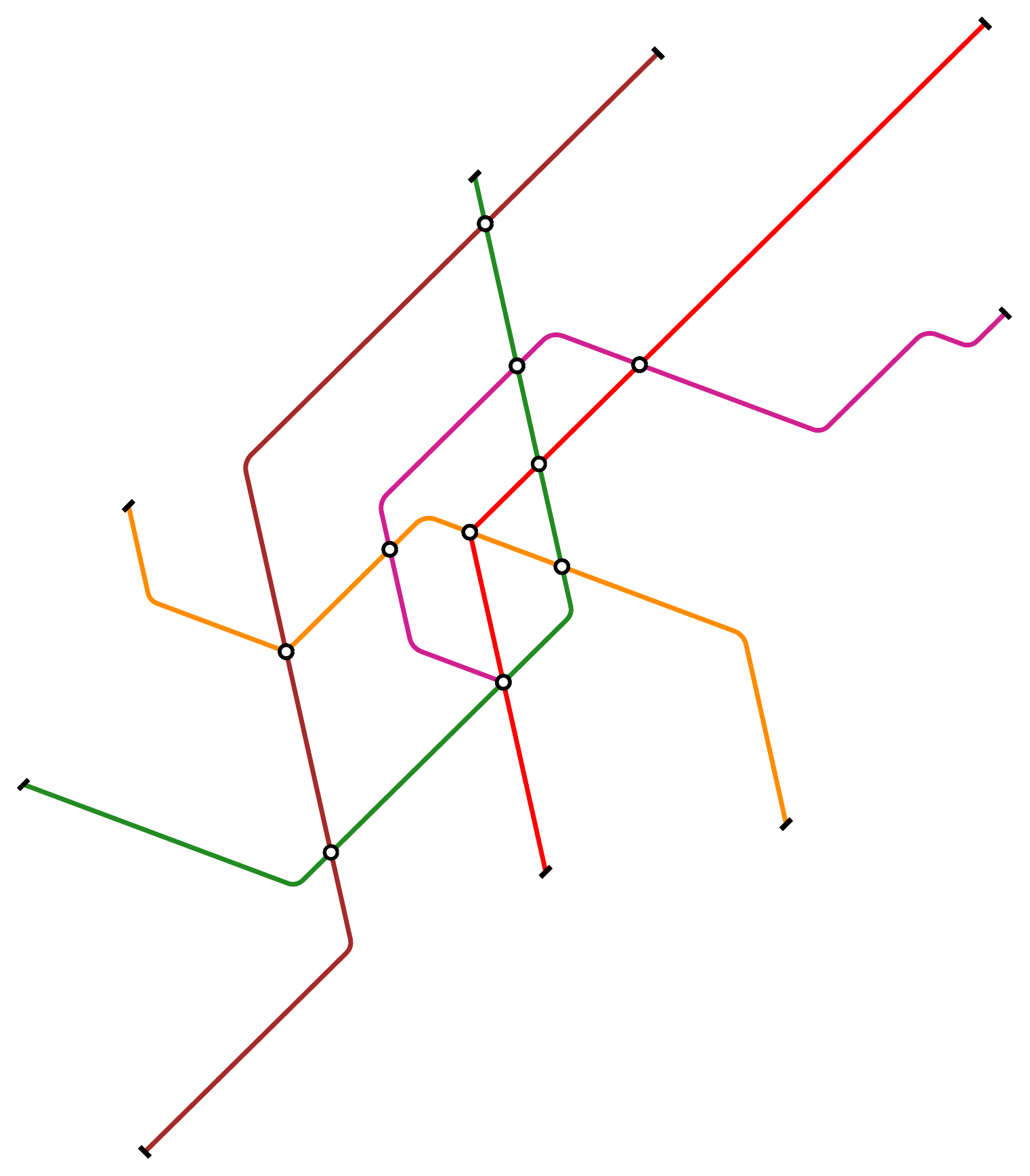} \\
		(a) 3-A & (b) 3-R & (c) 3-I \\[6pt]
		\includegraphics[scale=.25]{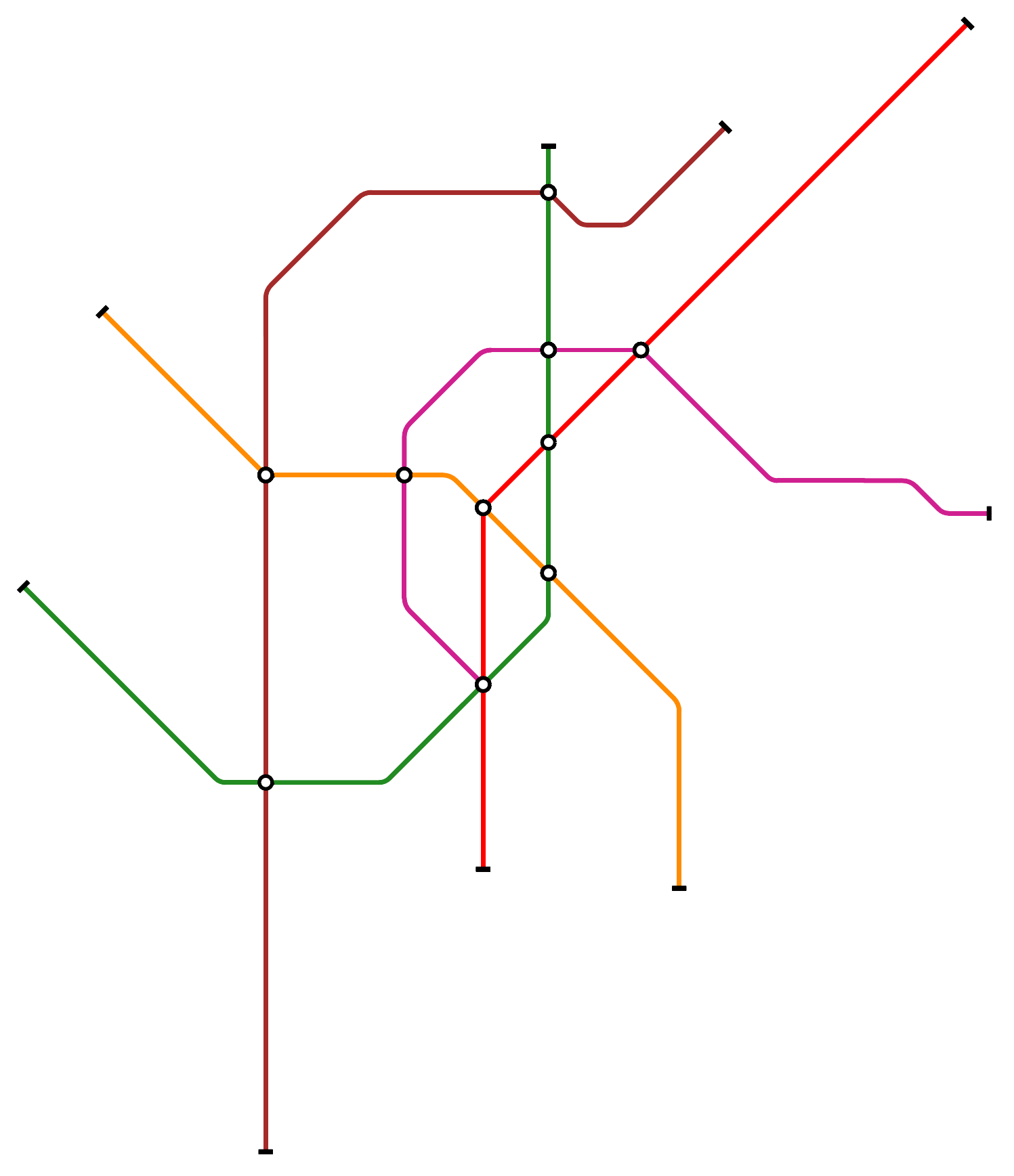} &   \includegraphics[scale=.25]{experiments/DIAGRAMS20-321/wien-4-CO.pdf} & \includegraphics[scale=.25]{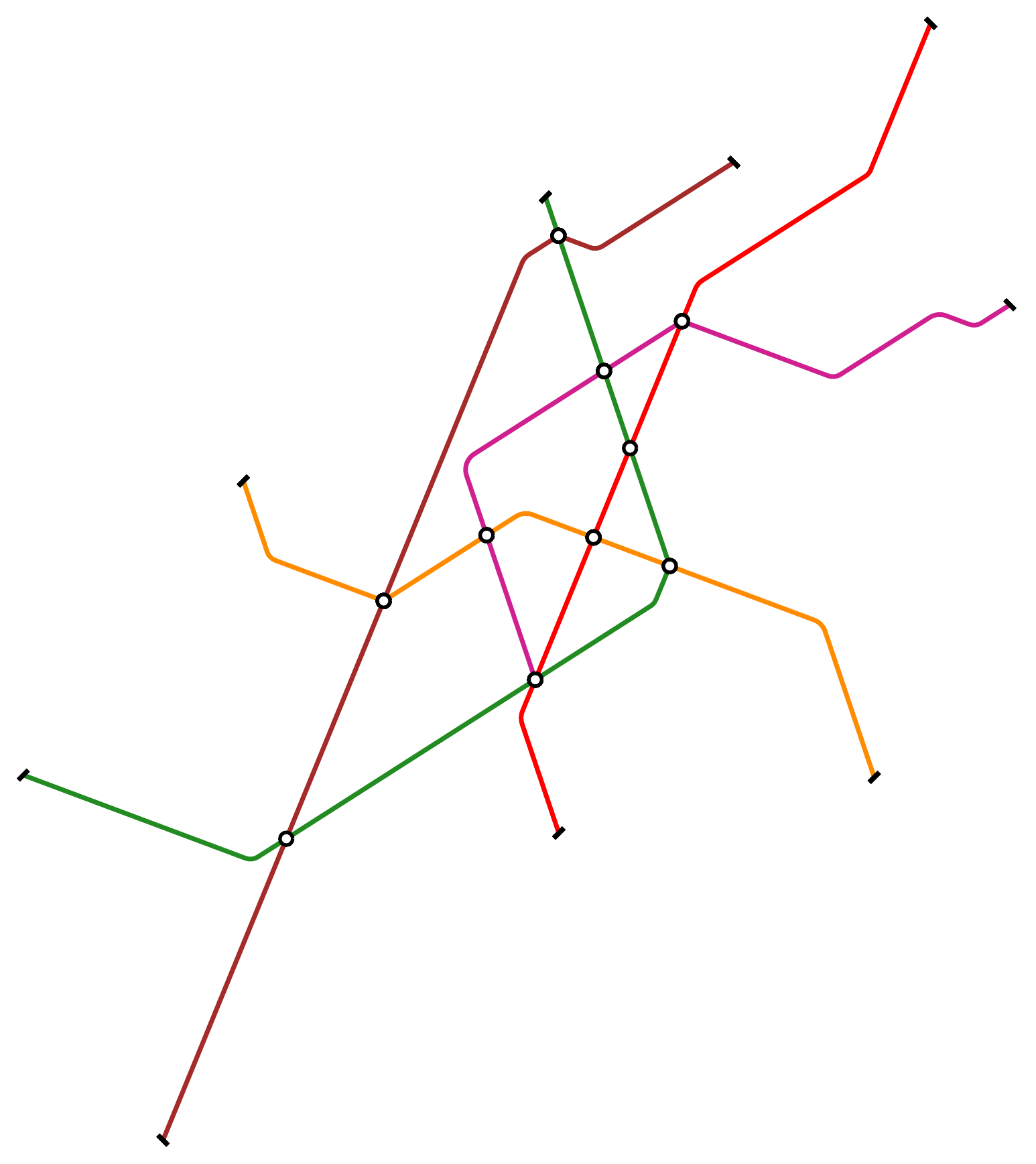} \\
		(d) 4-A & (e) 4-R & (f) 4-I \\[6pt]
		\includegraphics[scale=.25]{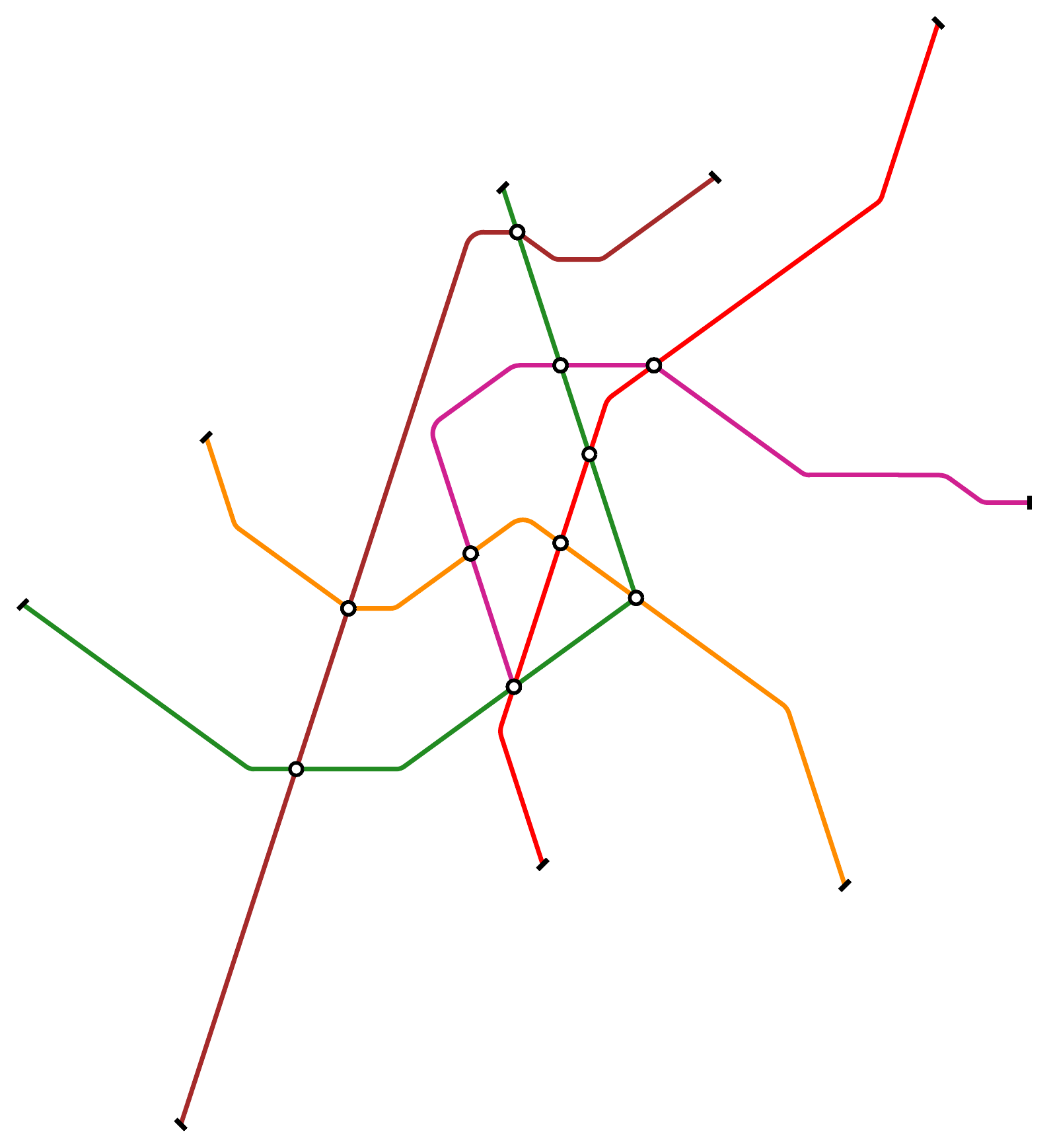} &   \includegraphics[scale=.25]{experiments/DIAGRAMS20-321/wien-5-CO.pdf} & \includegraphics[scale=.25]{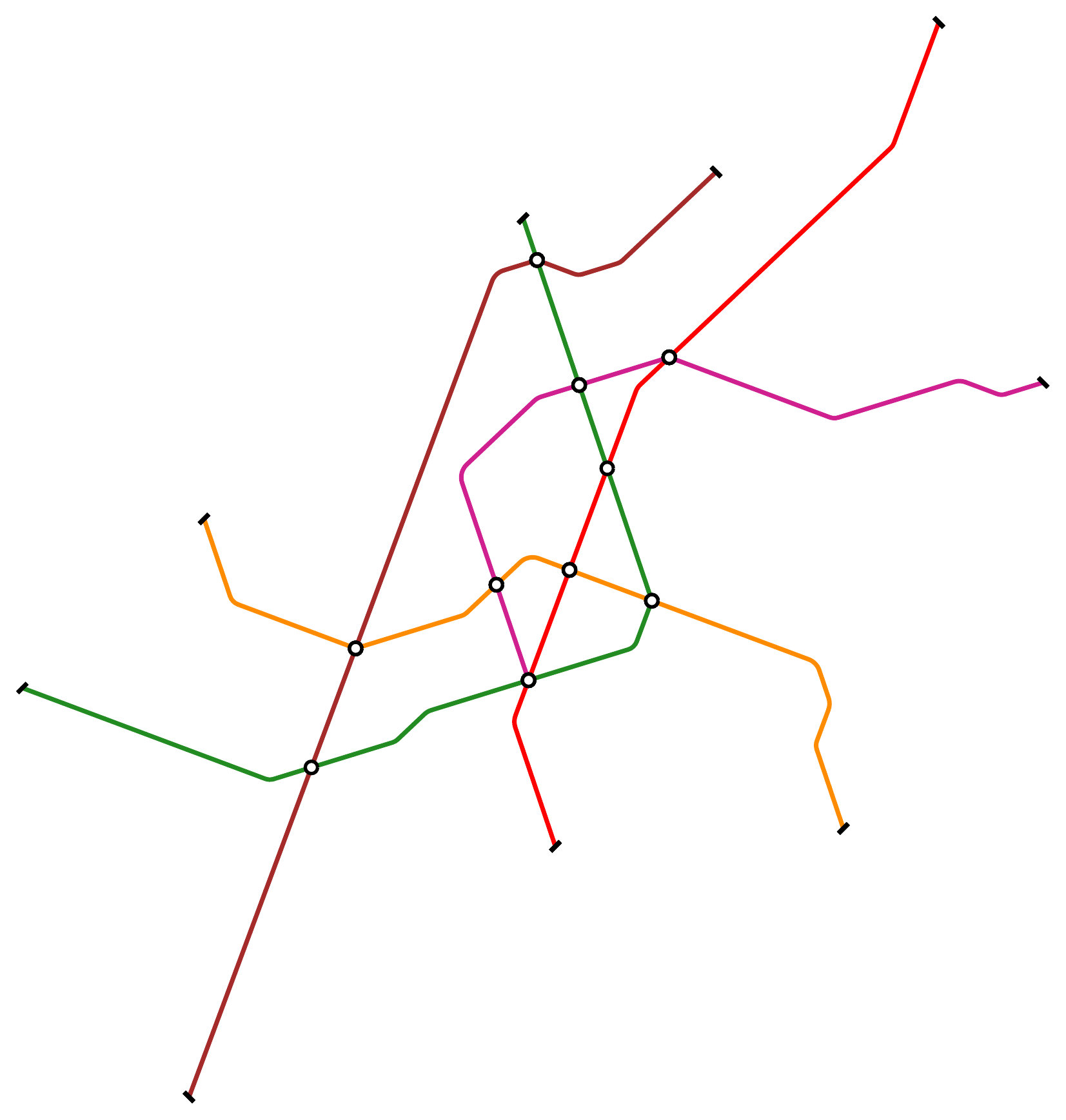} \\
		(g) 5-A & (h) 5-R & (i) 5-I \\[6pt]
	\end{tabular}
	\caption{Examples of Vienna generated with objective function weights $(f_1, f_2, f_3) = (10, 5, 1)$. Rows are $k = 3, k=4, k=5$ from top to bottom, columns are aligned ($k$-A), regular ($k$-R) and irregular ($k$-I) orientation system from left to right.}\label{fig:ap_vienna1051}
\end{figure}

\begin{figure}[b!]
\centering
	\begin{tabular}{ccc}
		&\fbox{\includegraphics[scale=.25]{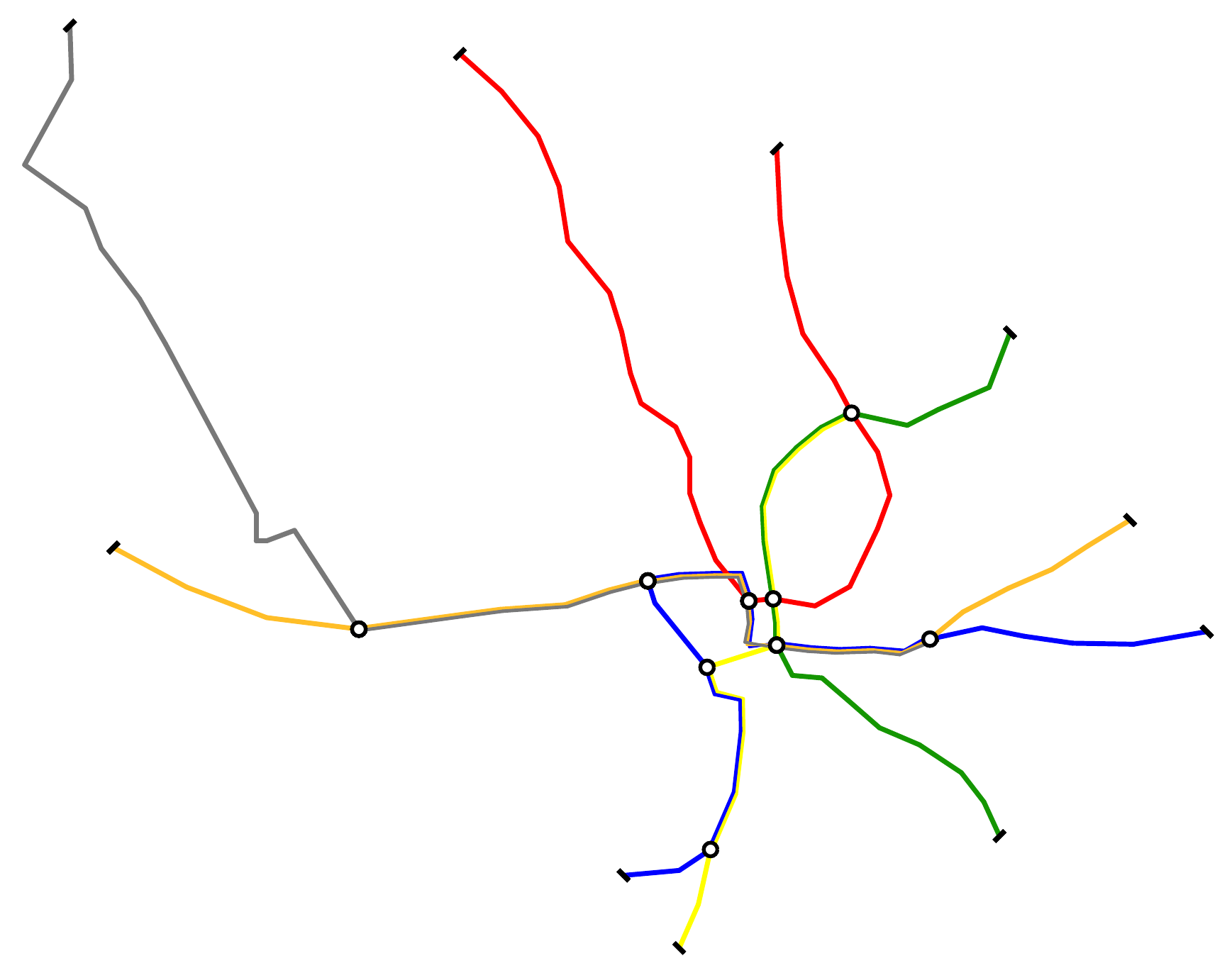}}&\\
		\includegraphics[scale=.25]{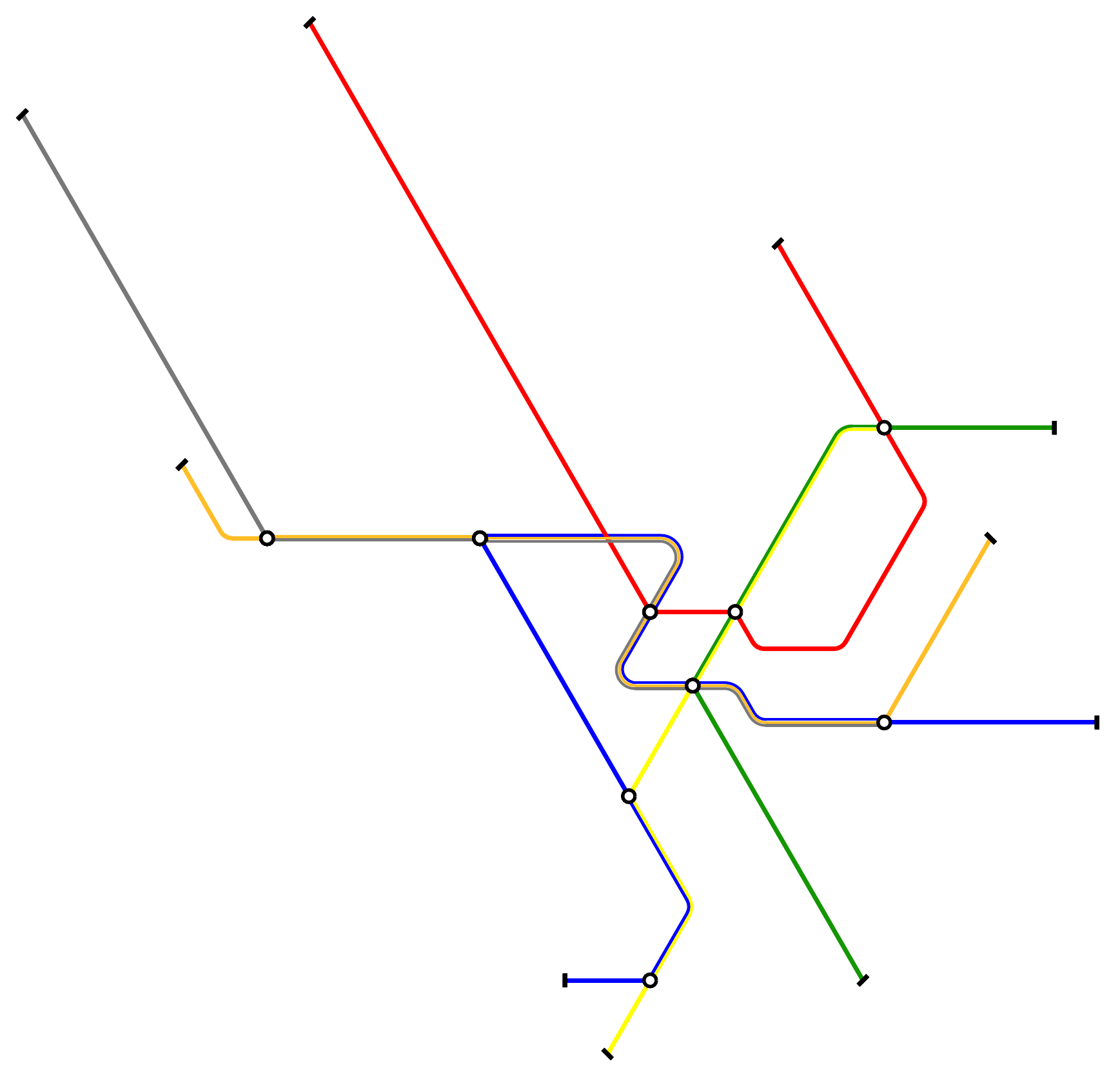} &   \includegraphics[scale=.25]{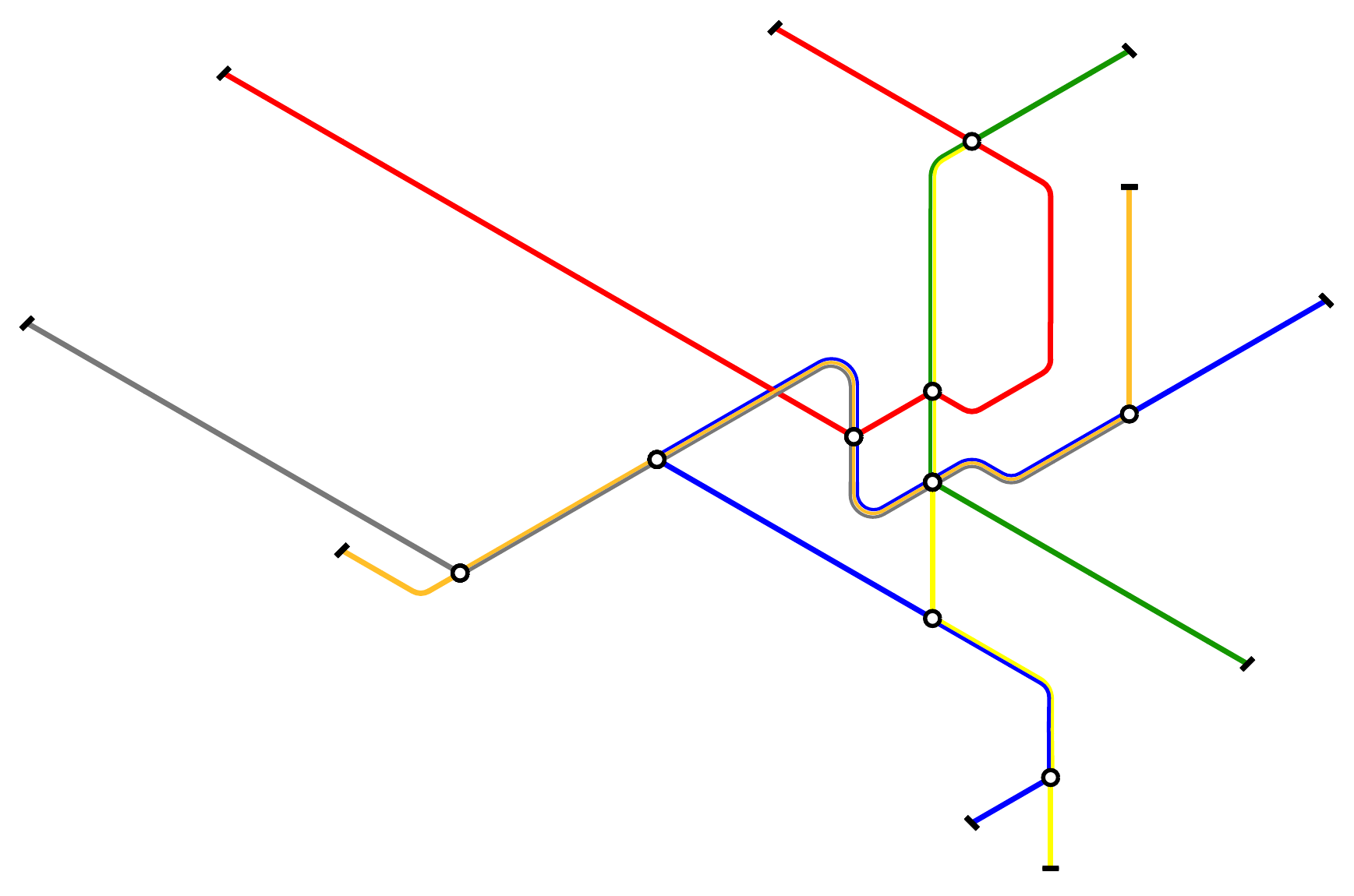} & \includegraphics[scale=.25]{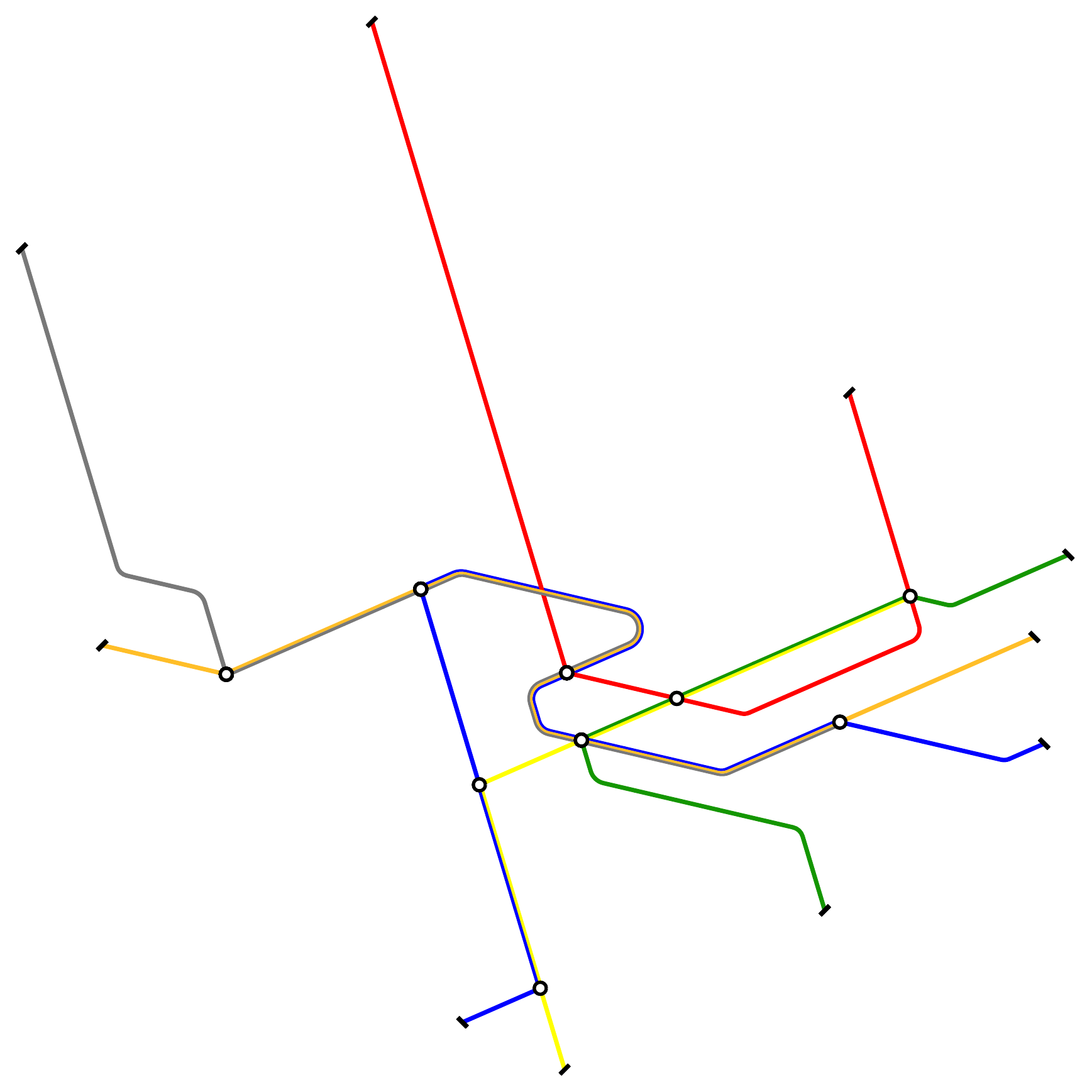} \\
		(a) 3-A & (b) 3-R & (c) 3-I \\[6pt]
		\includegraphics[scale=.25]{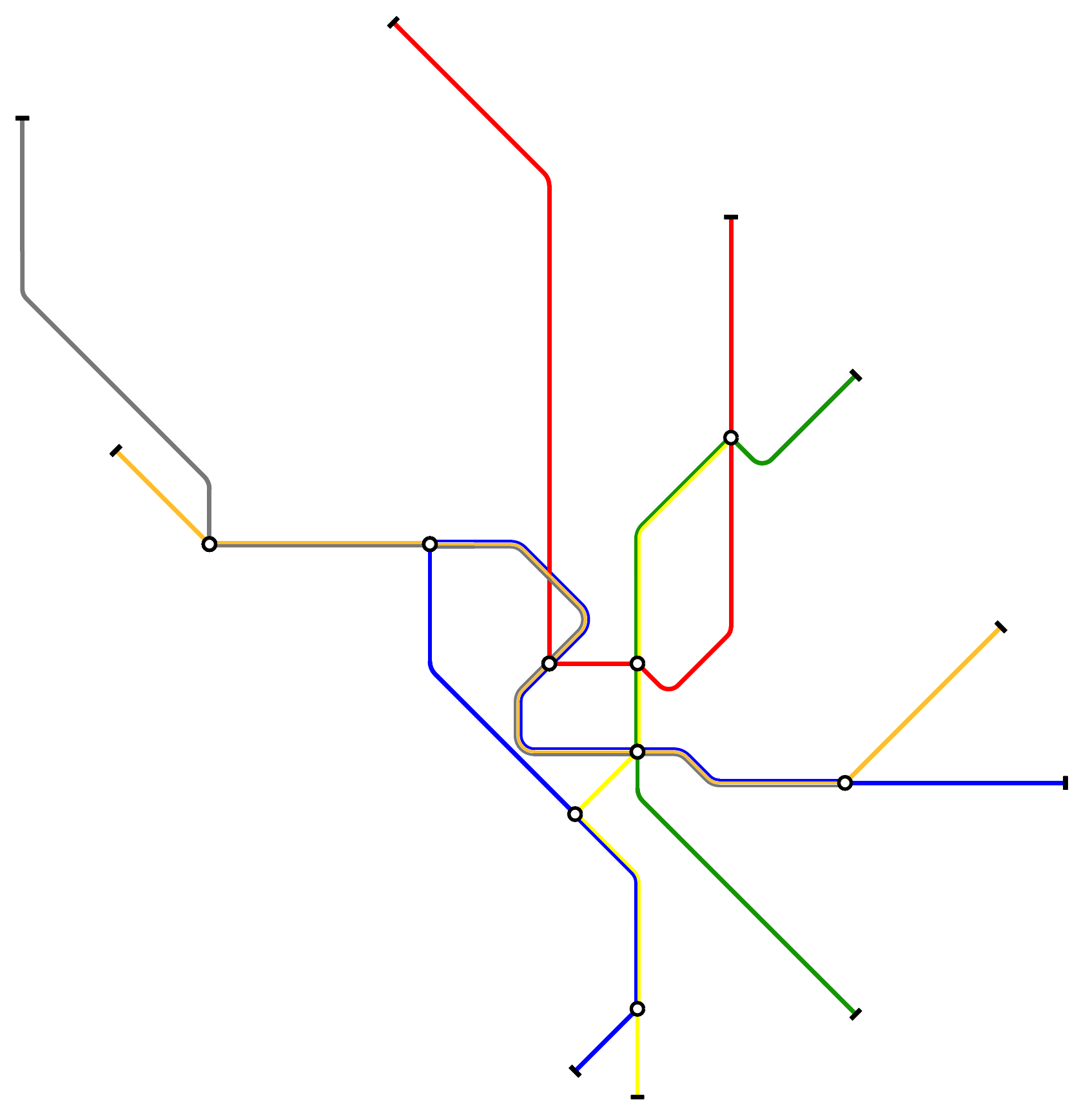} &   \includegraphics[scale=.25]{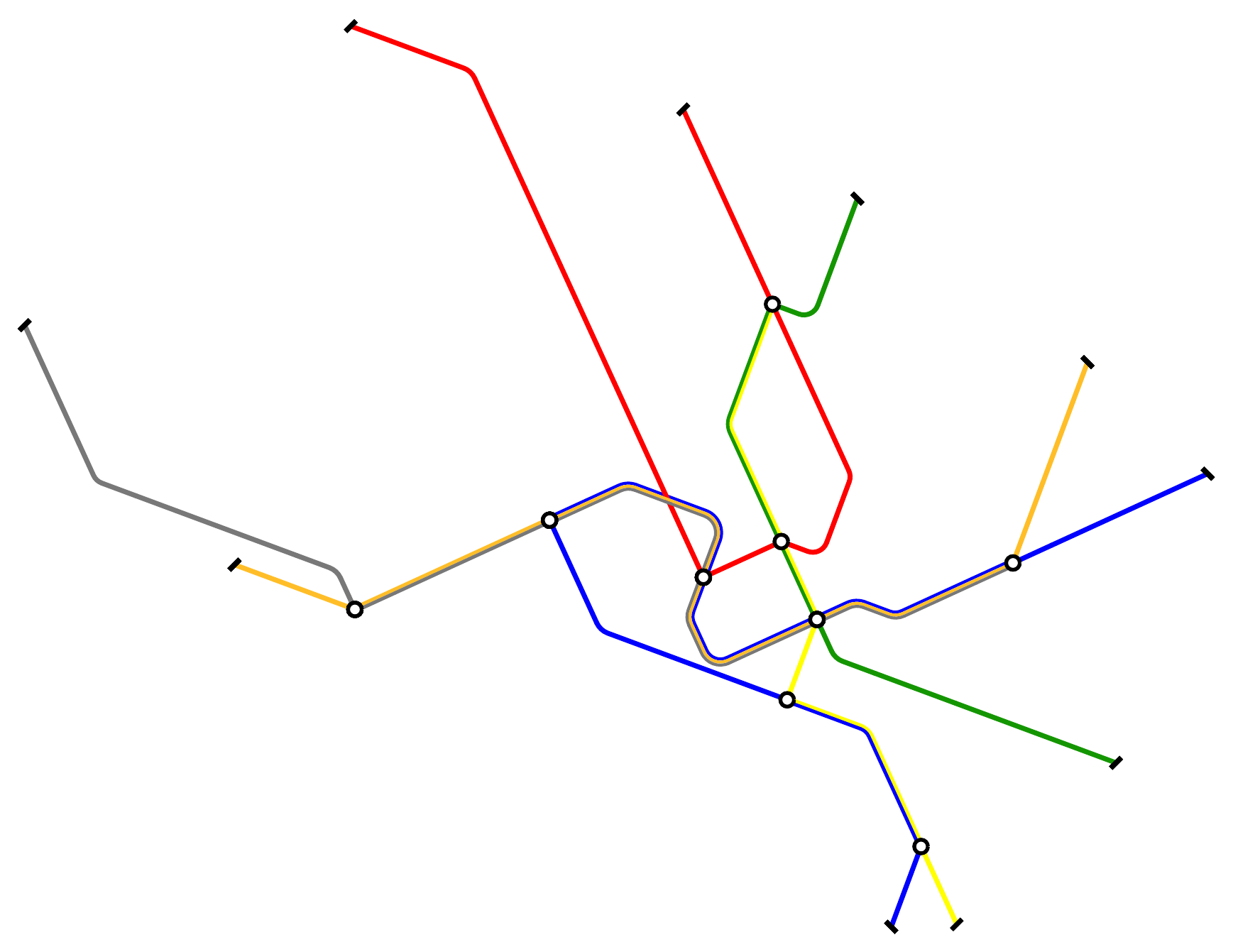} & \includegraphics[scale=.25]{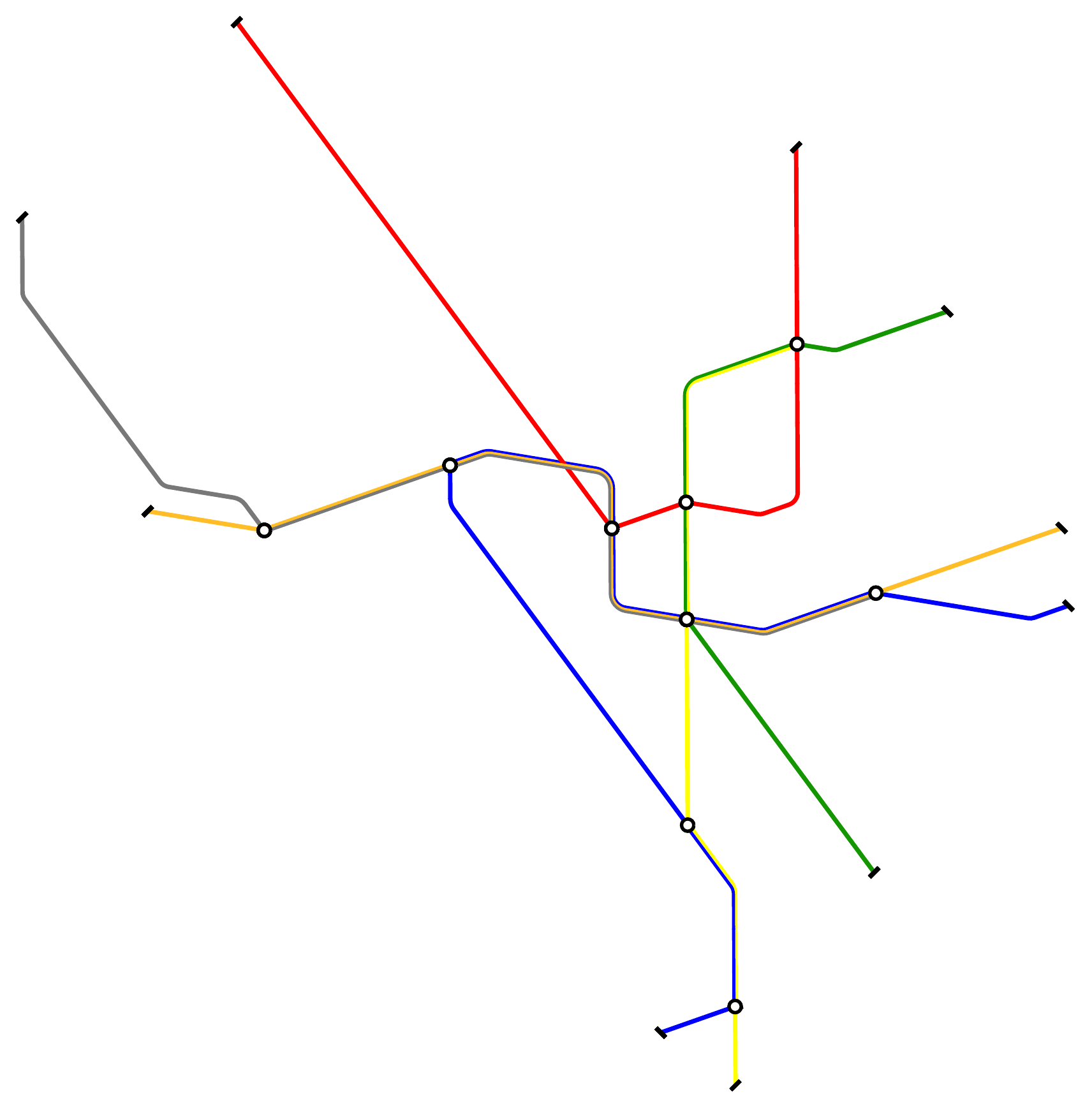} \\
		(d) 4-A & (e) 4-R & (f) 4-I \\[6pt]
		\includegraphics[scale=.25]{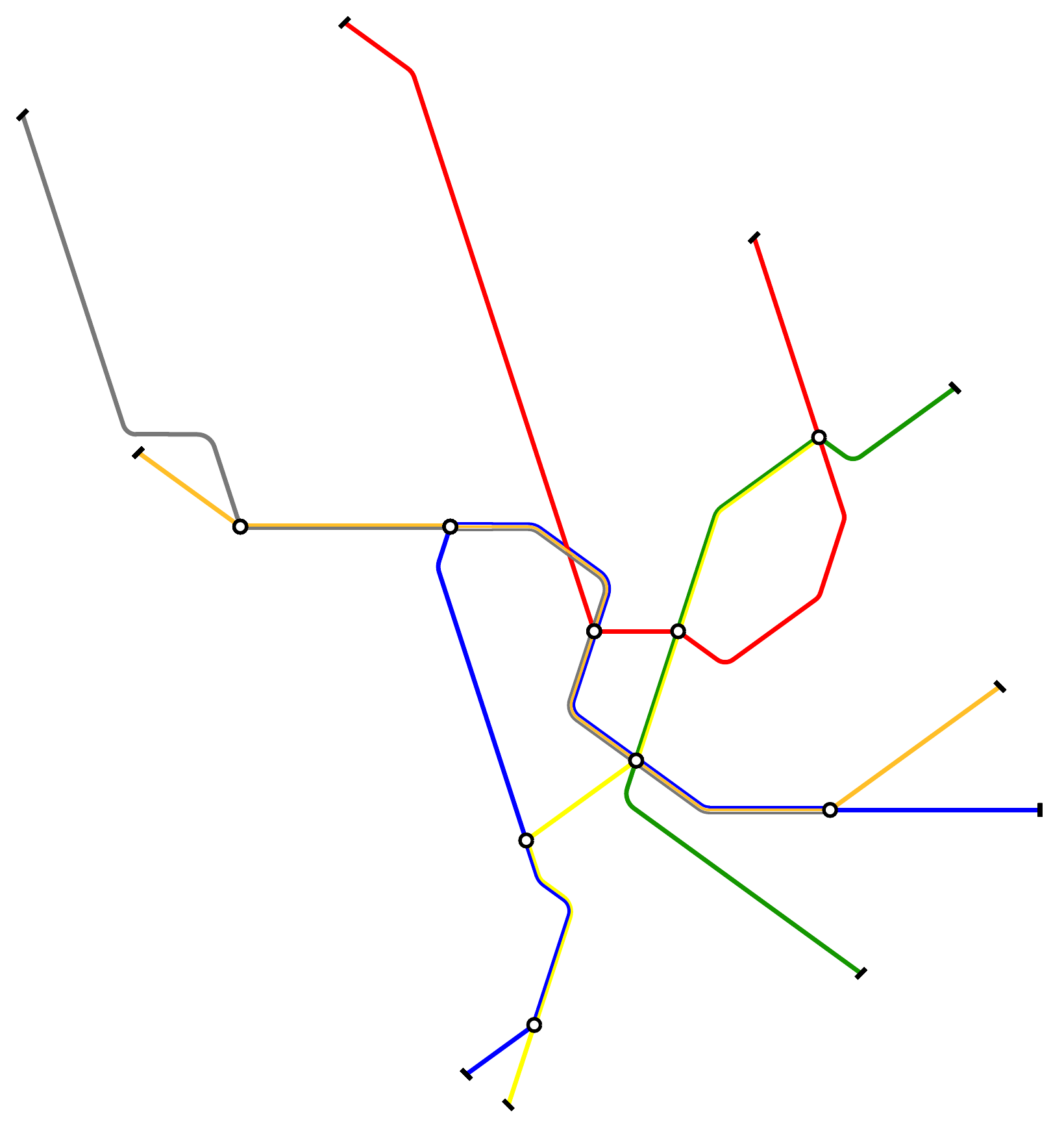} &   \includegraphics[scale=.25]{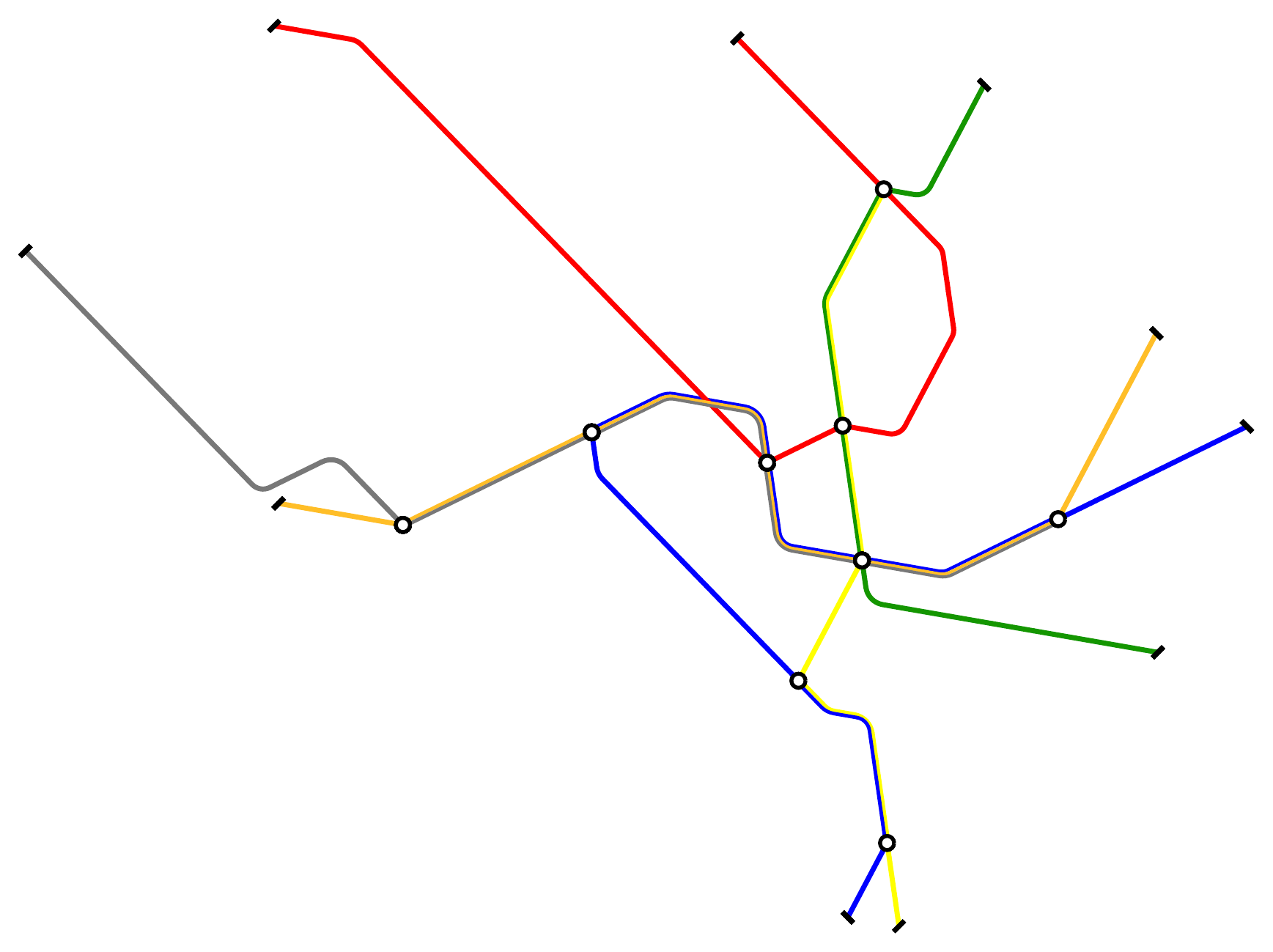} & \includegraphics[scale=.25]{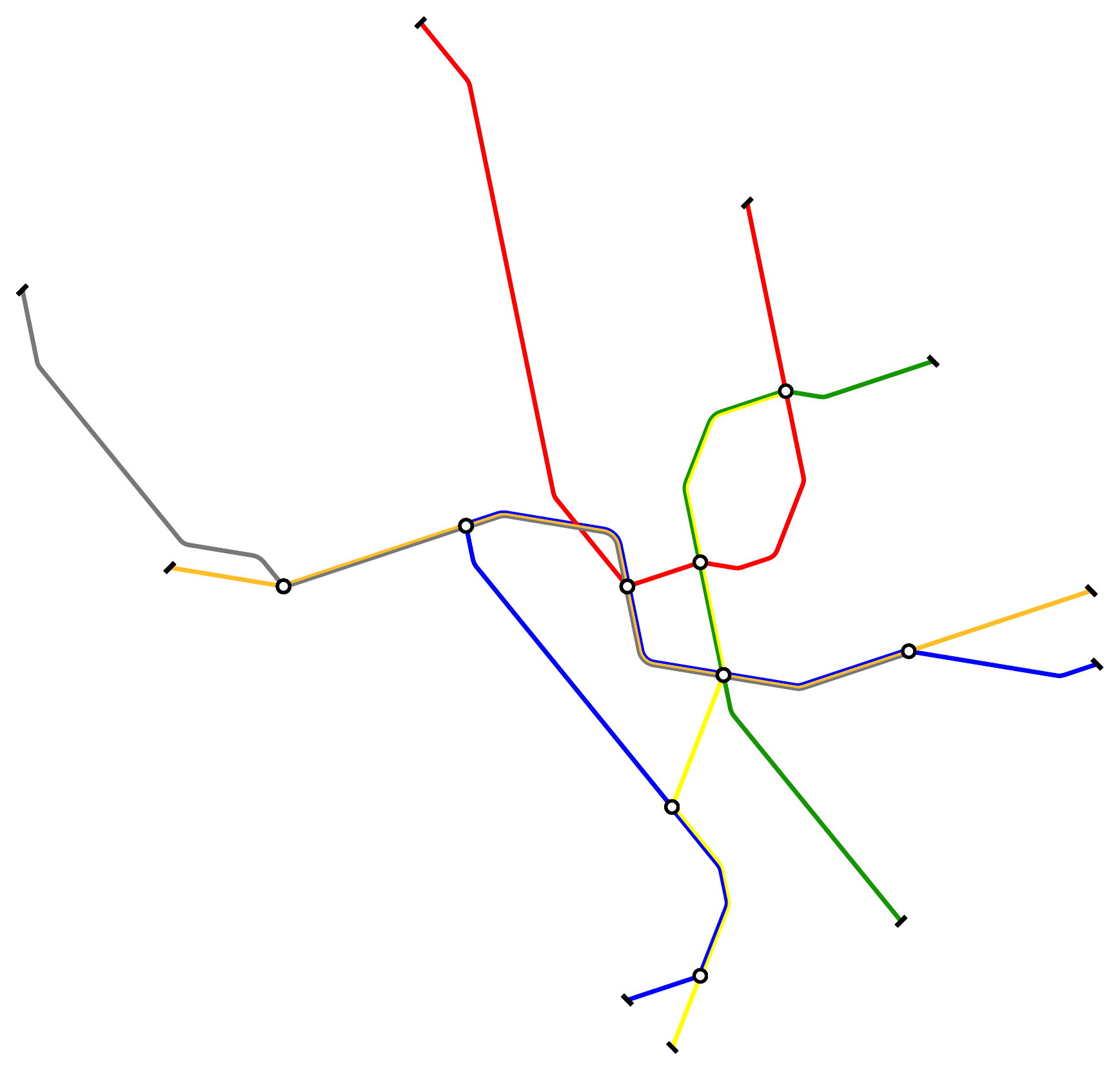} \\
		(g) -A5 & (h) 5-R & (i) 5-I \\[6pt]
	\end{tabular}
	\caption{Examples of Washington generated with objective function weights $(f_1, f_2, f_3) = (3, 2, 1)$. Rows are $k = 3, k=4, k=5$ from top to bottom, columns are aligned ($k$-A), regular ($k$-R) and irregular ($k$-I) orientation system from left to right.}\label{fig:ap_washington321}
\end{figure}
\begin{figure}[b!]
\centering
	\begin{tabular}{ccc}
		&\fbox{\includegraphics[scale=.25]{pictures/metros/input/washington_input.pdf}}&\\
		\includegraphics[scale=.25]{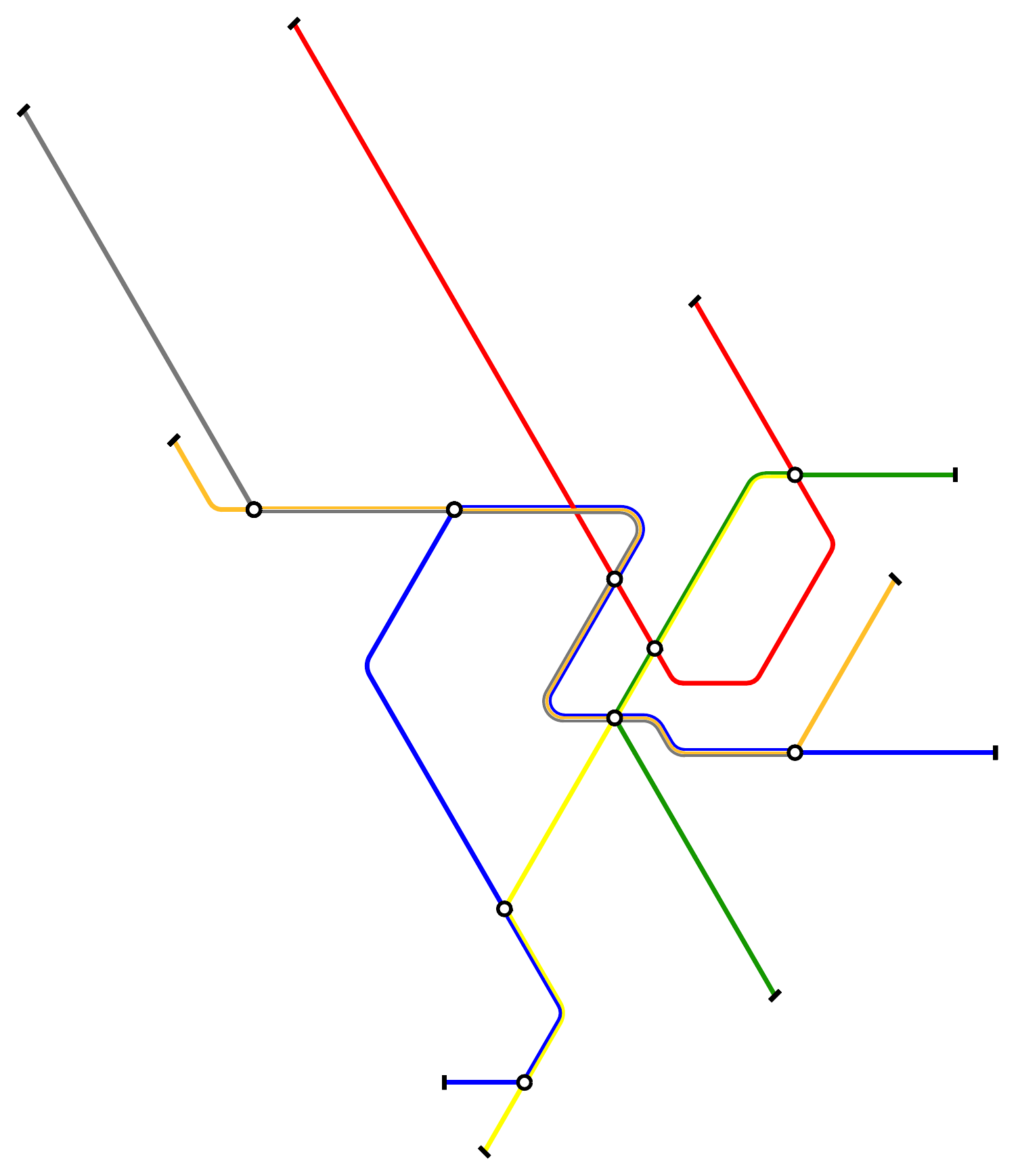} &   \includegraphics[scale=.25]{experiments/DIAGRAMS20-321/washington-3-CO.pdf} & \includegraphics[scale=.25]{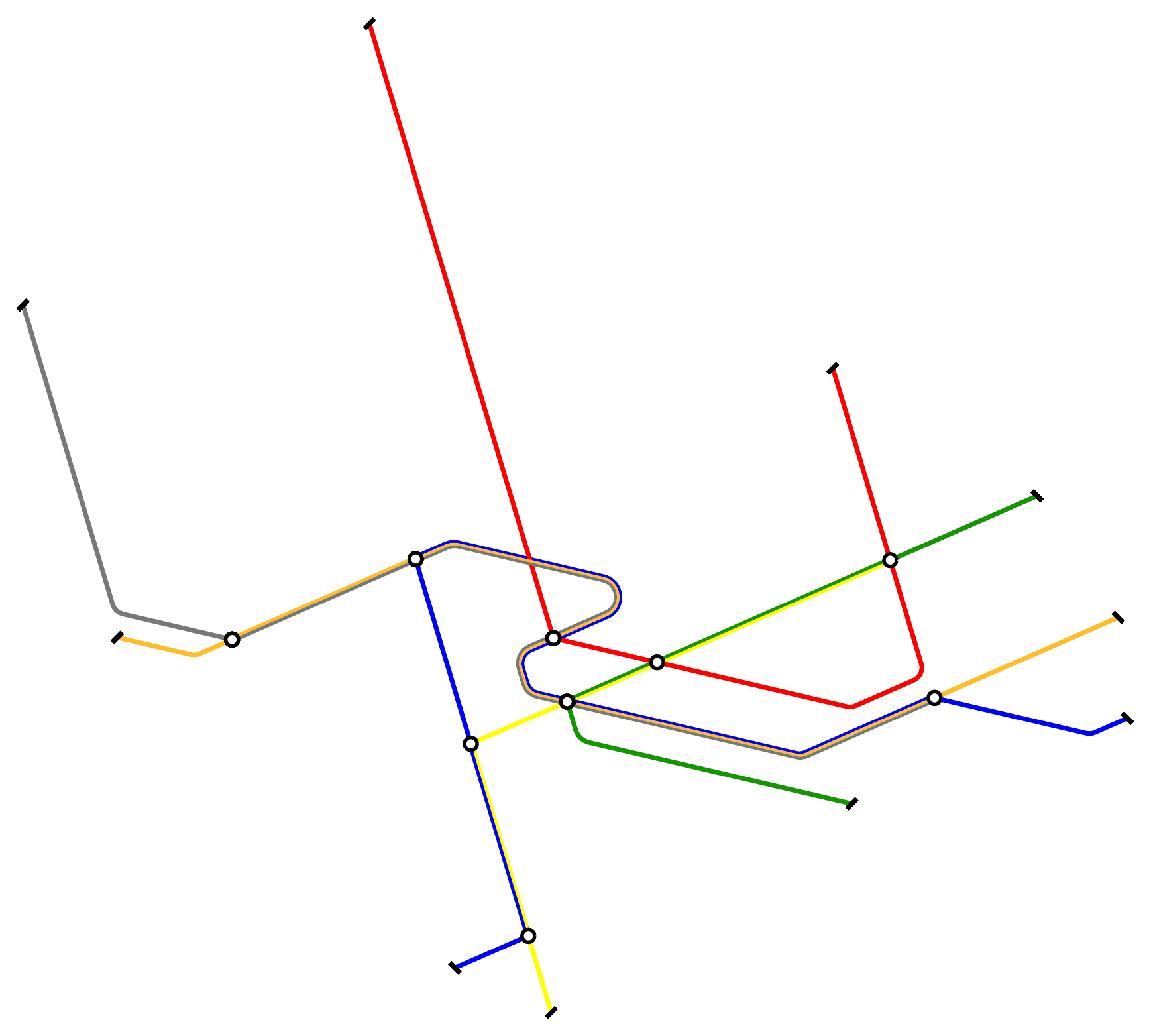} \\
		(a) 3-A & (b) 3-R & (c) 3-I \\[6pt]
		\includegraphics[scale=.25]{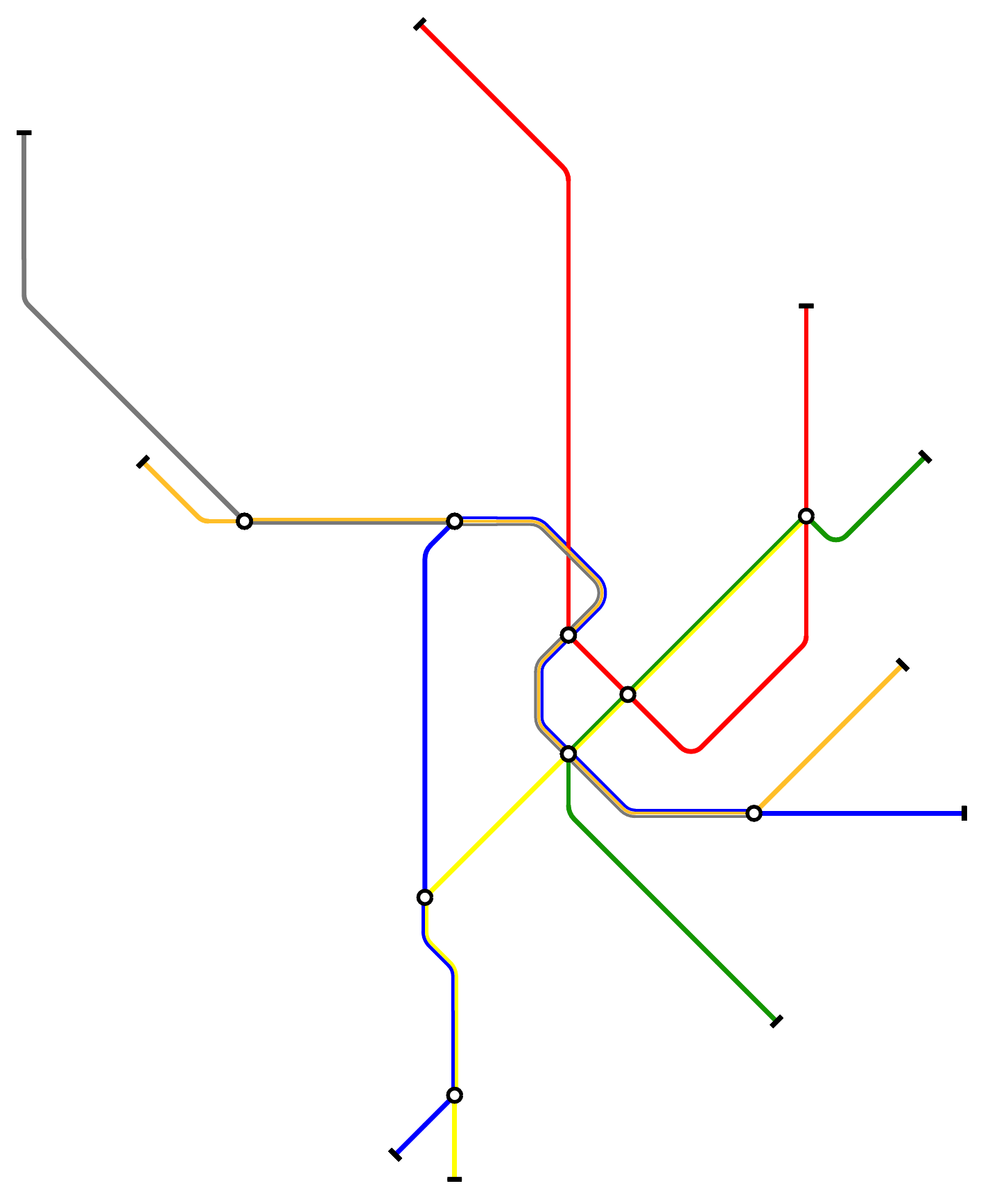} &   \includegraphics[scale=.25]{experiments/DIAGRAMS20-321/washington-4-CO.pdf} & \includegraphics[scale=.25]{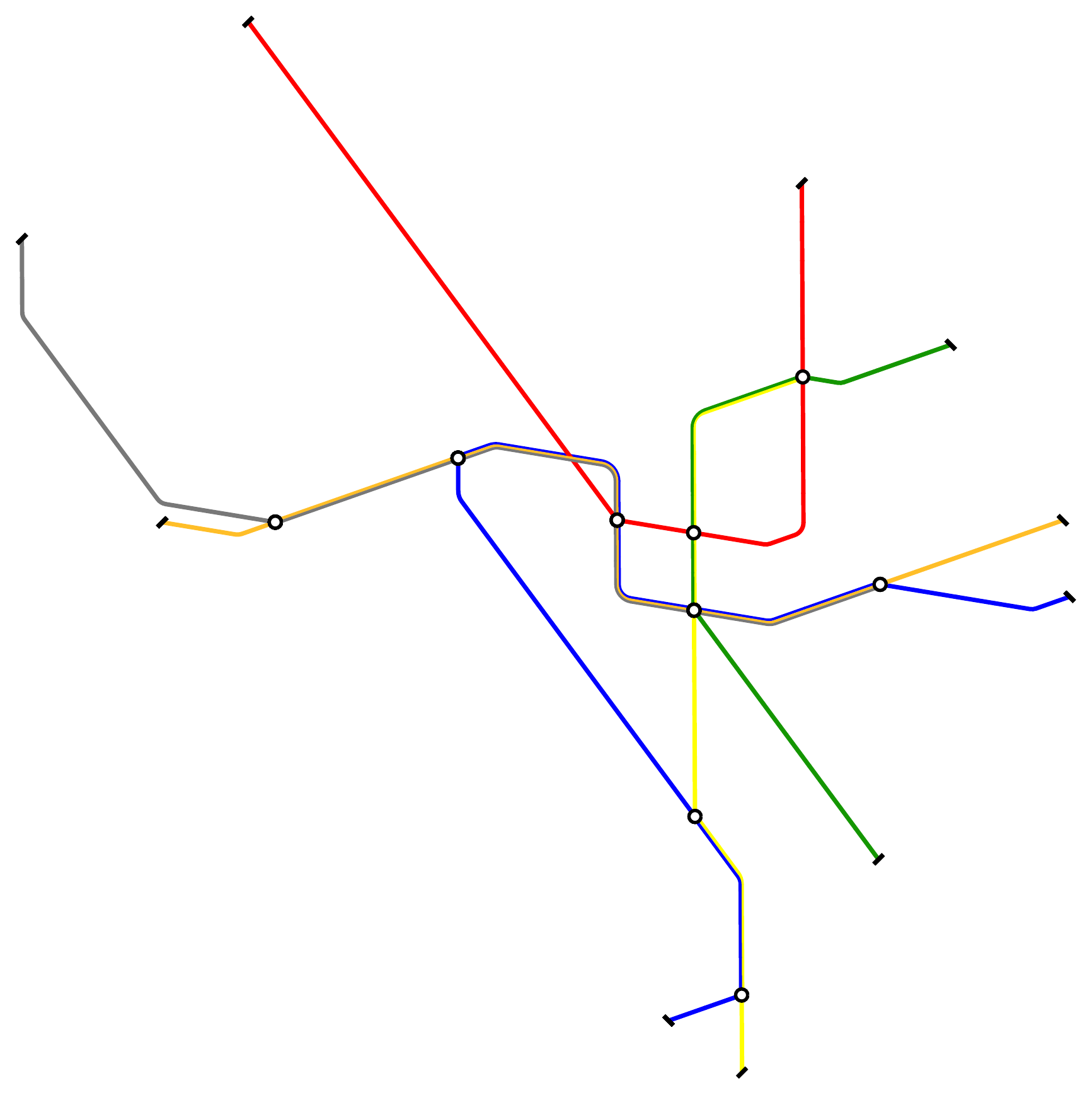} \\
		(d) 4-A & (e) 4-R & (f) 4-I \\[6pt]
		\includegraphics[scale=.25]{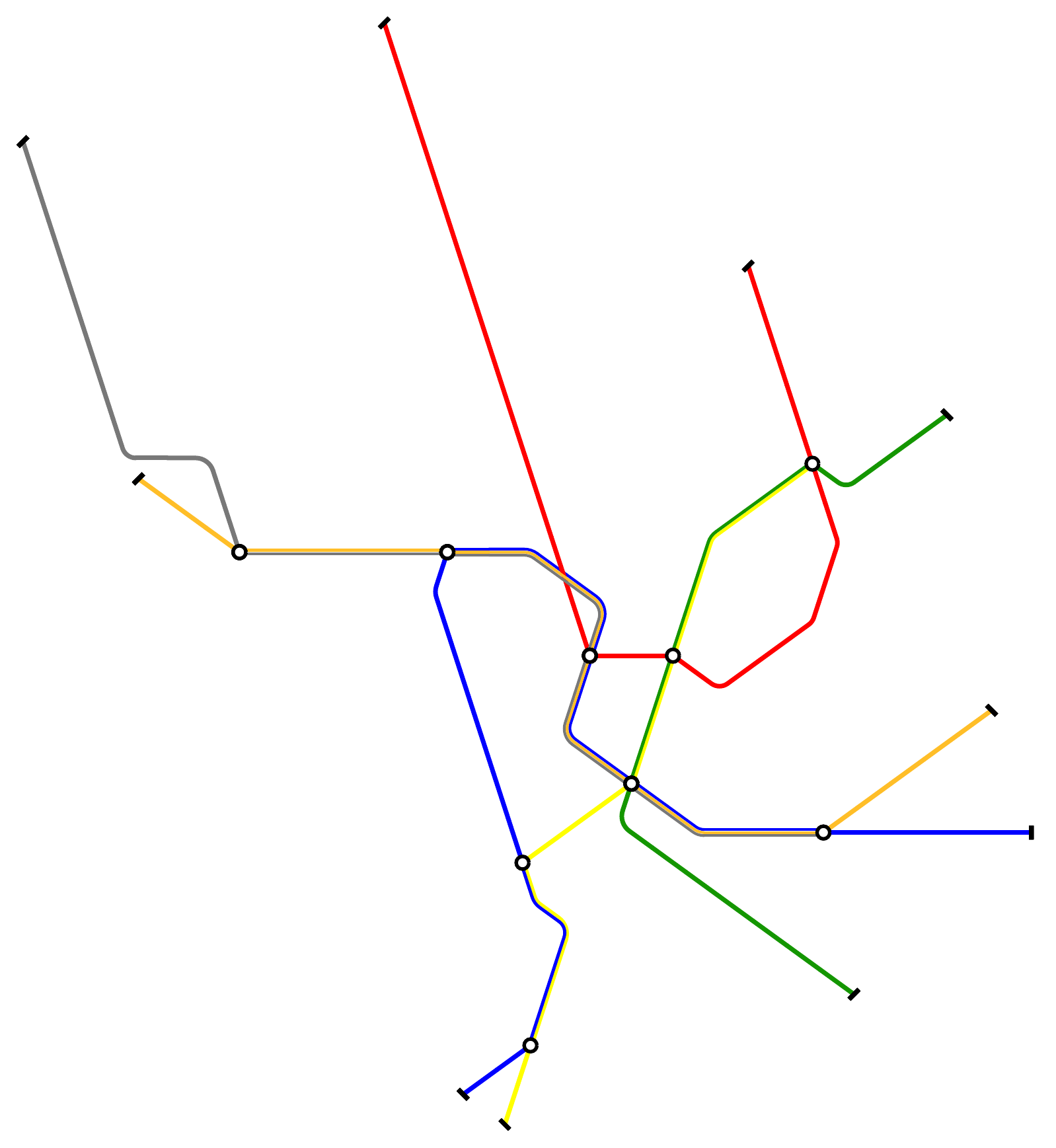} &   \includegraphics[scale=.25]{experiments/DIAGRAMS20-321/washington-5-CO.pdf} & \includegraphics[scale=.25]{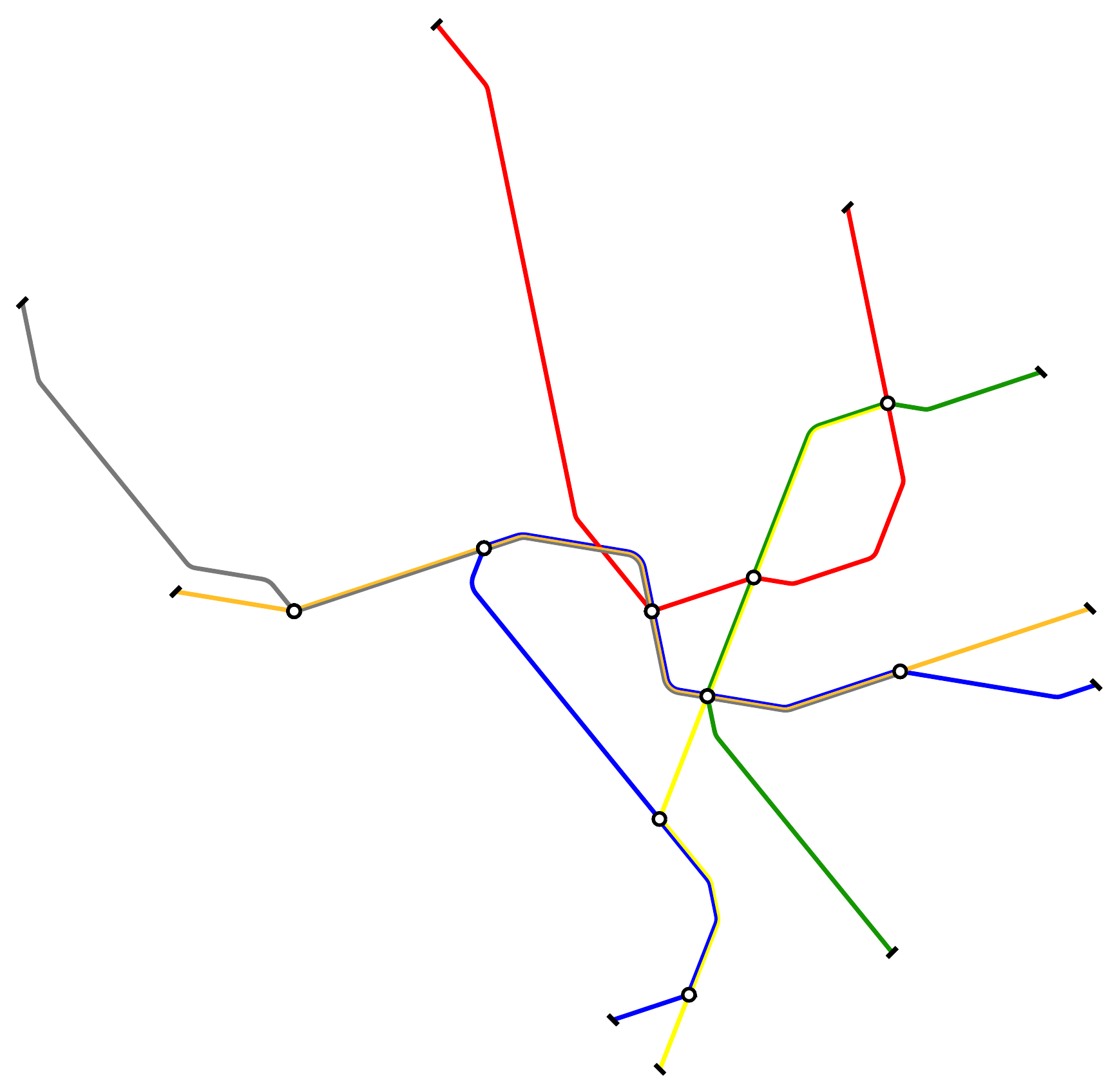} \\
		(g) 5-A & (h) 5-R & (i) 5-I \\[6pt]
	\end{tabular}
	\caption{Examples of Washington generated with objective function weights $(f_1, f_2, f_3) = (10, 5, 1)$. Rows are $k = 3, k=4, k=5$ from top to bottom, columns are aligned ($k$-A), regular ($k$-R) and irregular ($k$-I) orientation system from left to right.}\label{fig:ap_washington1051}
\end{figure}

\begin{figure}[b!]
\centering
	\begin{tabular}{ccc}
		&\fbox{\includegraphics[scale=.25]{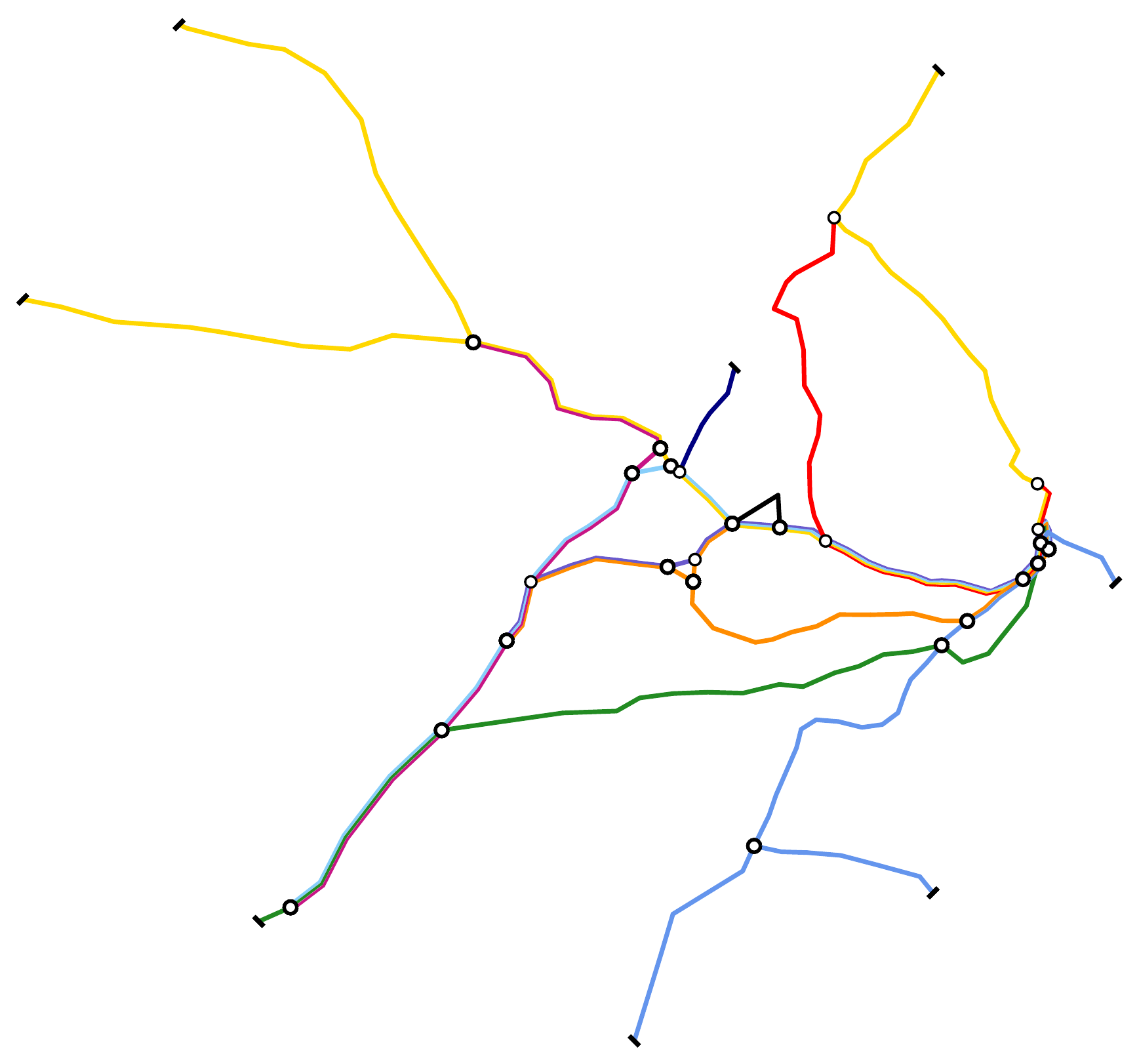}}&\\
		\includegraphics[scale=.25]{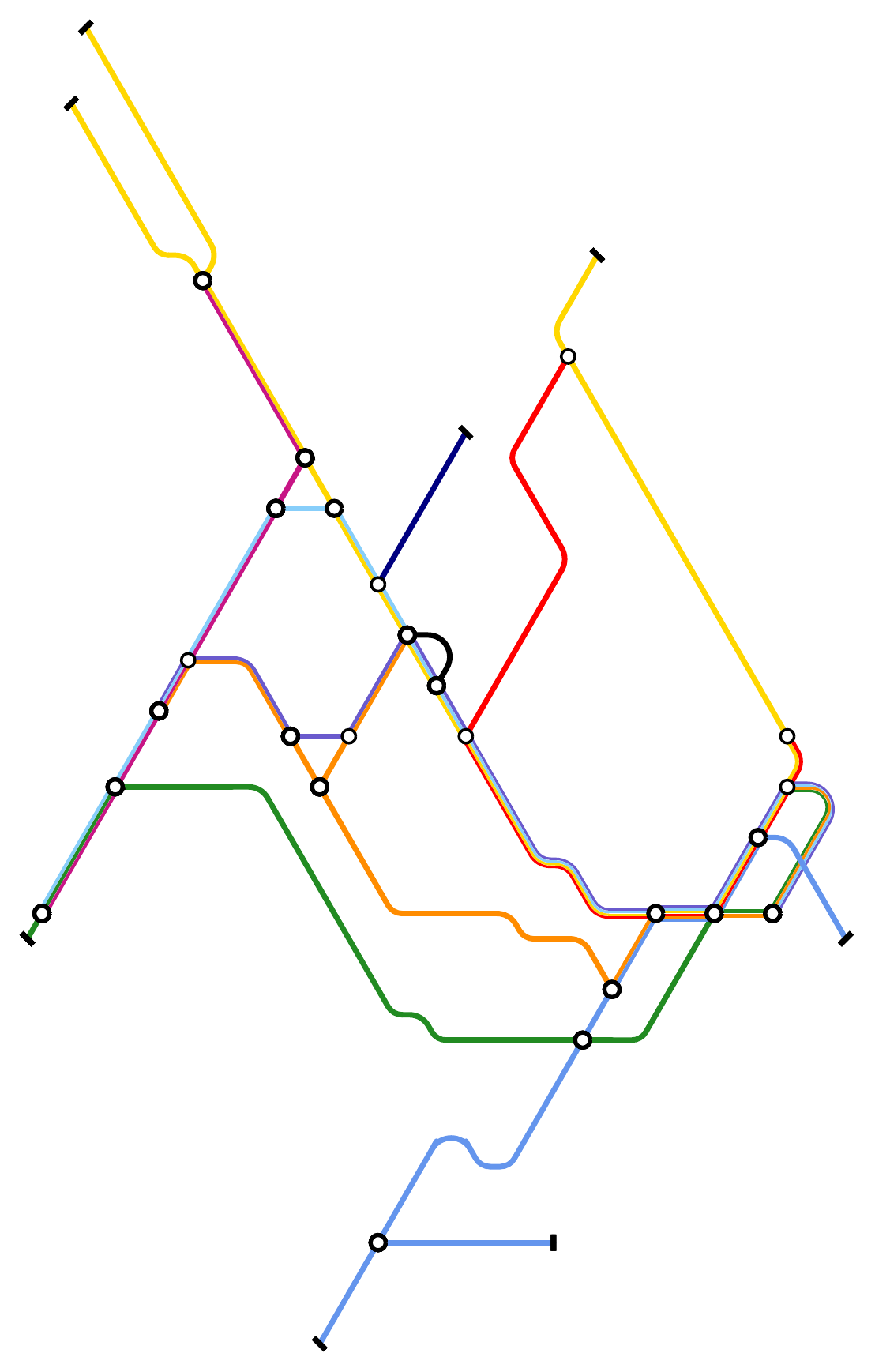} &   \includegraphics[scale=.25]{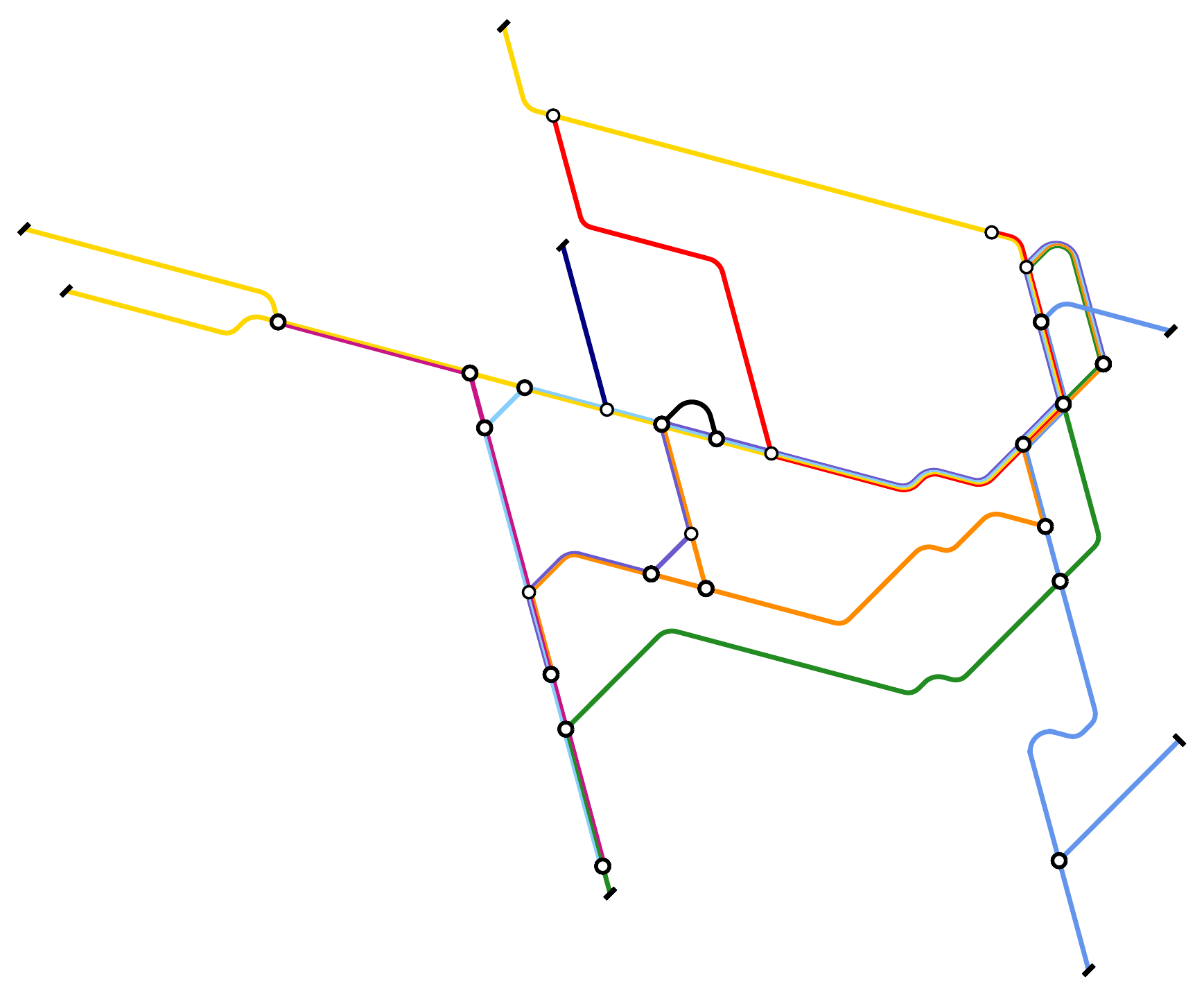} & \includegraphics[scale=.25]{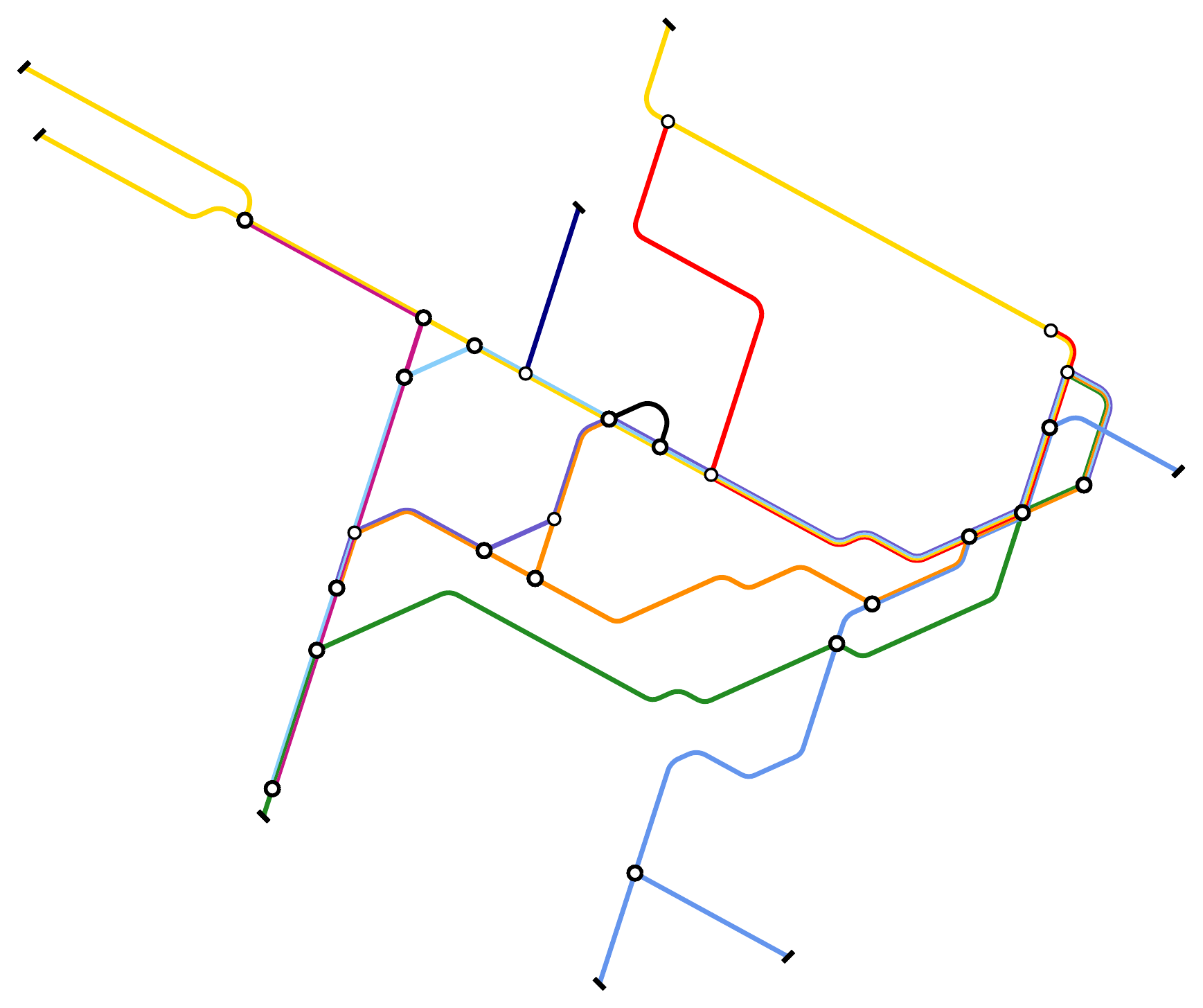} \\
		(a) 3-A & (b) 3-R & (c) 3-I \\[6pt]
		\includegraphics[scale=.25]{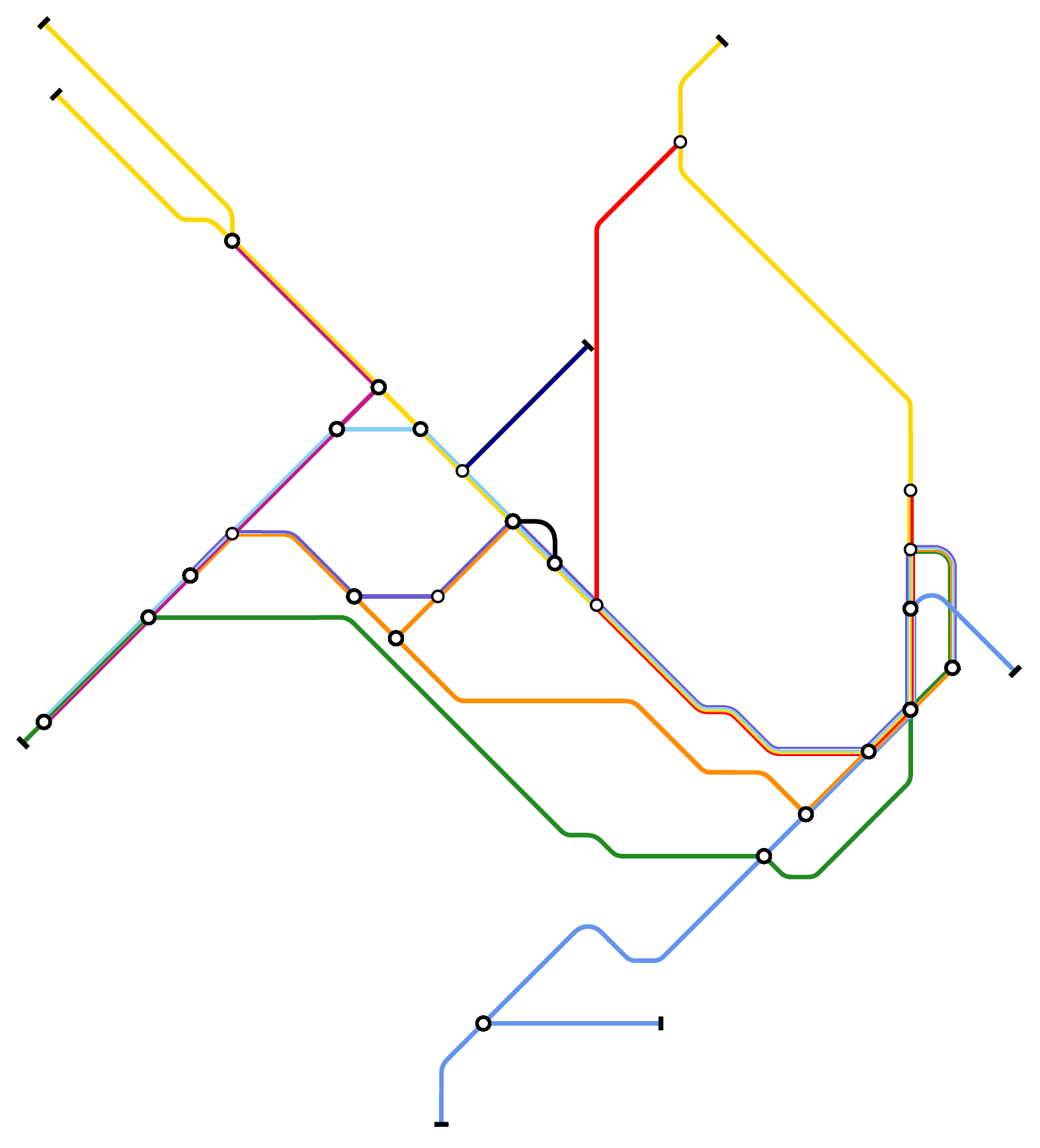} &   \includegraphics[scale=.25]{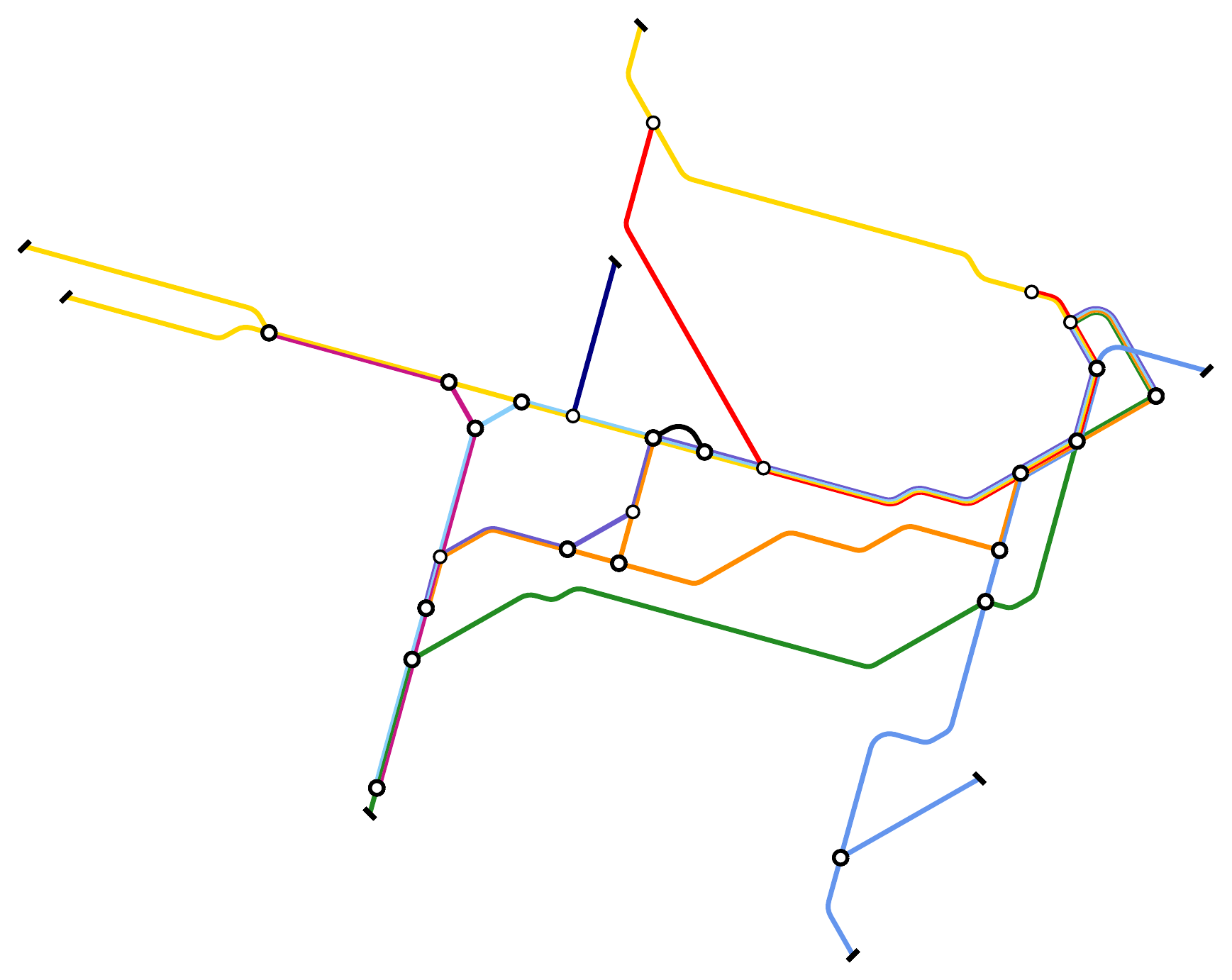} & \includegraphics[scale=.25]{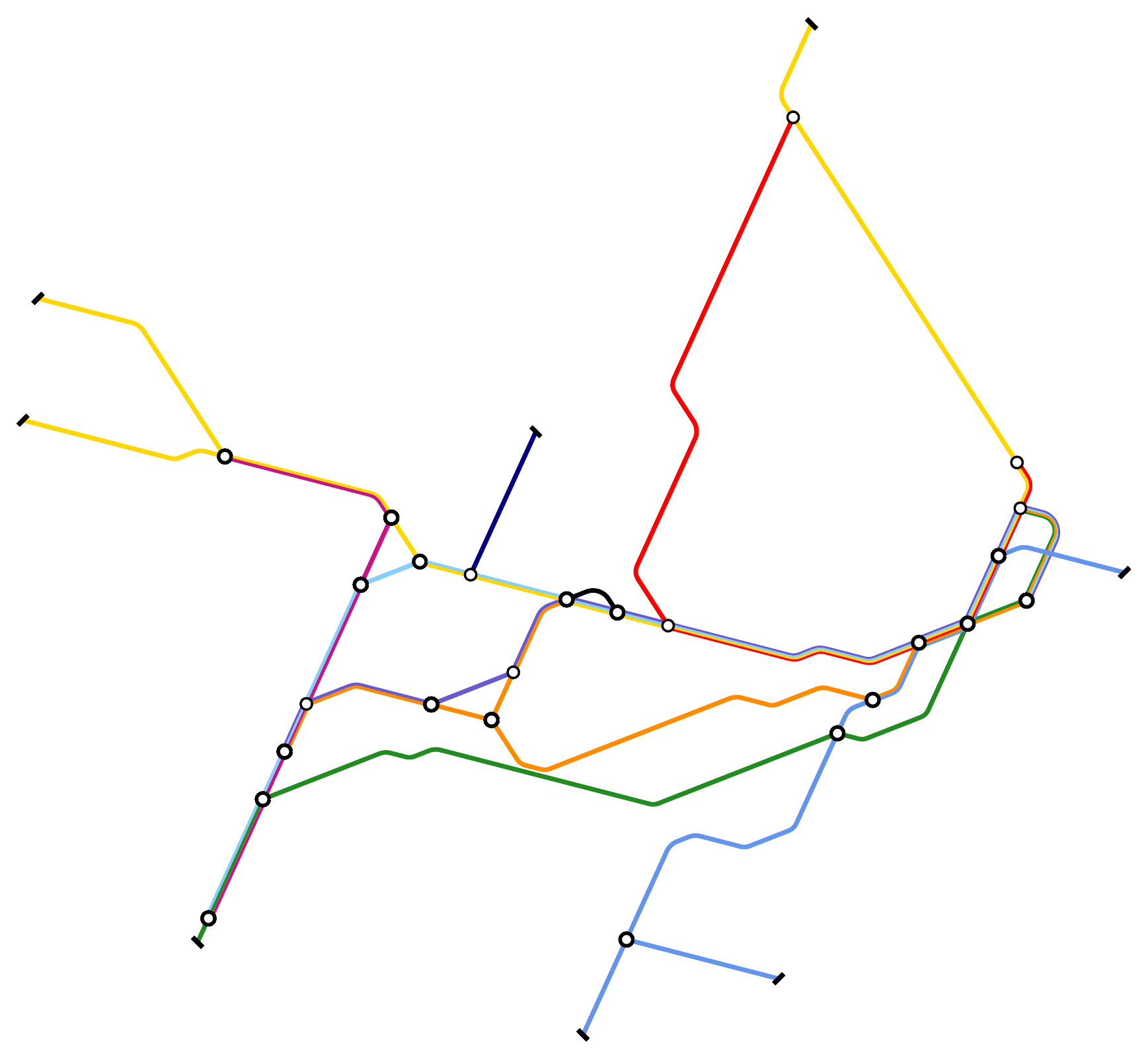} \\
		(d) 4-A & (e) 4-R & (f) 4-I \\[6pt]
		\includegraphics[scale=.25]{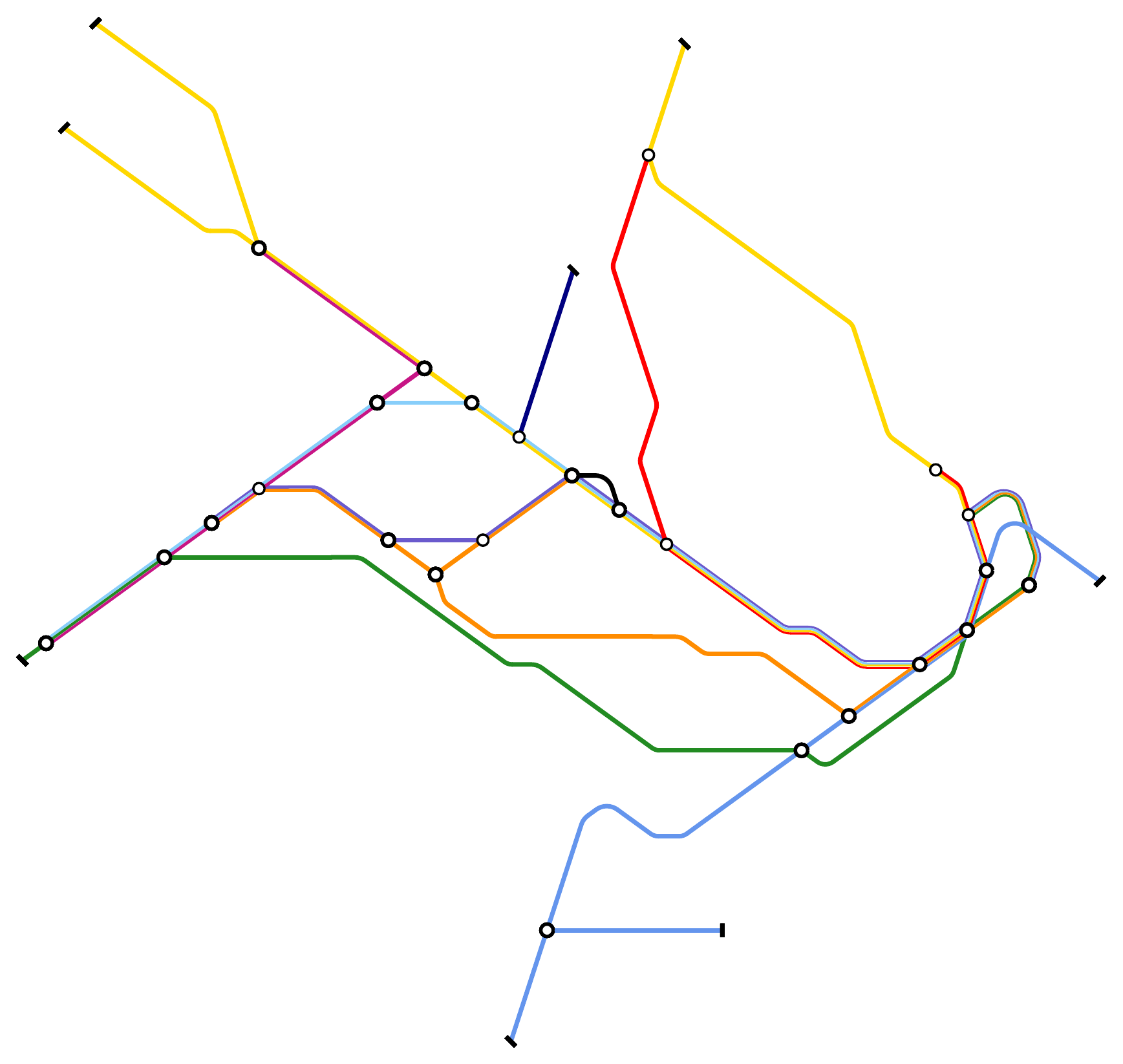} &   \includegraphics[scale=.25]{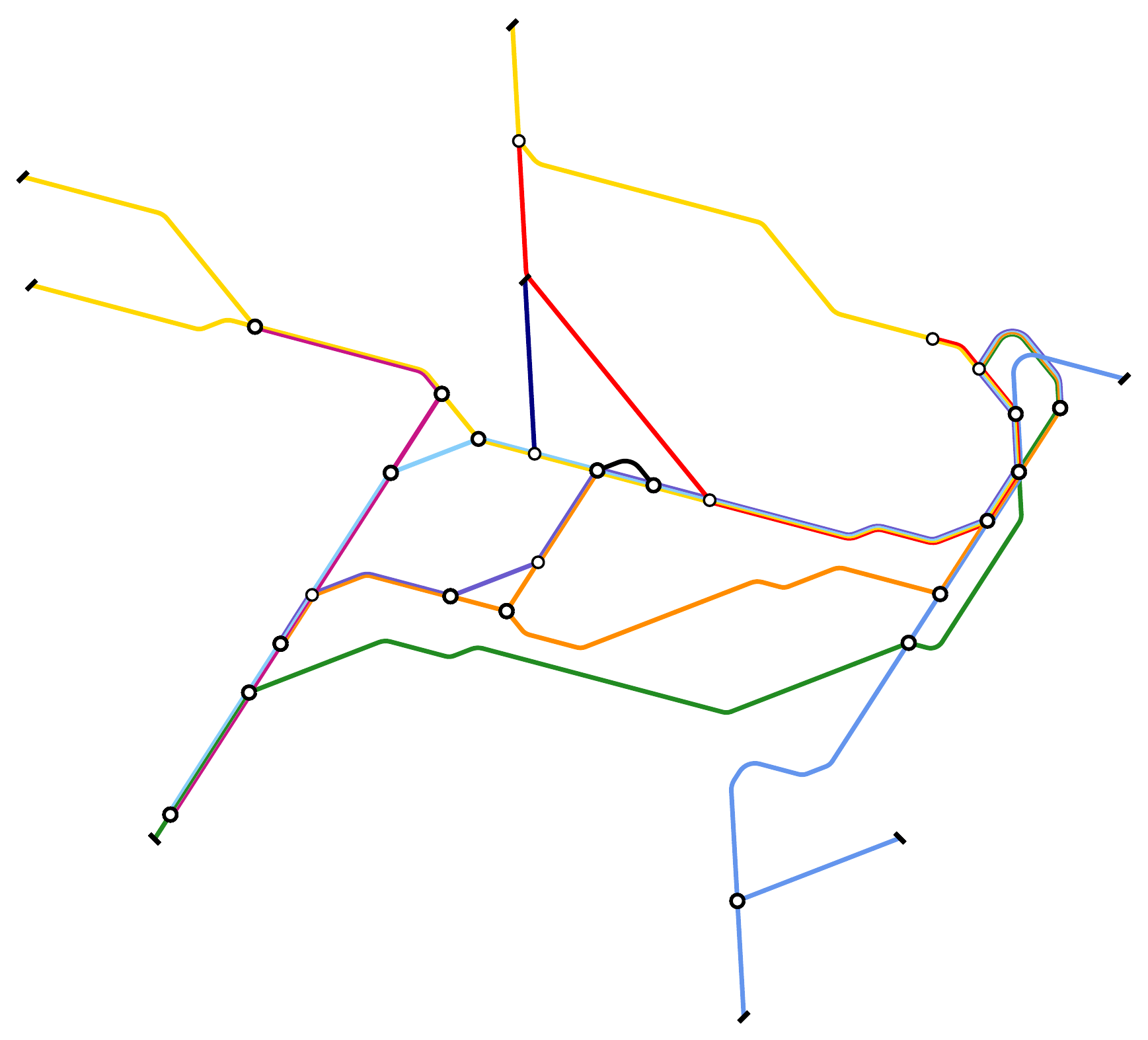} & \includegraphics[scale=.25]{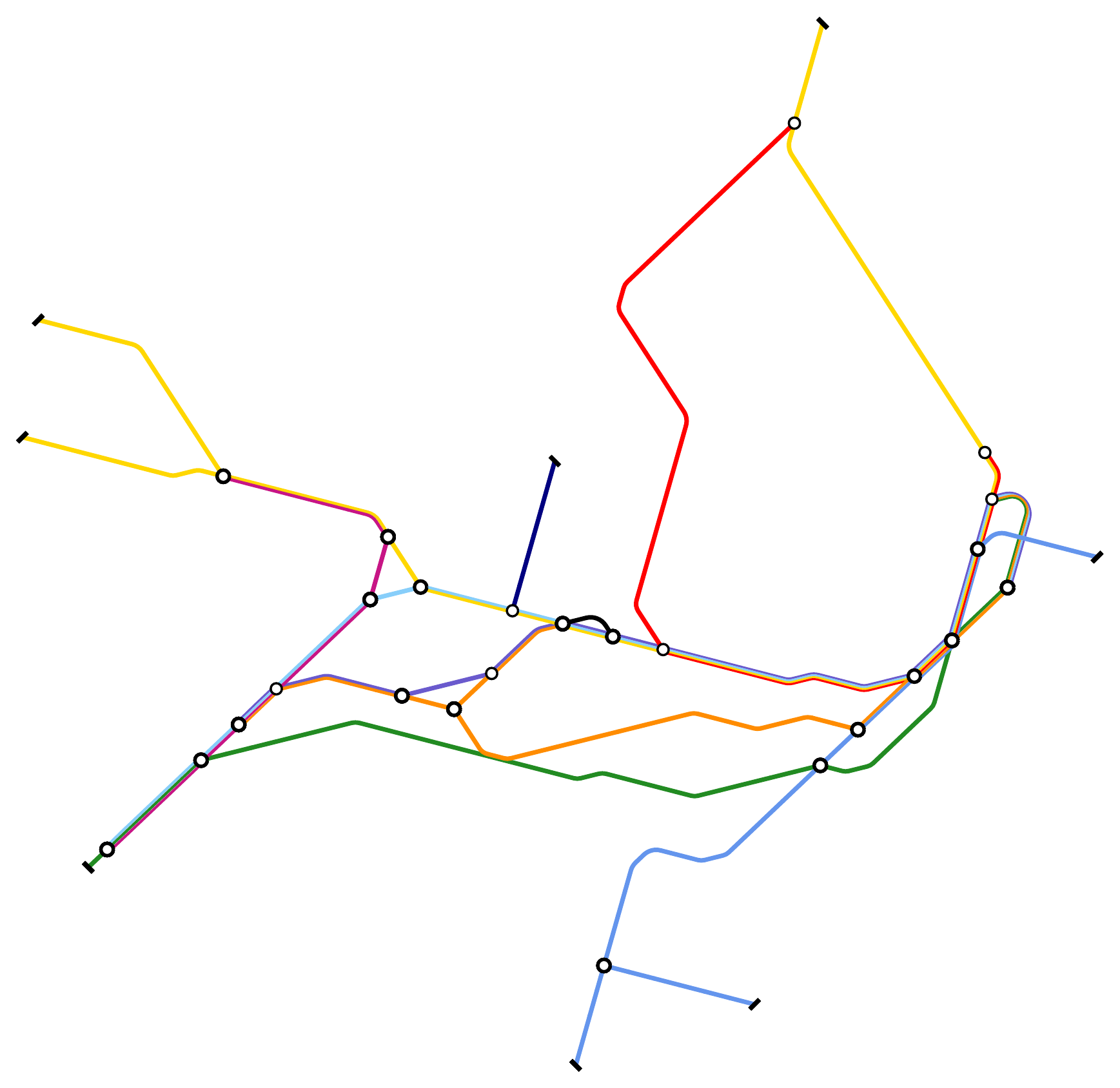} \\
		(g) 5-A & (h) 5-R & (i) 5-I \\[6pt]
	\end{tabular}
	\caption{Examples of Sydney generated with objective function weights $(f_1, f_2, f_3) = (3, 2, 1)$. Rows are $k = 3, k=4, k=5$ from top to bottom, columns are aligned ($k$-A), regular ($k$-R) and irregular ($k$-I) orientation system from left to right.}\label{fig:ap_sydney321}
\end{figure}
\begin{figure}[b!]
\centering
	\begin{tabular}{ccc}
		&\fbox{\includegraphics[scale=.25]{pictures/metros/input/sydney_input.pdf}}&\\
		\includegraphics[scale=.25]{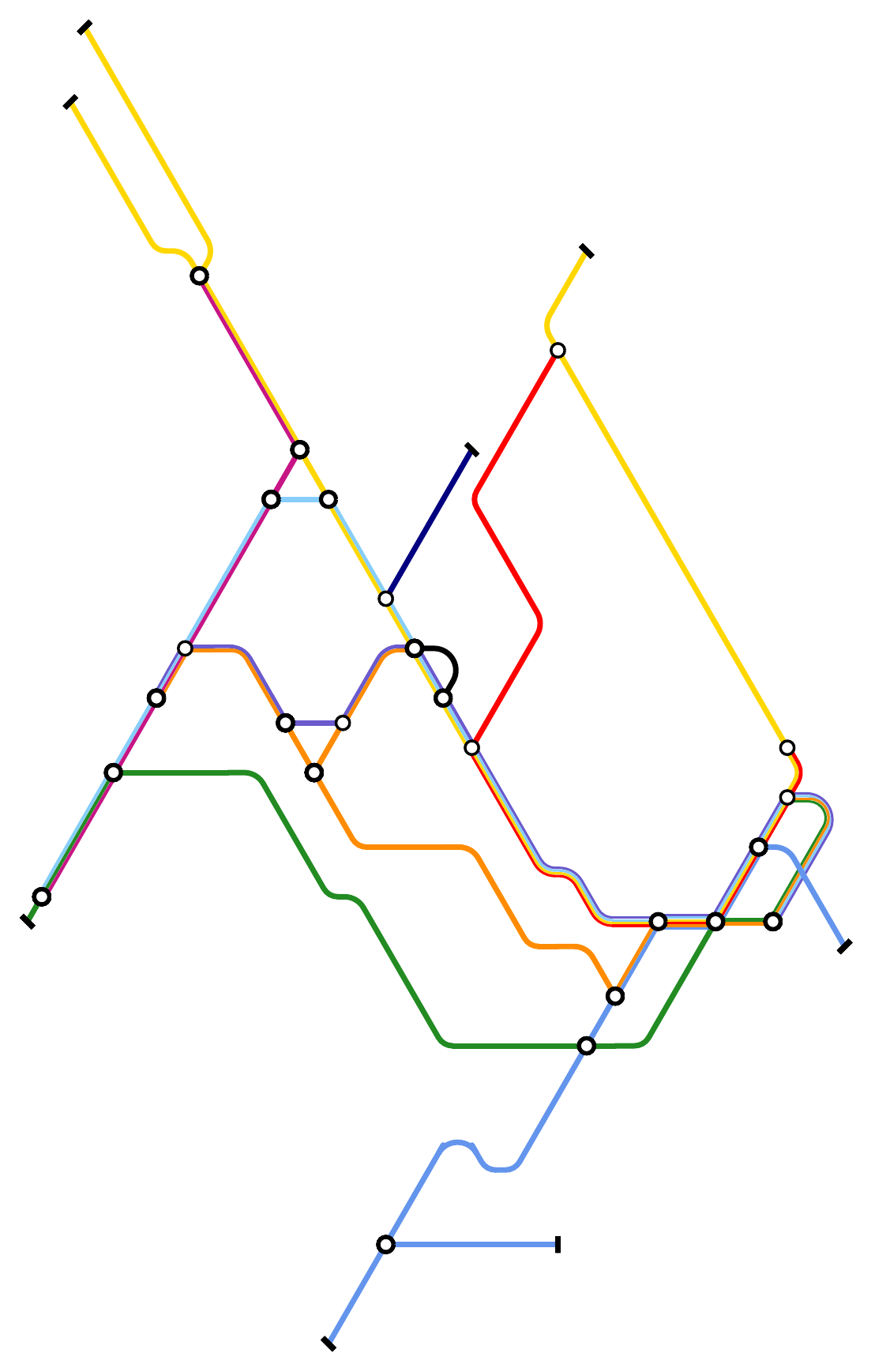} &   \includegraphics[scale=.25]{experiments/DIAGRAMS20-321/sydney-3-CO.pdf} & \includegraphics[scale=.25]{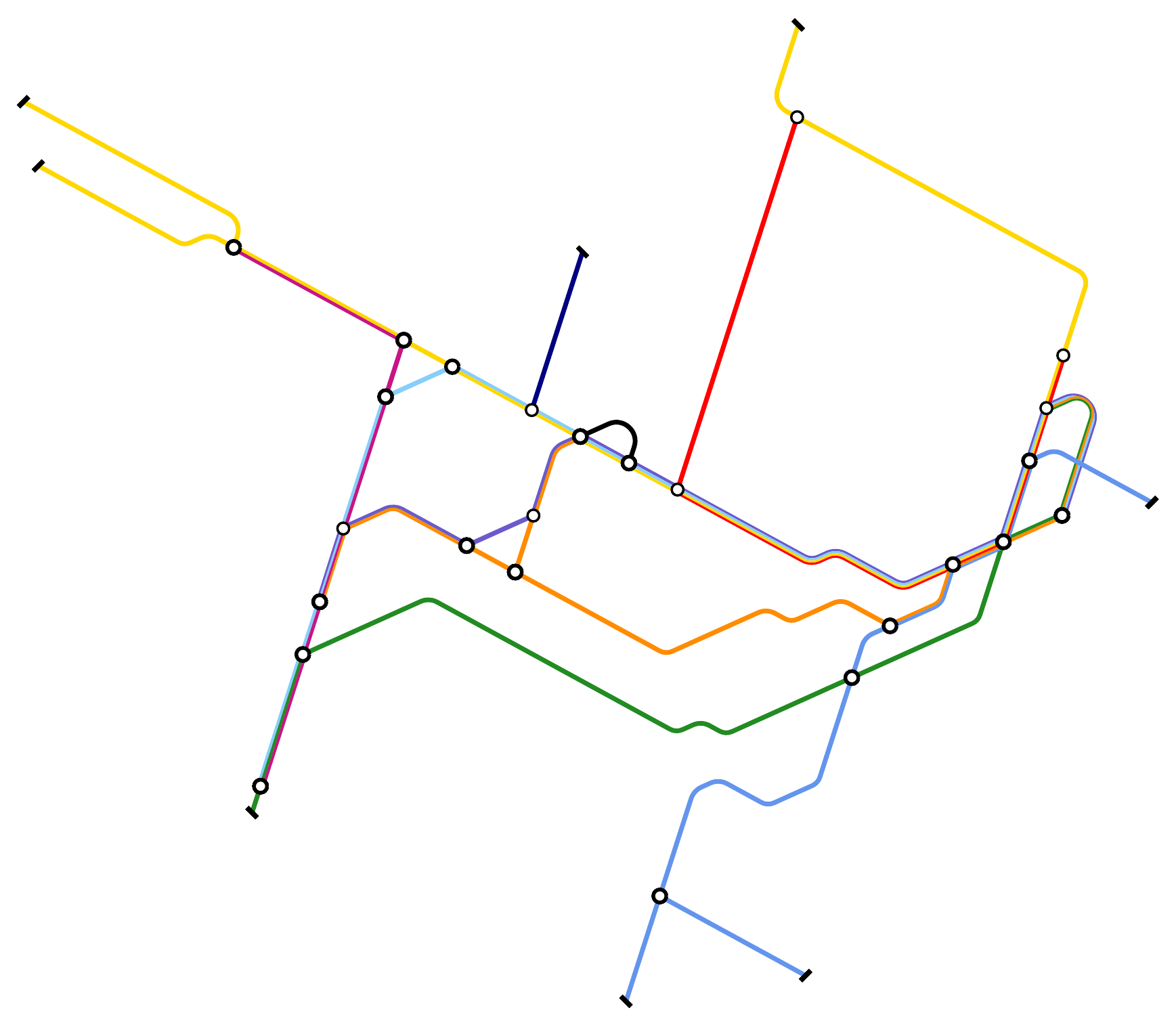} \\
		(a) 3-A & (b) 3-R & (c) 3-I \\[6pt]
		\includegraphics[scale=.25]{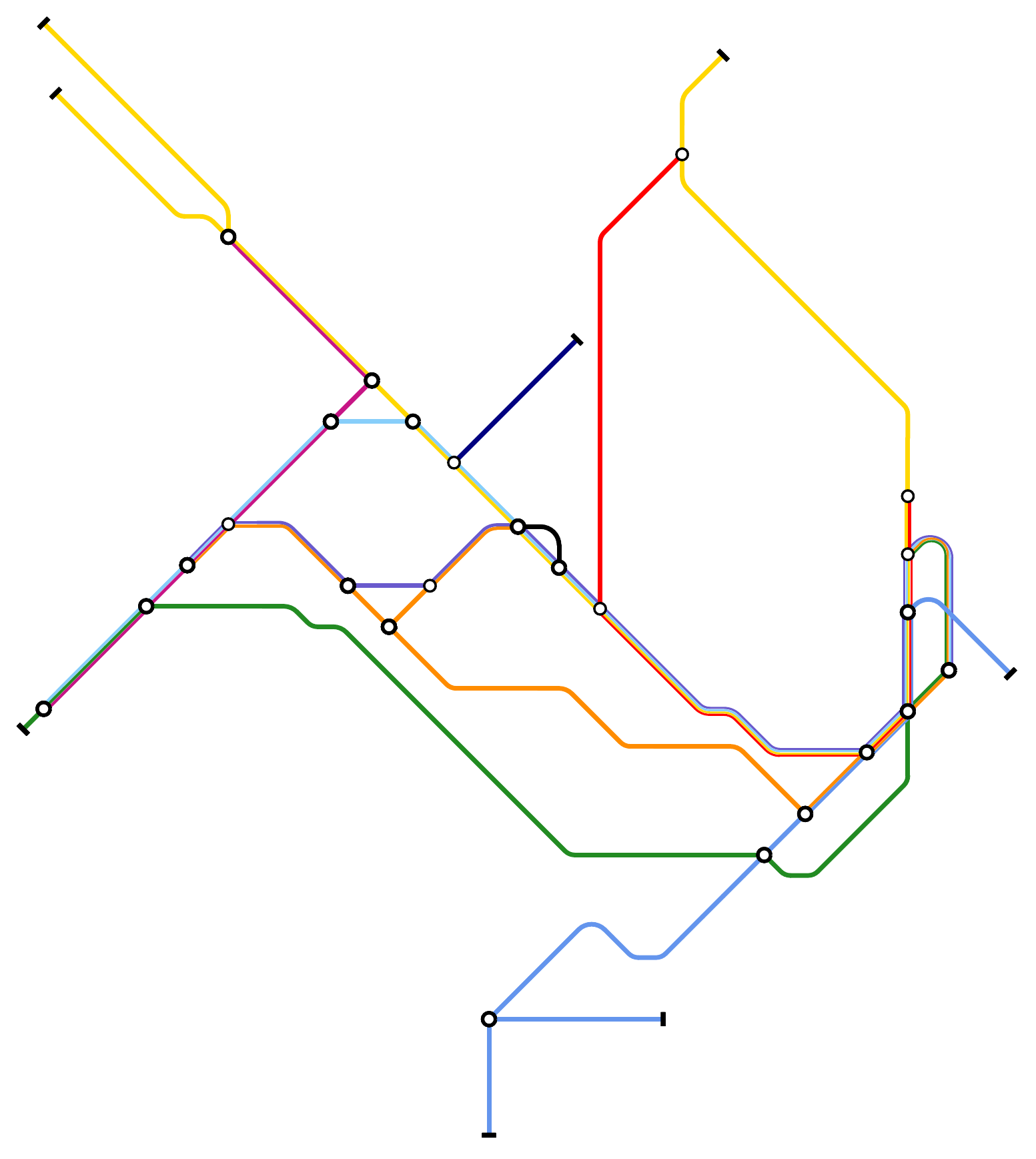} &   \includegraphics[scale=.25]{experiments/DIAGRAMS20-321/sydney-4-CO.pdf} & \includegraphics[scale=.25]{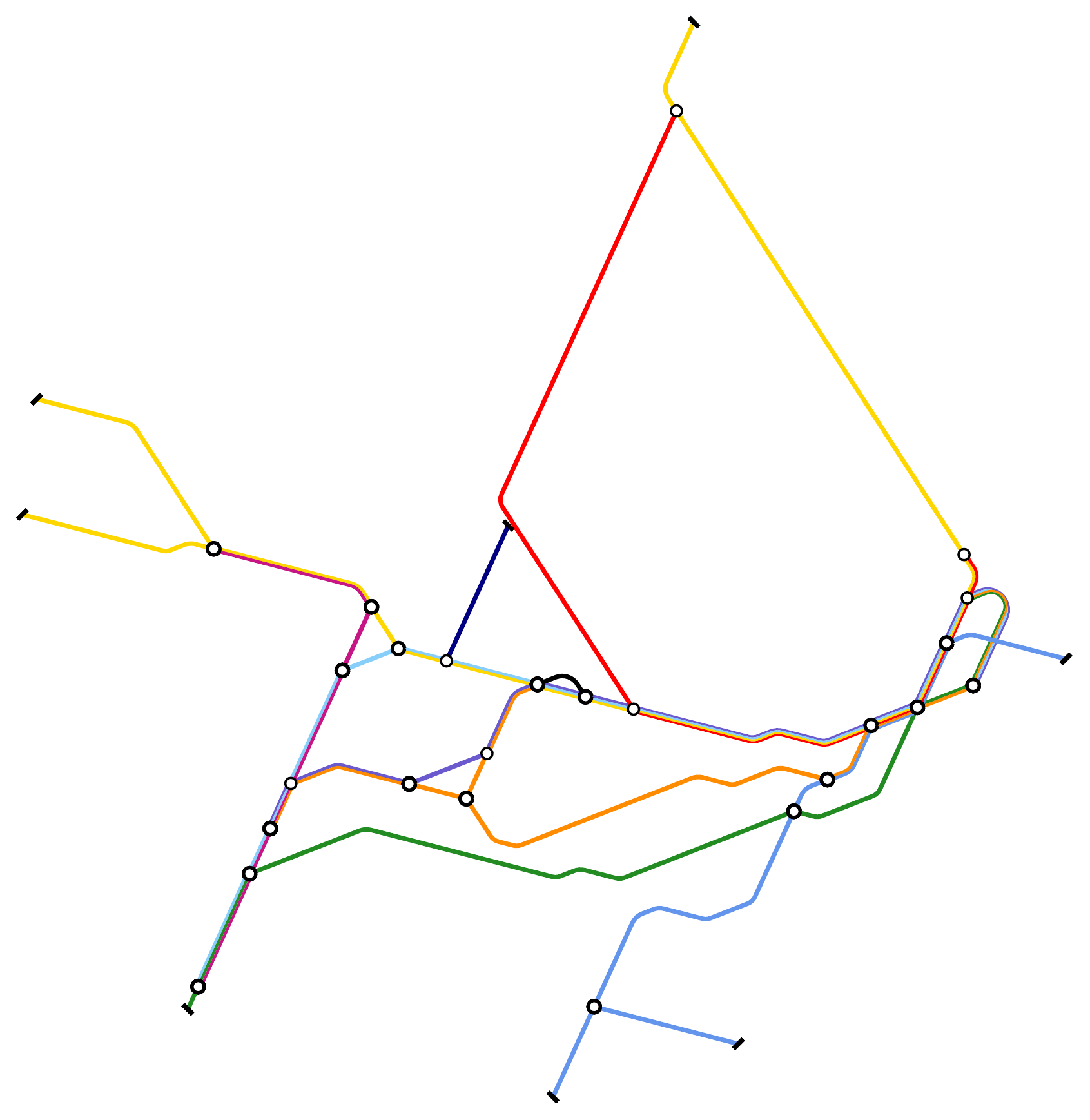} \\
		(d) 4-A & (e) 4-R & (f) 4-I \\[6pt]
		\includegraphics[scale=.25]{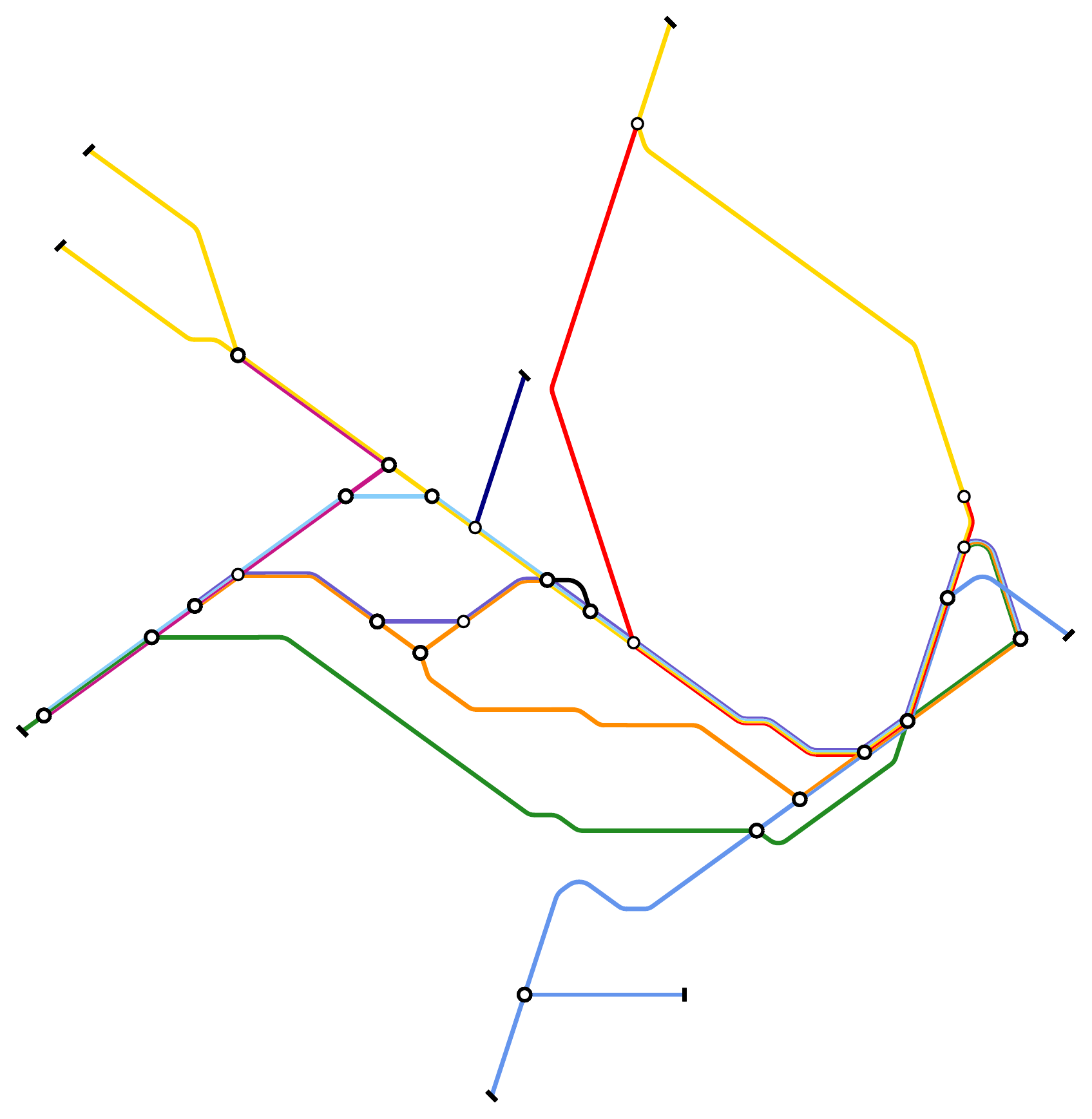} &   \includegraphics[scale=.25]{experiments/DIAGRAMS20-321/sydney-5-CO.pdf} & \includegraphics[scale=.25]{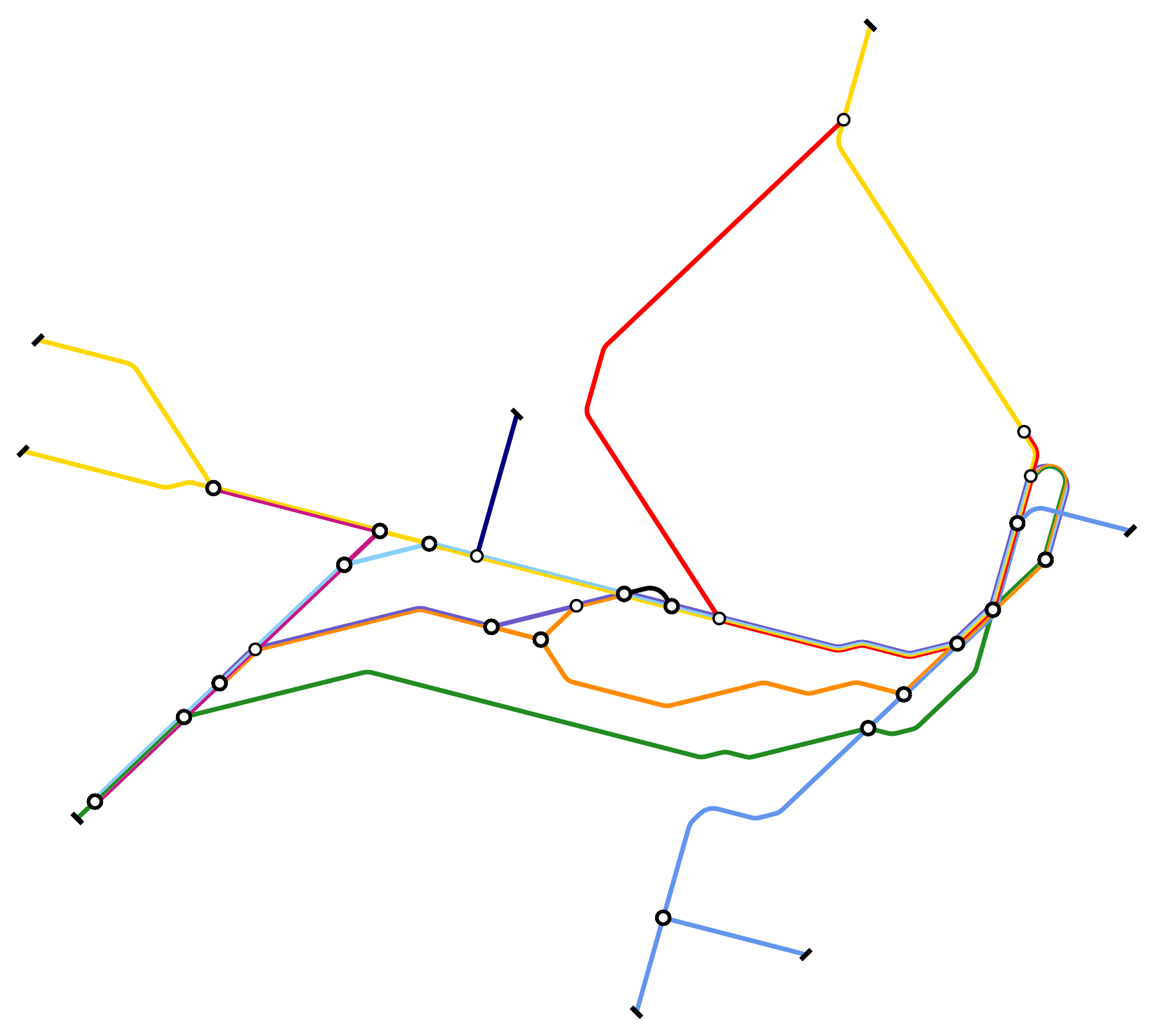} \\
		(g) 5-A & (h) 5-R & (i) 5-I \\[6pt]
	\end{tabular}
	\caption{Examples of Sydney generated with objective function weights $(f_1, f_2, f_3) = (10, 5, 1)$. Rows are $k = 3, k=4, k=5$ from top to bottom, columns are aligned ($k$-A), regular ($k$-R) and irregular ($k$-I) orientation system from left to right.}\label{fig:ap_sydney1051}
\end{figure}

\end{document}